\documentclass[a4paper,DIV=14]{scrartcl}
\usepackage[utf8]{inputenc}

\newenvironment{keywords}%
{\begin{quotation} \noindent \textbf{Keywords:}}
{\end{quotation}}

\setkomafont{title}{\normalfont\bfseries}
\setkomafont{section}{\normalfont\bfseries\LARGE}
\setkomafont{subsection}{\normalfont\bfseries\Large}
\setkomafont{subparagraph}{\normalfont\bfseries}

\usepackage[backend=bibtex,style=alphabetic,citestyle=alphabetic,natbib=true,maxbibnames=99]{biblatex} 
\usepackage[autostyle=true]{csquotes} 

\usepackage{amsmath}
\usepackage{amsthm}
\usepackage{amssymb}
\usepackage{microtype}
\usepackage{stmaryrd}
\usepackage{complexity}
\usepackage{tikz}
\tikzstyle{every node}=[font=\scriptsize]

\def\centerarc[#1](#2)(#3:#4:#5)
{ \draw[line width=0.6pt,#1] ($(#2)+({#5*cos(#3)},{#5*sin(#3)})$) arc (#3:#4:#5); }

\def\carc(#1:#2:#3:#4:#5)
{ \centerarc[](#1,#2)(#3-(#4/2):#3+(#4/2):#5); }

\def\carcred(#1:#2:#3:#4:#5)
{ \centerarc[color=red](#1,#2)(#3-(#4/2):#3+(#4/2):#5); }

\def\carcblue(#1:#2:#3:#4:#5)
{ \centerarc[color=blue](#1,#2)(#3-(#4/2):#3+(#4/2):#5); }

\def\carcpurple(#1:#2:#3:#4:#5)
{ \centerarc[color=purple](#1,#2)(#3-(#4/2):#3+(#4/2):#5); }

\def\arclabel(#1:#2:#3:#4:#5:#6)
{ \node (#1) at ($({#3+#5*cos(#6)},{#4+#5*sin(#6)})$) {#2};   }

\def\line(#1:#2:#3)
{ \draw[line width=0.6pt] (#1,#2) -- (#1+#3,#2); }

\def\vlinetikz[#1](#2:#3:#4)
{ \draw[line width=0.6pt,#1] (#2,#3) -- (#2,#3+#4); }

\def\beam[#1](#2:#3:#4:#5:#6)
{ \draw[line width=0.6pt,#1] ({#5*cos(#6)+#2},{#5*sin(#6)+#3}) -- ({#4*cos(#6)+#2},{#4*sin(#6)+#3});  }

\def\line(#1:#2:#3:#4)
{ \fill [white] (#1,#2) rectangle (#1+#3,#2+#4); }

\newcommand\strictinclusion[1][1]{%
    \begin{tikzpicture}[scale=#1]
    \draw (0,0) -- (0.25,0.15);
    \end{tikzpicture}
}
\newcommand\strictinclusiond[1][1]{%
    \begin{tikzpicture}[scale=#1]
    \draw (0,0) -- (0.25,-0.15);
    \end{tikzpicture}
}
\usepackage{tikz-qtree}
\usetikzlibrary{calc}

\theoremstyle{plain}
\newtheorem{theorem}{Theorem}[section]
\newtheorem{lemma}[theorem]{Lemma}
\newtheorem{fact}[theorem]{Fact}
\newtheorem{corollary}[theorem]{Corollary}
\newtheorem{conjecture}[theorem]{Conjecture}
\theoremstyle{definition}
\newtheorem{definition}[theorem]{Definition}

\usepackage{hyperref}

\newcommand{\N}{\mathbb{N}}

\newcommand{\setofgraphs}{\mathcal{G}}
\newcommand{\isomorph}{\cong}

\newcommand{\sign}{\mathrm{sign}}

\newcommand{\set}[2]{\ensuremath{\left\{ #1 \mid #2 \right\}}}

\DeclareMathOperator{\Doma}{dom}
\DeclareMathOperator{\Ima}{Im}
\DeclareMathOperator{\Vars}{Vars}

\newcommand\restr[2]{{
        \left.\kern-\nulldelimiterspace 
        #1 
        \vphantom{\big|} 
        \right|_{#2} 
}}

\newcommand{\gr}[2][]{\text{gr}_{#1}(#2)}
\newcommand{\ccg}[1][]{\ensuremath{\ComplexityFont{G}_{#1}}}

\newcommand{\minorfree}{\ensuremath{\mathrm{MF}}}

\newcommand{\ccex}{\ensuremath{\ComplexityFont{A}}}
\newcommand{\ccall}{\ensuremath{\ComplexityFont{ALL}}}
\newcommand{\cctwoexp}{\ensuremath{\ComplexityFont{2EXP}}}

\newcommand{\gccex}{\ensuremath{\ccg \ComplexityFont{A}}}
\newcommand{\gccall}{\ensuremath{\ccg \ccall}}
\newcommand{\gccr}{\ensuremath{\ccg \ComplexityFont{R}}}
\newcommand{\gccp}{\ensuremath{\ccg \P}}
\newcommand{\gccnp}{\ensuremath{\ccg \NP}}
\newcommand{\gccph}{\ensuremath{\ccg \PH}}
\newcommand{\gccl}{\ensuremath{\ccg \ComplexityFont{L}}}
\newcommand{\gccexp}{\ensuremath{\ccg \EXP}}
\newcommand{\gcctwoexp}{\ensuremath{\ccg \cctwoexp}}
\newcommand{\gccac}{\ensuremath{\ccg \AC^0}}
\newcommand{\gcctc}{\ensuremath{\ccg \TC^0}}

\newcommand{\gccreg}{\ensuremath{\ccg \ComplexityFont{REG}}}

\newcommand{\gccpbs}{\ensuremath{\ComplexityFont{PBS}}}
\newcommand{\gccbpbs}{\ensuremath{\ComplexityFont{PBS}}}

\newcommand{\cogcsh}{\ensuremath{[\mathrm{Sparse} \cap \mathrm{Hereditary}]_{\subseteq}}}
\newcommand{\cogcsmallh}{\ensuremath{[\mathrm{Small} \cap \mathrm{Hereditary}]_{\subseteq}}}
\newcommand{\cogcsmallhsu}{\ensuremath{[\mathrm{Small} \cap \mathrm{Hereditary} \cap \mathrm{Self\text{-}Universal}]_{\subseteq}}}
\newcommand{\cogcpmc}{\ensuremath{[\mathrm{Sparse} \cap \mathrm{Minor\text{-}Closed}]_{\subseteq}}}
\newcommand{\cogcth}{\ensuremath{[\mathrm{Tiny} \cap \mathrm{Hereditary}]_{\subseteq}}}

\newcommand{\gcplanar}{\ensuremath{\mathrm{Planar}}}
\newcommand{\gcforest}{\ensuremath{\mathrm{Forest}}}

\newcommand{\gcinterval}{\ensuremath{\mathrm{Interval}}}

\newcommand{\cloneclosure}{\ensuremath{\mathrm{clone}}}
\newcommand{\clonebf}{\ensuremath{\mathrm{BF}}}
\newcommand{\sgclosure}{\ensuremath{\mathrm{sg}}}

\newcommand{\gccfoq}{\ensuremath{\ccg \FO}}
\newcommand{\gccfo}{\ensuremath{\ccg \FO_{\mathrm{qf}}}}
\newcommand{\gccfoeq}{\ensuremath{\ccg \FO(=)}}
\newcommand{\gccfoeqqf}{\ensuremath{\ccg \FO_{\mathrm{qf}}(=)}}
\newcommand{\gccfolt}{\ensuremath{\ccg \FO(<)}}
\newcommand{\gccfoltqf}{\ensuremath{\ccg \FO_{\mathrm{qf}}(\ltp)}}

\newcommand{\bfreduction}{\ensuremath{\leq_{\mathrm{BF}}}}
\newcommand{\sgreduction}{\ensuremath{\leq_{\mathrm{sg}}}}

\newcommand{\sgdreduction}{\ensuremath{\leq_{\mathrm{sg}}^{\mathrm{diag}}}}

\DeclareMathOperator{\ltp}{\ensuremath{<}}
\DeclareMathOperator{\addp}{\ensuremath{+}}
\DeclareMathOperator{\mulp}{\ensuremath{\times}}

\begin{document}

\title{A Complexity Theory for Labeling Schemes} 
\author{Maurice Chandoo\footnote{Leibniz Universität Hannover,
        Institut für Theoretische Informatik,
        Appelstr.~4, 30167 Hannover, Germany; \hskip 2em      
        E-Mail: \href{mailto:chandoo@thi.uni-hannover.de}{chandoo@thi.uni-hannover.de} }}
\date{\vspace{-5ex}}

\maketitle

\begin{abstract}\noindent\textbf{Abstract.}
    In a labeling scheme the vertices of a given graph from a particular class are assigned short labels such that adjacency can be algorithmically determined from these labels. A representation of a graph from that class is given by the set of its vertex labels. Due to the shortness constraint on the labels such schemes provide space-efficient representations for various graph classes, such as planar or interval graphs. We consider what graph classes cannot be represented by labeling schemes when the algorithm which determines adjacency is subjected to computational constraints.     
\end{abstract}


\section{Introduction}
Suppose you have a database and in one of its tables you need to store large interval graphs in such a way that adjacency can be determined quickly. This means generic data compression algorithms are not an option. A graph is an interval graph if each of its vertices can be mapped to a closed interval on the real line such that two vertices are adjacent iff their corresponding intervals intersect. A naive approach would be to store the adjacency matrices of the interval graphs. This requires roughly $n^2$ bits. However, there are only $2^{\mathcal{O}(n \log n)}$ interval graphs on $n$ vertices which means that a space-wise optimal representation of interval graphs would only use $\mathcal{O}(n \log n)$ bits. Adjacency lists perform even worse space-wise if the graphs are dense. A simple and asymptotically optimal solution for this problem is to use the interval representation. More specifically, given an interval graph $G$ with $n$ vertices write down its interval model (the set of intervals that correspond to its vertices), enumerate the endpoints of the intervals from left to right and label each vertex with the two endpoints of its interval, see Figure~\ref{fig:intvex}. The set of vertex labels is a representation of the graph and, moreover, adjacency of two vertices can be determined quickly by comparing their four endpoints. Each endpoint is a number in the range $1,\dots,2n$ and therefore one vertex label requires $2 \log 2n$ bits. Thus the representation of the whole graph requires $\mathcal{O}(n \log n)$ bits which is asymptotically optimal. 

The idea behind this representation can be generalized to other graph classes as follows. Let $\mathcal{C}$ be a graph class with similarly many graphs on $n$ vertices as interval graphs. We say $\mathcal{C}$ has a labeling scheme (or implicit representation) if the vertices of every graph in $\mathcal{C}$ can be assigned binary labels of length $\mathcal{O}(\log n)$  such that adjacency can be decided by an (efficient) algorithm $A$ which gets two labels as input. The algorithm $A$ must only depend on $\mathcal{C}$. We remark that labeling schemes can also be constructed for graph classes which have asymptotically more graphs than interval graphs by adjusting the label length. However, many important graph classes do indeed only have $2^{\mathcal{O}(n \log n)}$ graphs on $n$ vertices and therefore we shall restrict our attention to them; we call such graph classes small. 

The central question of this paper is what small graph classes \emph{do not} admit an implicit representation when imposing computational constraints on the label decoding algorithm $A$. 
We propose a formal framework inspired by classical computational complexity in order to investigate this open-ended question in a structured way. A complexity class in this setting is a set of graph classes, usually defined in terms of labeling schemes. The introduced complexity classes can be seen as a complexity measure for the adjacency structure of a graph class. 

\begin{figure}
    \centering
    \resizebox{11cm}{!}{%
        \begin{tikzpicture}[shorten >=1pt,auto,node distance=1.2cm,
  main node/.style={circle,draw}]

\newcommand*{\xoff}{0}%
\newcommand*{\yoff}{0.3}%

\newcommand*{\xoffg}{10}%
\newcommand*{\yoffg}{0}%

%

\draw[line width=0.6pt] (\xoff+1+0.2,\yoff+0.2) -- (\xoff+6-0.2,\yoff+0.2);
\draw[line width=0.6pt] (\xoff+0.5,\yoff) -- (\xoff+2,\yoff);
\draw[line width=0.6pt] (\xoff+2.7,\yoff) -- (\xoff+4.3,\yoff);
\draw[line width=0.6pt] (\xoff+5,\yoff) -- (\xoff+6.5,\yoff);
\draw[line width=0.6pt] (\xoff+7,\yoff) -- (\xoff+8,\yoff);

\node (l1) at (\xoff+0.5+0.1,\yoff-0.2) {$1$};
\node (l3) at (\xoff+2-0.1,\yoff-0.2) {$3$};
\node (l4) at (\xoff+2.7+0.1,\yoff-0.2) {$4$};
\node (l5) at (\xoff+4.3-0.1,\yoff-0.2) {$5$};
\node (l6) at (\xoff+5+0.1,\yoff-0.2) {$6$};
\node (l8) at (\xoff+6.5-0.1,\yoff-0.2) {$8$};
\node (l9) at (\xoff+7+0.1,\yoff-0.2) {$9$};
\node (l10) at (\xoff+8-0.1,\yoff-0.2) {$10$};
\node (l2) at (\xoff+1+0.2+0.1,\yoff+0.38) {$2$};
\node (l7) at (\xoff+6-0.2-0.1,\yoff+0.38) {$7$};

\node[main node] (a) at (\xoffg,\yoffg) {$1,3$};
\node[main node] (b) at (\xoffg+1,\yoffg) {$4,5$};
\node[main node] (c) at (\xoffg+2,\yoffg) {$6,8$};
\node[main node] (d) at (\xoffg+3,\yoffg) {$9,10$};
\node[main node] (e) at (\xoffg+1,\yoffg+1) {$2,7$};

\path[-]
(a) edge (e)
(b) edge (e)
(c) edge (e)
;


\end{tikzpicture}
    }
    \caption{Interval model and the resulting labeling of the interval graph}
    \label{fig:intvex}                
\end{figure}
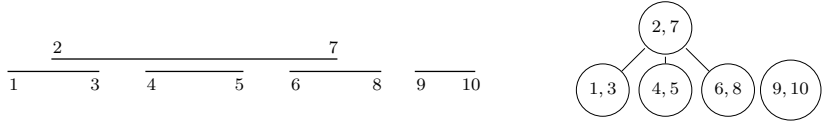 

The concept of labeling schemes was introduced by Muller \cite{muller} and by Kannan, Naor and Rudich \cite{kannan}. Among the first and most basic questions was whether every small and hereditary (= closed under vertex deletion) graph class has a labeling scheme. If one omits the hereditary condition then this is not true due to a simple counting argument \cite[Thm.~7]{spinrad}. Being hereditary is a uniformity condition that is commonly required of graph classes in order for them to have some meaningful structure. This restriction can be slightly weakened without affecting the previous question by considering all graph classes which are a subset of a small and hereditary graph class, the justification being that if a graph class $\mathcal{C}$ has an implicit representation then obviously every subset of $\mathcal{C}$ has an implicit representation as well. For instance, this weaker restriction also includes trees which are not hereditary but their hereditary closure, namely forests, are small and hereditary. This question remains unsolved and has been named the implicit graph conjecture (IGC) by Spinrad: every small and hereditary graph class has an implicit representation. He also gives an overview of graph classes known to have a labeling scheme and those not known to have one, which still remains quite accurate more than a decade later \cite{spinrad}. 

The following is a brief account on results related to the IGC.
Definitions of the mentioned graph classes and other properties are given in the next section. 
Let us call a small and hereditary graph class which is not known to have a labeling scheme a candidate for the IGC. Identifying such graph classes is important when trying  to study  the limitations of labeling schemes. The challenge is to determine whether a given hereditary graph class is small. If a graph class has a labeling scheme then it must be small since every graph on $n$ vertices in that class can be described using $\mathcal{O}(n \log n)$ bits. Clearly, this argument cannot be used to find candidates for the IGC. Another method to establish that a hereditary graph class is small is to apply Warren's theorem \cite{warren}. Roughly speaking, if adjacency in a graph class can be described by a set of polynomial inequations then it is small \cite[p.~54]{spinrad}; we give a possible formalization of this later on by what we call polynomial-boolean systems. This method can be used to show that certain geometrical intersection graph classes are small. Examples include line segment graphs, disk graphs and $k$-dot product graphs. All of the previous examples are hereditary and unknown to have a labeling scheme and thus are candidates for the IGC. The geometrical representations of these classes suggest potential labeling schemes. For instance, in the case of line segment graphs one could assign each vertex four numbers which represent the coordinates of the two endpoints of its line segment in the plane. However, this would only yield a correct labeling scheme if $\mathcal{O}(\log n)$ bits suffice to encode the coordinates for every line segment graph. In \cite{mcdiarmid} it is shown that this does not hold because there are line segment graphs whose geometrical representation requires an exponential number of bits. They prove that this lower bound also holds for disk graphs. In \cite{kang} it is shown that the geometrical representation of $k$-dot product also requires an exponential number of bits. Therefore for none of these graph classes their geometrical representation yields a labeling scheme as opposed to the case of interval graphs. 
In \cite{atminas} certain combinatorial tools are introduced that allow to infer that a hereditary graph class is small. By applying these tools they reveal candidates for the IGC which are defined in terms of their forbidden induced subgraphs. 
Coming from the other direction, one can try to identify sets of graph classes that have labeling schemes. In \cite{veryshort} it is shown that every tiny, hereditary graph class has a labeling scheme where the labels require only a constant number of bits. In \cite[p.~165]{spinrad} and in \cite{courcelle} labeling schemes for graph classes with bounded clique-width are constructed. In both cases the label decoding algorithm has to process the vertex labels in a bitwise fashion. This is a remarkable feature because we are not aware of any other hereditary graph class which requires this. 
The class of graphs with clique-width $k$ seems to have a higher complexity than other hereditary graph classes that are known to have a labeling scheme, as we shall see later on. 

\subparagraph*{Informative Labeling Schemes.}
Labeling schemes are usually understood as a much broader concept than what we consider here. In the following we try to give a more accurate picture of this notion. 
The idea of labeling schemes quickly evolved when Peleg recognized that instead of adjacency one can also infer other properties such as distance from the vertex labels \cite{peleg}.
In this more general form a labeling scheme can be described as follows. 
Suppose you have a graph class $\mathcal{C}$ (this can be any graph class, including the class of all graphs) and a property $P$ which applies to a $k$-tuple of vertices (in most cases  $k=2$ such as adjacency or distance). The goal is to label the vertices of every graph $G$ in $\mathcal{C}$ such that given a $k$-tuple of vertices $(v_1,\dots,v_k)$ of $G$ the property $P(v_1,\dots,v_k)$ can be quickly computed from the $k$ labels of these vertices. The three basic quality criteria of a labeling scheme are the label length measured in terms of the number of vertices, the time required to determine $P$ from the labels called decoding time and the time required to find a labeling for a given graph called encoding time. The encoding time is usually not regarded as important as the other two factors.  
The aim of such a labeling scheme is to reduce the time to determine $P$ by doing some preprocessing, i.e.~finding a labeling for a given graph (encoding time), and using some auxiliary space as determined by the label length.
The following is a non-exhaustive list of examples of labeling schemes for various properties.
 
In \cite{dlov} a distance labeling scheme for undirected graphs with label length $\frac{\log 3}{2}n + o(n)$ and constant decoding time is constructed. This means given an undirected graph $G$ with $n$ vertices one can find a representation of $G$ using $\mathcal{O}(n^2)$ bits such that it is possible to query distance between two vertices in constant time. 
Notice that it is impossible to achieve sublinear label length $o(n)$ in this case because this would contradict the $\Omega(n^2)$ bits required to store an arbitrary graph. For planar graphs a distance labeling scheme with label length $\mathcal{O}(\sqrt{n} \log n)$, decoding time $\mathcal{O}(\log n)$ and polynomial encoding time is given in \cite{gavd}. They also show a lower bound of $\Omega(n^{1/3})$ on the label length for distance in planar graphs. 

Another well-studied property is the ancestor relation in directed trees, which is motivated by the fact that  XML documents are viewed as labeled trees and search engines that operate on such documents naturally perform a lot of ancestor queries. In \cite{fraig} an ancestry labeling scheme is designed which has label length $\log n + \mathcal{O}(\log \log n)$, constant decoding time and linear encoding time. The label length is asymptotically tight, i.e.~there is a lower bound of $\log n + \Omega(\log \log n)$. In general, for the label length the multiplicative constant in front of the first-order term as well as the second-order term are considered to be theoretically important.

In \cite{katz} a flow labeling scheme for undirected graphs with optimal label length $\Theta(\log^2 n + \log n \cdot \log \hat{\omega} )$ ($\hat{\omega}$ is the maximum capacity) and polylogarithmic decoding time w.r.t.~$n$ and $\hat{\omega}$ is constructed. They also give labeling schemes for vertex- and edge-connectivity. The motivation for considering these properties stems from routing tasks in networks where it is useful to be aware of the capacity and connectivity between nodes. They note that for practical usability such a scheme has to be adapted to a dynamic setting since the topology of a network  cannot be assumed to be static. 

In \cite{courcelle} a $P$-labeling scheme for all graph classes of bounded clique-width is constructed where $P$ is an arbitrary property expressible as formula in monadic second-order logic using label length and decoding time $\mathcal{O}(\log n)$ and encoding time $\mathcal{O}(n \log n)$. This is also generalized to properties called monadic second-order optimization functions, which includes distance. 

A sequence of graphs $(G_n)_{n \in \N}$ is said to be universal for a graph class $\mathcal{C}$ if every graph on $n$ vertices in $\mathcal{C}$ occurs as induced subgraph of $G_n$ for all $n \in \N$. Universal graphs for undirected graphs and tournaments have already been studied in the 60's \cite{rado,moon,moon2}.
Adjacency labeling schemes can be seen as an alternative description of universal graphs. Suppose you have an adjacency labeling scheme $S$ with label length $f(n)$ for a graph class $\mathcal{C}$. Then the $n$-th universal graph $G_n$ for $\mathcal{C}$ derived from $S$ has $\{0,1\}^{f(n)}$ as vertex set and its adjacency is determined by the label decoder of $S$. Similarly, a sequence of universal graphs can be easily converted into an adjacency labeling scheme. The label length is logarithmically related to the size of the universal graphs. 
In \cite{alon17} it is shown that undirected graphs have a sequence of universal graphs of size $(1+o(1))\cdot 2^{(n-1)/2}$, and that this is tight up to a lower order additive term. In \cite{alstrup} it is shown that forests have an adjacency labeling scheme with optimal label length $\log n + \mathcal{O}(1)$, constant decoding time and linear encoding time. This implies that forests have a sequence of universal graphs of size $\mathcal{O}(n)$.

There are various relaxations and additional constraints for labeling schemes that have been proposed and investigated. Instead of measuring the maximum label length one can consider the average label length.
This allows for more flexibility when designing a labeling scheme because a few vertices can have longer labels. This measure is also called total label length \cite{gavd}. 
In \cite{korman} it is suggested to consider the labels of additional vertices when decoding a property. For example, in a distance labeling scheme one might use the labels of the two vertices $u,v$ for which distance has to be determined in order to identify a third vertex $w$. Then the label of $w$ can also be used to decode the distance between $u$ and $v$.   
They show that it is possible to circumvent certain lower bounds on the label length in this relaxed model.
In a dynamic labeling scheme the goal is to maintain a labeling of a graph which can change incrementally in a prescribed manner by updating the labels after each change. This can be formalized in different ways and depends on the class of graphs under consideration. A dynamic model for trees is introduced in \cite{dyn1} and a dynamic labeling scheme for XML documents is given in \cite{liu}.

In this paper we exclusively consider adjacency labeling schemes with label length at most $\mathcal{O}(\log n)$. For brevity we omit the qualifier `adjacency'.

\subparagraph*{Overview of the Paper.} 
In Section~\ref{sec:prelim}  we define how to interpret a set of languages (complexity class) such as $\P$ as a set of labeling schemes. For instance, we write $\gccp$ to denote the set of graph classes that can be represented by a labeling scheme where the label decoder can be computed `in' $\P$.  Our definition generalizes the ones given in \cite{kannan} and  \cite{muller}.    
Until now it has been unknown whether the computational aspect of label decoders matters at all, i.e.~$\gccac = \gccall$ would have been consistent with what was known so far ($\ComplexityFont{ALL}$ denotes the set of all languages and $\AC^0$ is a tiny circuit complexity class). 
In Section~\ref{sec:hierarchy} we show that there is a time hierarchy for these $\ccg \cdot$ classes akin to the one known from classical complexity.

In Section~\ref{sec:primitive} we provide a picture of the expressiveness of computationally primitive labeling schemes. For example, graph classes with bounded clique-width, uniformly sparse graph classes and various intersection graph classes can all be found in $\gccac$. We introduce two graph parameters called pointer numbers which are defined in terms of very rudimentary labeling schemes. Despite their simplicity they already generalize uniformly sparse graph classes and graph classes with bounded degree.
Any natural graph class with a labeling scheme that we are aware of can be found in $\gccac$. This suggests that a more restrictive model of computation might be needed to analyze the limitations of labeling schemes.

Before we concern ourselves with finding such a model of computation, we first introduce two reduction notions for graph classes in Section~\ref{sec:reductions}. Informally, we say a graph class is reducible to another one if the adjacency of graphs in the former class can be expressed as boolean combination of the adjacency of graphs in the latter one. This enables us to formally compare graph classes for which no labeling schemes are known in order to see whether there could be a common obstacle to designing a labeling scheme.  Moreover, we consider what a `complete graph class' for classes such as $\gccp$ looks like if it exists. The main part of this section is dedicated to establishing basic properties of these two reduction notions.

In Section~\ref{sec:logic} we define classes of labeling schemes whose label decoders can be expressed by formulas in first-order logic.  
In this setting a vertex label is interpreted as a constant number $k$ of polynomially bounded, non-negative integers and the label decoder is a formula with $2k$ free variables. 
Such formulas neatly capture the naive labeling schemes for many geometrical intersection graph classes such as interval graphs. 
First, we compare the expressiveness of such logical labeling schemes to classes such as $\gccac$ and $\gccp$. Then we show that two fragments of these logical labeling schemes have various natural, complete graph classes such as trees or interval graphs. 
After that, we consider a generalization of quantifier-free logical labeling schemes called polynomial-boolean systems. Roughly speaking, polynomial-boolean systems are obtained by dropping the constraint that the numbers of a vertex label need to be polynomially bounded. Such systems are interesting because they contain many of the candidates for the implicit graph conjecture.

In Section~\ref{sec:algo} we consider the difficulty of solving certain algorithmic tasks on graph classes with labeling schemes where the label decoder has a very low complexity. Intuitively, one would expect such graph classes to have a relatively simple adjacency structure which can be exploited for algorithmic purposes. Somewhat unexpectedly, this relative simplicity is still high enough to lead us to problems which seem to be situated on the frontier of algorithmic research.  
We show that this kind of question naturally fits into the framework of parameterized complexity. For example, whether the Hamiltonian cycle problem is W[1]-hard when parameterized by $\gccp$ is a well-defined question (the answer is: yes, it is). 

In Section~\ref{sec:reg} we introduce another class of labeling schemes defined in terms of regular languages. Since finite state automata are quite sensitive to the order of the input bits the labels of the two vertices for which adjacency has to be determined are interleaved.
We show that the set of graph classes represented by such labeling schemes is closed under reductions and it contains every hereditary graph class with a labeling scheme that we are aware of, just like $\gccac$. 

In Section~\ref{sec:final} we summarize the various complexity classes and their relations, and give some motivation in retrospect. We then point out what we believe to be interesting research directions and state a potentially easier to refute variant of the implicit graph conjecture.

Besides the preliminaries all sections can be read more or less independently of each other. The only exception is when we discuss closure under reductions and complete graph classes, which require Section~\ref{sec:reductions}. For the reader who is primarily interested in algorithmic problems we recommend reading about the pointer numbers and algebraic reductions and then going to Section~\ref{sec:algo}.

\section{Preliminaries}
\label{sec:prelim}
\paragraph*{General Notation and Terminology.}
Let $\N = \{1,2,\dots\}$ be the set of natural numbers and $\N_0$ is $\N \cup \{0\}$. For $n \in \N$ let $[n] = \{1,2,\dots,n\}$ and let $[n]_0 = [n] \cup \{0\}$.  For a set $A$ we write $\mathcal{P}(A)$ to denote the power set of $A$. 
When we say $\log n$ we mean $\lceil \log_2 n \rceil$. Let $\exp(n) = 2^n$ and let $\exp^0(n) = n$ and $\exp^i(n) = \exp({\exp^{i-1}(n)})$ for all $i \geq 1$. For a function $f$ we write $\Ima(f)$ and $\Doma(f)$ to denote its image and domain, respectively. 
Let $f$ be a $k^2$-ary boolean function and let $A= (a_{i,j})_{i,j \in [k]}$ be a $(k \times k)$-matrix over $\{0,1\}$ for some $k \in \N$. We write $f(A)$ to mean $f(a_{1,1},a_{1,2},\dots,a_{1,k},a_{2,1},\dots,a_{k,k})$, i.e.~plugging in the values of $A$ going from left to right and top to bottom.  
In general, we consider an undirected graph to be a special case of a directed graph with symmetric edge relation. In Section~\ref{sec:primitive} it is more adequate to only consider undirected graphs and in that case we say graph to mean undirected graph. At certain points there might be more than one sensible interpretation. For example, suppose there are two sets of graph classes $\mathbb{A}$ and $\mathbb{B}$ with $\mathbb{A} \subseteq \mathbb{B}$ and $\mathbb{A}$ only contains undirected graph classes whereas $\mathbb{B}$ also contains directed graph classes. It trivially follows that $\mathbb{A} \subsetneq \mathbb{B}$. However, the more interesting question is whether $\mathbb{B}$ also contains an undirected graph class which is not in $\mathbb{A}$. If we do not know this we shall not consider $\mathbb{A}$ to be a proper subset of $\mathbb{B}$. We write $\mathcal{G}$ to denote the set of all graphs or the set of all undirected graphs, depending on the context. 
For two graphs $G,H$ we write $G \isomorph H$ to denote that they are isomorphic and $G \subseteq H$ to mean that $G$ is an induced subgraph of $H$. For a graph $G$ and a subset of its vertices $V'$ we write $G[V']$ to mean the subgraph of $G$ which is induced by $V'$. We write $\overline{G}$ to denote the edge-complement of a graph $G$. For a directed graph $G$ and a vertex $u$ of $G$ we write $N_{\mathrm{in}}(u)$ and $N_{\mathrm{out}}(u)$ to denote the in- and out-neighbors of $u$ in $G$, respectively. For an undirected graph $G$ and a vertex $u$ of $G$ we write $N(u)$ to denote the vertices adjacent to $u$ in $G$. We say two vertices $u$ and $v$ of an undirected graph $G$ are twins if $N(u) \setminus \{v\} = N(v) \setminus \{u\}$. The twin relation is an equivalence relation.
The graph $K_n$ denotes the complete graph on $n$ vertices for $n \in \N$.
We speak of $G$ as unlabeled graph to emphasize that we talk about the isomorphism class of $G$ rather than a specific adjacency matrix of $G$. A graph class is a set of finite and unlabeled graphs, i.e.~it is closed under isomorphism. For a graph class $\mathcal{C}$ and $n \in \N$ we write $\mathcal{C}_n$  to denote the set of all graphs in $\mathcal{C}$ with $n$ vertices. Similarly, we write $\mathcal{C}_{\leq n}$ ($\mathcal{C}_{\geq n}$) to denote all graphs in $\mathcal{C}$ with at most (at least) $n$ vertices. For a set of graph classes $\mathbb{A}$ we write $[\mathbb{A}]_{\subseteq}$ to denote its closure under subsets, i.e.~$\set{ \mathcal{C} \subseteq \mathcal{D}  }{ \mathcal{D} \in \mathbb{A}}$.

\paragraph*{Complexity Classes.}
We use the term complexity class informally to mean a countable set of languages with computational restrictions. Unless specified otherwise we consider languages over the binary alphabet $\{0,1\}$. We will talk about the standard complexity classes $\ComplexityFont{L}$ (logspace), $\P$, $\NP$, $\PH$ (polynomial-time hierarchy), $\EXP$ and $\R$ (set of decidable languages), which can be found in most textbooks on complexity theory such as \cite{selman}, and the circuit complexity classes $\AC^0$  and $\TC^0$. The class $\AC^0$ is the set of languages decidable by constant-depth, polynomially sized circuit families using `and', `or' and `negation' gates. $\TC^0$ is defined analogously with the addition of majority gates. Precise definitions of these two classes can be found in \cite{vollmer}. We assume our circuit families to be logspace-uniform, i.e.~for every family of circuits $(C_n)_{n \in \N}$ considered here there is a logspace-transducer $M$ which outputs a representation of $C_i$ on input $i$ in unary.  
We write $\ComplexityFont{ALL}$ to denote the set of all languages.
For a function $t \colon \N \rightarrow \N$ let $\TIME(t)$ denote the set of languages that can be decided by a deterministic Turing machine in time $t$. 
For a set of functions $T$ where every $t \in T$ has signature $\N \rightarrow \N$ let $\TIME(T) = \cup_{t \in T} \TIME(t)$. The class $k\EXP$ is defined as $\TIME(\exp^k(n^{\mathcal{O}(1)}))$ for $k \geq 0$. This means $\ComplexityFont{0}\EXP = \P$. If $k=1$ we simply write $\EXP$.

\paragraph*{First-Order Logic.}
Let $\mathcal{N}$ be the structure that has $\N_0$ as universe equipped with the order  relation `$\ltp$' and addition `$\addp$' and multiplication `$\mulp$' as functions. 
For $n \geq 1$ let $\mathcal{N}_n$ be the structure that has $[n]_0 = \{0,1,\dots,n \}$ as universe, the order relation `$\ltp$' and addition as well as multiplication defined as:
$$\addp(x,y) = \begin{cases}
x+y &, \text{if } x+y \leq n \\
0 &,   \text{if } x+y > n \\
\end{cases}
\: \: , \: \: 
\mulp(x,y) = \begin{cases}
xy &, \text{if } xy \leq n \\
0 &,   \text{if } xy > n \\
\end{cases}
$$
For $\sigma \subseteq \{ \ltp, \addp, \mulp \}$ let $\FO_k(\sigma)$ be the set of first-order formulas with boolean connectives $\neg,\vee,\wedge$, quantifiers $\exists,\forall$ and  $k$ free variables using only equality and the relation and function symbols from $\sigma$. If $\sigma = \{\ltp, \addp, \mulp \}$ we simply write $\FO_k$ and if $\sigma = \emptyset$ we write $\FO_k(=)$.
A formula is called an atom if it contains no boolean connectives and no quantifiers. For a formula $\varphi$ with $a$ atoms let us call the $a$-ary boolean function that results from replacing every atom in $\varphi$ by a proposition the underlying boolean function of $\varphi$. 
Let $\Vars(\varphi)$ be the set of free variables in $\varphi$. Given $\varphi \in \FO_k(\sigma)$, $\Vars(\varphi) = (x_1,\dots,x_k)$ and an assignment $a_1,\dots,a_k \in [n]_0$ we write $\mathcal{N}_n, (a_1,\dots,a_k) \models \varphi$ if the interpretation $\mathcal{N}_n, (a_1,\dots,a_k)$ satisfies $\varphi$ under the semantics of first-order logic. 

Let $\varphi$ be a formula in $\FO_k$. We define the bounded model checking problem for $\varphi$ as follows. On input $a_1,\dots,a_k,n \in \N$ with $a_i \in [n]_0$ for all $i \in [k]$ decide whether $\mathcal{N}_n,(a_1,\dots,a_k) \models \varphi$. We assume that the input is encoded in binary.

\paragraph*{Graph Class Properties.}
Let $\mathcal{C}$ be a graph class.
We call $\mathcal{C}$ small if it has at most $n^{\mathcal{O}(n)}$ graphs on $n$ vertices. Stated differently, $|\mathcal{C}_n| \in n^{\mathcal{O}(n)} = 2^{\mathcal{O}(n \log n)}$; in the literature this is also called factorial speed of growth (\cite{balogh,atminas}). 
$\mathcal{C}$ is tiny if there exists a $c < \frac{1}{2}$ such that $|\mathcal{C}_n| \leq n^{cn}$ for all sufficiently large $n$ \cite{veryshort}. 
$\mathcal{C}$ is hereditary if it is closed under taking induced subgraphs, i.e.~if $G$ is in $\mathcal{C}$ then every induced subgraph of $G$ is in $\mathcal{C}$. 
We write $[\mathcal{C}]_{\subseteq}$ to denote the hereditary closure of $\mathcal{C}$, i.e.~$[\mathcal{C}]_{\subseteq}$ is the set of graphs that occur as induced subgraph of some graph in $\mathcal{C}$. 
$\mathcal{C}$ is sparse if there exists a $c \in \N$ such that every graph in $\mathcal{C}_n$ has at most $cn$ edges for all $n \in \N$.
$\mathcal{C}$ is uniformly sparse if it is a subset of a hereditary and sparse graph class \cite{courcelleus}. Stated differently, $\mathcal{C}$ is uniformly sparse iff it is in $\cogcsh$. For example, planar graphs are uniformly sparse. 
$\mathcal{C}$ is inflatable if for every graph $G$ in $\mathcal{C}$ with $n$ vertices there exists a graph $H$ in $\mathcal{C}$ with $m$ vertices which contains $G$ as induced subgraph for all $n < m \in \N$. 
$\mathcal{C}$ is self-universal if for every finite subset $X$ of $\mathcal{C}$ there exists a graph $G$ in $\mathcal{C}$ which contains every graph in $X$ as induced subgraph. Every graph class that is closed under disjoint union is self-universal.

Let $\mathcal{F}$ be a family of sets over some ground set $U$, i.e.~$\mathcal{F} \subseteq \mathcal{P}(U)$. Let $X$ be a finite subset of $\mathcal{F}$ and $G_X$ is the undirected graph which has $X$ as vertex set and there is an edge $\{u,v\}$ in $G_X$ iff $u$ and $v$ intersect. Let $\mathcal{C}_{\mathcal{F}}$ be the following graph class. A graph $G$ is in $\mathcal{C}_{\mathcal{F}}$ iff there exists a subset $X$ of $\mathcal{F}$ such that $G$ is isomorphic to $G_X$. A graph class $\mathcal{C}$ for which there exists a family of sets $\mathcal{F}$ such that $\mathcal{C} = \mathcal{C}_{\mathcal{F}}$ is called intersection graph class. Note that every intersection graph class is undirected, hereditary, self-universal and inflatable. For example, interval graphs are defined as the class of graphs $\mathcal{C}_{\mathcal{F}}$ where $\mathcal{F}$ is the set of (closed) intervals on the real line.

An undirected graph $H$ is called a minor of an undirected graph $G$ if $H$ can be obtained from $G$ by deleting vertices and edges, and contracting edges (merging two adjacent vertices into one vertex which inherits the neighbors of the two old vertices). $\mathcal{C}$ is minor-closed if every graph that occurs as minor of some graph in $\mathcal{C}$ is in $\mathcal{C}$ as well.
For a graph $G$ we call $\mathcal{C}$  $G$-minor free if no graph in $\mathcal{C}$ contains $G$ as minor. For a graph class $\mathcal{C}$ let $\minorfree(\mathcal{C})$ denote the set of graph classes that are $G$-minor free for some $G$ in $\mathcal{C}$.   
Being minor-closed is a property that only applies to undirected graph classes.

\paragraph*{Graph Classes and Parameters.}
We consider a graph parameter $\lambda$ to be a total function which maps unlabeled graphs to natural numbers. We say a graph class $\mathcal{C}$ is bounded by a graph parameter $\lambda$ if there exists a $c \in \N$ such that $\lambda(G) \leq c$ for all $G \in \mathcal{C}$. A graph parameter can be interpreted as the set of graph classes that are bounded by it. We say two graph parameters are equivalent if they bound the same set of graph classes.

The degeneracy of a graph $G$ is the least $k \in \N$ such that every induced subgraph of $G$ contains a vertex of degree at most $k$. For example, every forest has degeneracy 1 because it either has a leaf or every vertex is isolated.
The arboricity of a graph $G$ is the least $k \in \N$ such that there are $k$ forests $F_1,\dots,F_k$ with the same vertex set as $G$ such that $E(G) = \cup_{i \in [k]} E(F_i)$. 
The thickness of a graph $G$ is the least $k \in \N$ such that there are $k$ planar graphs $H_1,\dots,H_k$ with the same vertex set as $G$ such that $E(G) = \cup_{i \in [k]} E(H_i)$. 
It is well-known that arboricity, thickness and degeneracy are equivalent: they bound exactly the set of uniformly sparse graph classes. 
The boxicity of a graph $G$ is the least $k \in \N$ such that $G$ is the intersection graph of $k$-dimensional axis-parallel boxes. For example, a graph class has boxicity 1 iff it is an interval graph.
The twin index of a graph $G$ is the index of its twin relation. 
The intersection number of a graph $G$ is the smallest $k \in \N$ such that there exists a family of sets $\mathcal{F}$ over a ground set $U$ with $|U| = k$ and $G$ is isomorphic to $G_\mathcal{F}$. 
The tree-width and clique-width of a graph are two algorithmically important graph parameters. Their definitions can be found in \cite{cwtwd} but they are not immediately relevant in our context, hence we omit them. For us it suffices to know that a graph class has bounded tree-width iff it is in $\minorfree(\gcplanar)$ (it is the subset of a minor-closed graph class which does not contain all planar graphs) \cite{twplanar}. The relevant property of graph classes with bounded clique-width is cited when used.

An interval graph is an intersection graph of closed intervals on the real line. A $k$-interval graph is the intersection graph of a union of $k$ disjoint intervals on the real line. 
A chord of a circle is a line segment whose endpoints lie on that circle. A circle graph is an intersection graph of chords of a circle. A line segment graph is an intersection graph of line segments in the plane. Obviously, interval and circle graphs are a subset of line segment graphs. A disk graph is an intersection graph of disks in the plane. A $k$-ball graph is the intersection graph of $k$-dimensional balls. 2-ball graphs are disk graphs. A $k$d-line segment graph is the intersection graph of line segments in $\mathbb{R}^k$. A graph $G$ is a $k$-dot product graph if there exists a function $f \colon V(G) \rightarrow \mathbb{R}^k$ such that two distinct vertices $u,v$ in $G$ are adjacent iff $f(u) \cdot f(v) \geq 1$ \cite{fiduccia}.

\paragraph*{Labeling Schemes.}
We use the terms implicit representation and labeling scheme interchangeably. 

\begin{definition}
    A labeling scheme is a tuple $S = (F,c)$ where $F \subseteq \{0,1\}^* \times \{0,1\}^*$ is called  label decoder and $c \in \mathbb{N}$ is the label length. A graph $G$ on $n$ vertices is in the class of graphs spanned by $S$, denoted by $G \in \gr{S}$, if there exists a labeling $\ell \colon V(G) \rightarrow \{0,1\}^{c \log n}$ such that for all $u, v \in V(G)$:
    $$ (u,v) \in E(G) \Leftrightarrow (\ell(u),\ell(v)) \in F $$
    We say a graph class $\mathcal{C}$ is represented by (or has) a labeling scheme $S$ if $\mathcal{C} \subseteq \gr{S}$.
\end{definition}

The labeling scheme for interval graphs that we have seen in the introduction can be formalized as follows. We define the label decoder $F_{\mathrm{Intv}}$ such that $(x_1x_2,y_1y_2)$ is in $F_{\mathrm{Intv}}$ iff neither $x_2$ is (lexicographically) smaller than $y_1$ nor $y_2$ is smaller than $x_1$ for all $x_1,x_2,y_1,y_2 \in \{0,1\}^m$ and $m \in \N$. If $x_1,x_2,y_1,y_2$ are interpreted as natural numbers and we assume that $x_1 \leq x_2$ and $y_1 \leq y_2$ then the label decoder says that neither of the two intervals $[x_1,x_2]$, $[y_1,y_2]$ ends before the other one starts; this means that they must intersect. The label length in this case is $c=4$ because
two numbers from $[2n]$ can be encoded using $4 \log n$ bits for $n \geq 2$. The labeling scheme $(F_{\mathrm{Intv}},4)$ represents interval graphs. 

\begin{definition}
    A language $L \subseteq \{0,1\}^*$ induces the following label decoder $F_L$. For all $x,y \in \{0,1\}^*$ it holds that $(x,y) \in F_L$ iff $xy \in L$ and $|x| = |y|$.
    For a set of languages $\ComplexityFont{A}$  we say that a graph class $\mathcal{C}$ is in $\ccg \ComplexityFont{A}$ if there exists a language $L \in \ComplexityFont{A}$ and $c \in \N$ such that $\mathcal{C}$ is represented by the labeling scheme $(F_L,c)$.        
    \label{def:genericgcc}
\end{definition}

We say a labeling scheme $S=(F,c)$ is in $\gccex$ if there exists a language $L$ in $\ccex$ such that $F = F_L$. 

For every set of languages $\ccex$ the set of graph classes $\ccg \ccex$ is trivially closed under subsets. Also, if $\ccex$ is closed under complement then $\ccg \ccex$ is closed under edge-complement. A graph class has a polynomially sized sequence of universal graphs (polynomial universal graphs for short) iff it is in $\gccall$ (no computational complexity constraint).  

It is not difficult to see that there is a language $L$ in $\AC^0$ such that $F_L = F_{\mathrm{Intv}}$ and therefore interval graphs are in $\gccac$. It is also an easy exercise to show that forests, circle graphs and all graph classes with bounded interval number, arboricity or boxicity are in $\gccac$.

In this terminology the implicit graph conjecture can be stated as follows.

\begin{conjecture}[Implicit Graph Conjecture]
    Every small and hereditary graph class is in $\gccp$. 
\end{conjecture}

We remark that is not even known whether every such graph class is in $\gccall$. This can be seen as a purely graph-theoretical question which asks whether every small, hereditary graph class has polynomial universal graphs.

\section{Hierarchy of Labeling Schemes}
\label{sec:hierarchy}
When labeling schemes were introduced by Muller in \cite{muller} the label decoder was required to be computable. Clearly, this is a reasonable restriction since otherwise it would be impossible to query edges in a labeling scheme with an undecidable label decoder. Taking this consideration a step further, in order for a labeling scheme to be practical querying an edge should be a quick operation, i.e.~at least sublinear with respect to the number of vertices. Kannan et al acknowledged this by stating in their definition of an implicit representation that the label decoder must be computable in polynomial time  \cite{kannan}. As a consequence querying an edge in such a labeling scheme takes only polylogarithmic time. 

There can be different labeling schemes that represent the same graph class, just as there are different Turing machines that decide the same language.  Therefore one might ask whether every graph class that is represented by some labeling scheme with an undecidable label decoder can also be represented by a labeling scheme with a decidable label decoder. Similarly, can every labeling scheme with a decidable label decoder be replaced by one that has a polynomial-time decidable label decoder? The latter question is equivalent to asking whether Muller's definition coincides with that of Kannan et al. Spinrad remarked that it is not known whether these definitions are equivalent and no difference could be observed on the graph classes considered so far \cite[p.~22]{spinrad}. In our terminology the former question can be phrased as $\gccall \stackrel{?}{=} \gccr$ and the latter as $\gccr \stackrel{?}{=} \gccp$.  
We resolve these two questions by proving an analogue of the time hierarchy from classical complexity which shows that all of these three classes are distinct.

We say a function $t \colon \N \rightarrow \N$ is time-constructible if there exists a Turing machine that halts after exactly $t(n)$ steps for every input of length $n$ and all $n \in \N$.

\begin{theorem}
    For every time-constructible function $t \colon \N \rightarrow \N$ it holds that
    $\ccg \TIME(t(n)) \subsetneq \ccg \TIME(\exp^2(n) \cdot t(n))$.
    \label{thm:hierarchy}
\end{theorem}

Note that separations in the classical context do not necessarily extend to this setting, i.e.~given two sets of languages $\ccex \subsetneq \ccex'$ it must not be the case that $\gccex \subsetneq \gccex'$. Therefore the previous theorem does not directly follow from the original time hierarchy theorem.

\begin{corollary}
 $\gccexp \subsetneq \gcctwoexp \subsetneq \dots \subsetneq \gccr \subsetneq \gccall$.
\end{corollary}
\begin{proof}
    We explain why $\gccexp \subsetneq \gcctwoexp$ follows from Theorem~\ref{thm:hierarchy}. The same argument shows that $\ccg k \EXP \subsetneq \ccg (k+1) \EXP$ for all $k \geq 1$. Recall that $\EXP$ equals the infinite union of $\TIME(\exp(n^c))$ over all $c \in \N$ and therefore Theorem~\ref{thm:hierarchy} cannot be applied directly. However, $\EXP \subseteq \TIME(\exp(1.5^n))$ and $\ccg \TIME(\exp(1.5^n)) \subsetneq \gcctwoexp$ does follow from Theorem~\ref{thm:hierarchy} and therefore $\gccexp \subsetneq \gcctwoexp$ holds. As a consequence $\ccg k \EXP$ is a strict subset of $\gccr$ for every $k \geq 0$. That $\gccr$ is a strict susbet of $\gccall$ follows from the fact that its diagonalization graph class $\mathcal{C}_{\R}$ is not in $\gccr$ (see Lemma~\ref{lem:ca_not_in_a}) and every diagonalization graph class has a labeling scheme and therefore lies in $\gccall$ (see Definition~\ref{def:dgcls} and the subsequent paragraph). 
\end{proof}

The basic idea behind the proof of Theorem \ref{thm:hierarchy} is the following diagonalization argument. Let $\ComplexityFont{A} = \{F_1,F_2, \dots \}$ be a set of label decoders. Then a labeling scheme in $\ccg \ComplexityFont{A}$ can be seen as pair of natural numbers, one for the label decoder and one for the label length. Let $\tau : \N \rightarrow \N^2$ be a surjective function and $S_{\tau(x)}$ is the labeling scheme $(F_y,z)$ with $\tau(x) = (y,z)$.
It follows that for every labeling scheme $S$ in $\ccg \ComplexityFont{A}$ there exists an $x \in \N$ such that $S = S_{\tau(x)}$.  
The following graph class cannot be in $\ccg \ComplexityFont{A}$:
$$ G \in \mathcal{C}_\ComplexityFont{A} \Leftrightarrow G \text{ is the smallest graph on $n$ vertices s.t.~} G \notin \gr{S_{\tau(n)}} $$ 
where smallest is meant w.r.t.~some order such as the lexicographical one. Note that the order must be for unlabeled graphs. However, an order for labeled graphs can be easily adopted to unlabeled ones.   
Assume $\mathcal{C}_\ComplexityFont{A}$ is in $\ccg \ComplexityFont{A}$ via the labeling scheme $S$. There exists an $n \in \N$ such that $S = S_{\tau(n)}$ and it follows that $\mathcal{C}_\ComplexityFont{A}$ contains a graph on $n$ vertices that cannot be in $S$ per definition, contradiction. Then it remains to show that $\mathcal{C}_\ComplexityFont{A}$ is in the class that we wish to separate from $\ccg \ComplexityFont{A}$. 

For the remainder of this section we formalize this idea in three steps. First, we state the requirements for a pairing function $\tau$ and show that such a function exists. We continue by arguing that the diagonalization graph class $\mathcal{C}_{\ComplexityFont{A}}$ is not contained $\ccg \ComplexityFont{A}$.  In the last step we construct a label decoder for $\mathcal{C}_{\TIME(t(n))}$ and show that it can be computed in time $\exp^2(n) \cdot t(n)$. 

\begin{definition}
    A surjective function $\tau : \N \rightarrow \N^2$ is an admissible pairing if all of the following holds: 
    \begin{enumerate}
        \item $|\tau^{-1}(y,z)|$ is infinite for all $y,z \in \N$
        \item $y,z \leq \log x$ for all $x \geq 1$ and $\tau(x) = (y,z)$
        \item $\tau(x)$ is undefined if $x$ is not a power of two
        \item $\tau$ is computable in polynomial time given its input in unary.
    \end{enumerate}
    \label{def:pairing}
\end{definition}

Note, that a graph on $n$ vertices has labels of the same length as a graph on $m$ vertices whenever $\log n = \log m$ (rounded up). The third condition prevents this from happening, i.e.~for all $G \neq H \in \mathcal{C}_{\ComplexityFont{A}}$ it holds that their vertices are assigned labels of different length.
\begin{lemma}
    There exists an admissible pairing function.
    \label{lem:pairing}
\end{lemma}
\begin{proof}
    Consider the function $\tau(x) = (y,z)$ iff $x = \exp({2^y \cdot 3^z \cdot  5^w})$ for some $w \geq 0$. The first condition of Definition \ref{def:pairing} holds because for every $w \in \N$ there exists an $x$ with $\tau(x) = (y,z)$. For the second condition assume that there exists an $x \geq 1$ and $\tau(x) = (y,z)$ such that $y > \log x$. This cannot be the case because then $x < 2^y$ which contradicts $\log x = 2^y 3^z 5^w$. The same applies to $z$.
    The third condition is obvious. For the fourth condition observe that on input $1^x$ it suffices to consider $y,z,w$ between $0$ and $\log (\log x) $ such that $\log x = 2^y3^z5^w$.   
\end{proof}

\begin{definition}
    Let $\ComplexityFont{A} = \{L_1,L_2,\dots \}$ be a countable set of languages, $\prec$ an order on unlabeled graphs and $\tau$ an admissible pairing. For $n \in \N$ and $\tau(n) = (y,z)$ let $S_{\tau(n)}$ be $(F_{L_y},z)$.
    The diagonalization graph class of $\ComplexityFont{A}$ is defined as:
    $$ \mathcal{C}_\ComplexityFont{A} = \bigcup_{n \in \Doma(\tau)} \left\{ G \in \mathcal{G}_n    \: \middle|            
    G \text{ is the smallest graph w.r.t.~$\prec$ not in } \gr{S_{\tau(n)}}                                
    \right\} $$
    where $\mathcal{G}_n$ denotes the set of all graphs on $n$ vertices. 
    \label{def:diaggc}
\end{definition}

When we consider the diagonalization graph class of a set of languages we assume the lexicographical order for $\prec$ and the function given in the proof of Lemma~\ref{lem:pairing} for $\tau$.
\begin{lemma}
    For every countable set of languages $\ComplexityFont{A}$ it holds that  $\mathcal{C}_\ComplexityFont{A} \notin \ccg \ComplexityFont{A}$.
    \label{lem:ca_not_in_a}
\end{lemma}
\begin{proof}
    As argued in the paragraph after Theorem \ref{thm:hierarchy} it holds that for any labeling scheme $S$ in $\ccg \ComplexityFont{A}$ there exists a graph $G$ that is in $\mathcal{C}_{\ComplexityFont{A}}$ but not in $\gr{S}$ and thus this lemma holds. 
    If the labeling scheme $S$ is in $\ccg \ComplexityFont{A}$ then there exists an $n \in \N$ such that $S = S_{\tau(n)}$ where $S_{\tau(n)} = (F_y,z)$, $\tau(n)=(y,z)$ and  $\ComplexityFont{A} = \{F_1,F_2,\dots \}$. Therefore there must be a graph $G$ on $n$ vertices in $\mathcal{C}_{\ComplexityFont{A}}$ which is not in $\gr{S}$. Observe that this argument is not quite correct. It only works if $\gr{S}$ does not contain all graphs on $n$ vertices, otherwise such a graph $G$ does not exist. However, due to the fact that $|\tau^{-1}(y,z)|$ is infinite it follows that there exists an arbitrarily large $n \in \N$ such that $S = S_{\tau(n)}$. And since $\gr{S}$ is small it follows that it does not contain all graphs on $n$ vertices for sufficiently large $n$.    
\end{proof}

To show that $\mathcal{C}_\ComplexityFont{A}$ is in some class $\ccg \ComplexityFont{B}$ we need to define a labeling scheme $S_{\ComplexityFont{A}} = (F_{\ComplexityFont{A}},1)$ that represents $\mathcal{C}_\ComplexityFont{A}$ and consider the complexity of computing its label decoder.
\begin{definition}
    Let $\mathcal{C}$ be a graph class such that $|\mathcal{C}_n| = 0$ whenever $n$ is not a power of two and $|\mathcal{C}_n| \leq 1$  whenever $n$ is a power of two.
    For $G \in \mathcal{C}$ let $G_0$ denote the lexicographically smallest labeled graph with $G_0 \isomorph G$ and $V(G_0) = \{0,1\}^m$. We define the label decoder $F_{\mathcal{C}}$ as follows. For every $m \in \N$ such that there exists $G \in \mathcal{C}$ on $2^m$ vertices and for all $x, y \in \{0,1\}^m$ let
    $$ (x,y) \in F_{\mathcal{C}} \Leftrightarrow (x,y) \in E(G_0)  $$  
    \label{def:dgcls}
\end{definition}

Observe that the diagonalization graph class $\mathcal{C}_{\ccex}$ of some set of languages $\ccex$ satisfies the prerequisite of the previous definition and the labeling scheme $(F_{\mathcal{C}_{\ccex}},1)$ represents $\mathcal{C}_{\ccex}$. Instead of $F_{\mathcal{C}_{\ccex}}$ we simply write $F_{{\ccex}}$ for the label decoder.  

Up to this point the exact correspondence between $y \in  \N$ and the label decoder $F_y$ was not important since we only required the set of label decoders $\ComplexityFont{A}$ to be countable. To show that the label decoder $F_{\TIME(t(n))}$ can be computed in time $\exp^2(n)\cdot t(n)$ it is important that the label decoder $F_y$ for a given $y \in \N$ can be effectively computed. 

\begin{lemma}
    For every time-constructible $t \colon \N \rightarrow \N$ there exists a mapping $f \colon \N \rightarrow \ComplexityFont{ALL}$ such that $\Ima(f) = \TIME(t(n))$. Additionally, on input $x \in \N$ in unary and $w \in \{0,1\}^*$ the question $w \in f(x)$ can be decided in time $n^{\mathcal{O}(1)} \cdot t(|w|)$  with $n = |w|+ x$.
    \label{lem:effective}
\end{lemma}
\begin{proof}
    Let $p \colon \N \rightarrow \N^2$ be a surjective function such that $y,z \leq x$ for all $x$ in the domain of $p$ and given $x$ in unary $p(x)$ can be computed in polynomial time. For example, $p(x) = (y,z) \Leftrightarrow x = 2^y3^z$. Then the desired mapping $f \colon \N \rightarrow \ComplexityFont{ALL}$ can be constructed from $p$ as follows. If $p(x) = (y,z)$ then $f(x)$ is the language that is decided by the Turing machine $M_y$ when running at most $z\cdot t(|w|)$ steps on input $w$. If $p(x)$ is undefined then $f(x)$ shall be the empty language. Clearly, $\Ima(f) = \TIME(t(n))$ because a language $L$ is in $\TIME(t(n))$ iff there exists a Turing machine $M$ and a $c \in \N$ such that $M$ decides $L$ and runs in time $c\cdot t(|w|)$. For the second part we construct a universal Turing machine with the required time bound. On input $x \in \N$ in unary and $w \in \{0,1\}^*$ it computes $p(x) = (y,z)$.  If $p(x)$ is undefined it rejects (this corresponds to recognizing the empty language). Otherwise, it simulates $M_y$ on input word $w$ for $z \cdot t(|w|)$ steps. Due to the fact that $t$ is time-constructible it is possible to run a counter during the simulation of $M_y$ in order to not exceed the $z \cdot t(|w|)$ steps. The input length is $n := x + |w|$. The simulation can be run in time $n^{\mathcal{O}(1)} \cdot z \cdot t(|w|)$. Since $z \leq x \leq n$ the desired time bound follows.
\end{proof}

\begin{lemma}
    For every time-constructible function $t \colon \N \rightarrow \N$ it holds that the label decoder $F_{\TIME(t(n))}$ can be computed in $\TIME(\exp^2(n)\cdot t(n))$.    
    \label{lem:fa_compute}
\end{lemma}
\begin{proof}
    On input $xy$ with $x,y \in \{0,1\}^m$ and $m \geq 1$ compute $\tau(2^m) = (y,z)$. If it is undefined then reject. Otherwise there is a labeling scheme $S_{\tau(2^m)}=(F_y,z)$ and we need to compute the smallest graph $G_0$ on $2^m$ vertices such that $G_0 \notin \gr{S_{\tau(2^m)}}$. If $G_0$ exists we assume that its vertex set is $\{0,1\}^m$ and accept iff $(x,y) \in E(G_0)$. If it does not exist then reject. 
    
    The graph $G_0$ can be computed as follows. Iterate over all labeled graphs $H$ with $2^m$ vertices in order and over all functions $\ell \colon V(H) \rightarrow \{0,1\}^{zm}$. Verify if $H \in \gr{S_{\tau(2^m)}}$ by checking if $(u,v) \in E(H) \Leftrightarrow  (\ell(u),\ell(v)) \in F_y$ holds for all $u, v \in V(H)$. If this condition fails for all labelings $\ell$ then $G_0 = H$. To query the label decoder $F_y$ we apply the previous Lemma \ref{lem:effective}.  
    
    Let us consider the time requirement w.r.t.~$m$. To compute $\tau(2^m)$ we write down $2^m$ in unary and compute $\tau$ in polynomial time w.r.t.~$2^m$ which is in the order $2^{\mathcal{O}(m)}$. To compute $G_0$ there are four nested loops. The first one goes over all labeled graphs on $2^m$ vertices which is bounded by $\exp(\exp(m)^2) = \exp^2(2m)$. The second loop considers all possible labelings $\ell$ of which there can be at most $\exp(zm)^{\exp(m)} =\exp(\exp(m)zm)  \leq \exp^2(zm^2)$. It holds that $y,z \leq \log (2^m) = m$ due to the second condition of Definition~\ref{def:pairing} and therefore $\exp^2(zm^2) \leq \exp^2(m^3)$. The other two loops go over all vertices of $H$ of which there are $2^m$. Due to Lemma~\ref{lem:effective} the time required to compute $(\ell(u),\ell(v)) \in F_y$ is $y^{\mathcal{O}(1)}\cdot t(2m) = m^{\mathcal{O}(1)}\cdot t(2m)$. In total this means the algorithm runs in time $\mathcal{O}(\exp^2(m^{\mathcal{O}(1)})\cdot t(2m) )$ which is the required time bound since the input length is $2m$.
\end{proof}

Lemma~\ref{lem:ca_not_in_a} states that $\mathcal{C}_{\TIME(t(n))} \notin \ccg \TIME(t(n))$ and from Lemma~\ref{lem:fa_compute} it follows that $\mathcal{C}_{\TIME(t(n))} \in \ccg \TIME(\exp^2(n) \cdot t(n))$ therefore proving Theorem \ref{thm:hierarchy}. Notice, this argument fails to separate  $\ccg \P$ from $\ccg \EXP $ because the runtime to compute the label decoder $F_\P$ is at least double exponential due to the first two loops mentioned in the proof of Lemma \ref{lem:fa_compute}. 
Also, we remark that if the hereditary closure of the diagonalization graph class of $\P$ is small then the implicit graph conjecture is false. 

\section{Expressiveness of Primitive Labeling Schemes}
\label{sec:primitive}
To understand the limitations of labeling schemes it is reasonable to start with very simple ones first and then gradually increase the complexity. In this section we present two such simple families of labeling schemes and explain how they relate to other well-known sets of graph classes.  These two families of labeling schemes can be seen as generalizations of uniformly sparse graph classes and graph classes with bounded degree. Therefore they represent many graph classes that are of theoretical and practical importance. For this section we assume all graphs and graph classes to be undirected. This also means that for sets of graph classes such as $\gccac$ we only consider its restriction to undirected graph classes.  

In \cite[p.~20]{spinrad} a labeling scheme for every sparse and hereditary graph class is described. 
A graph class $\mathcal{C}$ is sparse and hereditary iff it has bounded degeneracy. Therefore there exists a constant $c$ such that every graph $G \in \mathcal{C}$ has a vertex  with degree at most $c$. The following labeling scheme represents $\mathcal{C}$. Given a graph $G$ from $\mathcal{C}$ assign each vertex a unique identifier $1,\dots,n$. Choose a vertex $v$ in $G$ with at most $c$ neighbors. Store the identifier of $v$ along with the identifiers of its $c$ neighbors in the label of $v$. Delete the vertex $v$ from $G$ and repeat this process. Since $\mathcal{C}$ is hereditary it follows that $G$ without $v$ also has a vertex of degree at most $c$. Two vertices $u,v$ with labels $u_0,u_1,\dots,u_c$ and $v_0,v_1,\dots,v_c$ are adjacent iff $u_0 \in \{v_1,\dots,v_c\}$ or $v_0 \in \{u_1,\dots,u_c\}$. For every $c \in \N$ this construction yields a labeling scheme $S_c$. Let us call the smallest number $c$ such that a graph can be represented by $S_c$ its or-pointer number. This can be further generalized and leads to the following four graph parameters. 

\begin{definition}[Pointer Numbers]
    The (bijective) and/or-pointer number of a graph $G$ with $n$ vertices is the least $k \in \N$ such that there exist a (bijective) function $\ell_{\mathrm{id}} \colon V(G) \rightarrow [n]$ and a function $\ell \colon V(G) \rightarrow [n]^k$ for which it holds that $\{u,v\} \in E(G)$ iff $\ell_{\mathrm{id}}(u) \in \ell(v)$ and/or $\ell_{\mathrm{id}}(v) \in \ell(u)$ for all $u \neq v \in V(G)$.
\end{definition}

The bijectiveness constraint can be understood as restriction on the possible labelings that are allowed, i.e.~the id  of each vertex must be unique. In the bijective case the function $\ell_{\mathrm{id}} \colon V(G) \rightarrow [n]$ becomes obsolete and the function $\ell$ can be understood as a mapping from $V(G)$ to $V(G)^k$.
Notice how this constraint is satisfied in the case of the labeling described in the previous paragraph. 

\begin{fact}
    The bijective or-pointer number and degeneracy are equivalent. 
    \label{fact:bopnus}
\end{fact}
\begin{proof}
    ``$\Rightarrow$'': Let $G$ have bijective or-pointer number at most $k$. It holds that $G$ has degeneracy at most $2k$. Every induced subgraph of $G$ has bijective or-pointer number at most $k$. Additionally, every graph with bijective or-pointer number at most $k$ can have at most $kn$ edges which implies that such a graph must have a vertex with degree at most $2k$.
    
    ``$\Leftarrow$'': Let $G$ have degeneracy at most $k$. Then $G$ has bijective or-pointer number at most $k$ due to the labeling described in the first paragraph of this subsection.
\end{proof}

\begin{fact}
    The bijective and-pointer number of a graph equals its maximum degree. 
\end{fact}

\begin{fact}
    Planar graphs have unbounded and-pointer number. 
    \label{fact:planarapn}
\end{fact}
\begin{proof}
    Consider the graph $G_k$ shown in Figure~\ref{fig:planar_apn}. We show that for every $l \in \N$ there exists a $k \in \N$ such that the and-pointer number of $G_k$ is larger than $l$. For a given $l$ let $k = l^2+1$. For the sake of contradiction, assume that $G_k$ has and-pointer number $l$ via $(\ell_{\mathrm{id}},\ell)$. This means $\ell(\cdot)$ has at most $l$ elements.
    For two vertices $u,v$ let us say that they are equivalent if $\ell_{\mathrm{id}}(u) =\ell_{\mathrm{id}}(v)$. Since $x$ is adjacent to $x_i$ it holds that $\ell_{\mathrm{id}}(x_i) \in \ell(x)$ for all $i \in [k]$. Therefore $\{x_1,\dots,x_k\}$ consists of at most $l$ equivalence classes. The same holds for $\{y_1,\dots,y_k\}$. 
    Due to the pigeonhole principle it follows that there are $r := \lceil \frac{k}{l} \rceil $ vertices $a_1,\dots,a_r \in \{x_1,\dots,x_k\}$ that are equivalent. For $a_i$ let $b_i$ denote the vertex in $\{y_1,\dots,y_k\}$ that is adjacent to $a_i$. For all $i\neq j \in [r]$ it holds that $b_i$ and $b_j$ are not equivalent, otherwise $G_k$ would contain $K_{2,2}$. This implies that $\{y_1,\dots,y_k\}$ must consist of at least $r$ different equivalence classes. However, $\{y_1,\dots,y_k\}$ consists of at most $l$ different classes and $r = \lceil \frac{k}{l} \rceil = \lceil \frac{l^2+1}{l} \rceil = l+1$. Therefore the and-pointer number of $G_k$ must be larger than $l$.   
\end{proof}

We remark that every forest has and-pointer number at most two. See Figure~\ref{fig:tree_apn} for an example of how to label a tree; the number left of the bar is the id of the vertex. 

In comparison, it seems not quite as simple to prove that a small graph class has unbounded or-pointer number. The following observation might be helpful in that regard. Consider a graph $G$ with a  $k$-or-pointer representation $(\ell_{\mathrm{id}},\ell)$. Let $c$ denote the number of unique ids, i.e.~the cardinality of the image of $\ell_{\mathrm{id}}$. There exists an induced subgraph of $G$ with $c$ vertices which has bijective or-pointer number at most $k$ and therefore this subgraph can have at most $kc$ edges. Informally, if a graph has many edges then in any $k$-or-pointer representation of this graph there cannot be many unique ids. As a consequence the structure of such a graph is quite constrained.
Therefore it seems reasonable to suspect that the edge-complement of some sparse graph class such as planar graphs has unbounded or-pointer number. Also, does every graph class with bounded and-pointer number have bounded or-pointer number?

\begin{figure}
    \centering
    \begin{minipage}[b]{.5\textwidth}
        \centering
        \begin{tikzpicture}[shorten >=1pt,auto,node distance=1.2cm,
  main node/.style={circle,draw}]

%

%
%

\newcommand*{\xoff}{0}%
\newcommand*{\yoff}{0}%

\newcommand*{\movex}{1}%
\newcommand*{\movey}{0.8}%

\node[main node,inner sep=1.25mm] (x) at (\xoff,\yoff) {$x$};
\node[main node,inner sep=0.8mm] (x1) at (\xoff+\movex,\yoff+1.5*\movey) {$x_1$};
\node[main node,inner sep=0.8mm] (x2) at (\xoff+\movex,\yoff+0.5*\movey) {$x_2$};
\node[] (x3) at (\xoff+\movex,\yoff-0.4*\movey) {$\vdots$};
\node[main node,inner sep=0.8mm] (x4) at (\xoff+\movex,\yoff-1.5*\movey) {$x_k$};

\node[main node,inner sep=0.8mm] (y1) at (\xoff+2.5*\movex,\yoff+1.5*\movey) {$y_1$};
\node[main node,inner sep=0.8mm] (y2) at (\xoff+2.5*\movex,\yoff+0.5*\movey) {$y_2$};
\node[] (y3) at (\xoff+2.5*\movex,\yoff-0.4*\movey) {$\vdots$};
\node[main node,inner sep=0.8mm] (y4) at (\xoff+2.5*\movex,\yoff-1.5*\movey) {$y_k$};
\node[main node,inner sep=1.25mm] (y) at (\xoff+3.5*\movex,\yoff) {$y$};
%

\path[-]
(x) edge (x1)
(x) edge (x2)
(x) edge (x4)

(y) edge (y1)
(y) edge (y2)
(y) edge (y4)

(x1) edge (y1)
(x2) edge (y2)
(x4) edge (y4)
;


\end{tikzpicture}
        \captionof{figure}{Unbounded and-pointer number $G_k$}
        \label{fig:planar_apn}
    \end{minipage}%
    \begin{minipage}[b]{.5\textwidth}
        \centering
        \begin{tikzpicture}[shorten >=1pt,auto,node distance=1.2cm,
  main node/.style={circle,draw}]

\newcommand*{\xoff}{0}%
\newcommand*{\yoff}{0}%

\newcommand*{\movex}{0.7}%
\newcommand*{\movey}{-1}%

\node[main node,inner sep=0.1mm,minimum size=0.8cm] (r) at (\xoff,\yoff) {$1|2$};

\node[main node,inner sep=0.1mm,minimum size=0.8cm] (b) at (\xoff,\yoff+\movey) {$2|1$};

\node[main node,inner sep=0.1mm,minimum size=0.8cm] (a) at (\xoff-2*\movex,\yoff+\movey) {$2|1,3$};
\node[main node,inner sep=0.1mm,minimum size=0.8cm] (a1) at (\xoff-2.8*\movex,\yoff+2*\movey) {$3|2$};
\node[main node,inner sep=0.1mm,minimum size=0.8cm] (a2) at (\xoff-1.3*\movex,\yoff+2*\movey) {$3|2$};

\node[main node,inner sep=0.1mm,minimum size=0.8cm] (c) at (\xoff+2*\movex,\yoff+\movey) {$2|1,4$};
\node[main node,inner sep=0.1mm,minimum size=0.8cm] (c1) at (\xoff+2*\movex,\yoff+2*\movey) {$4|2$};

\path[-]
(r) edge (a)
(r) edge (b)
(r) edge (c)

(a) edge (a1)
(a) edge (a2)
(c) edge (c1)
;

%
%
%
%


\end{tikzpicture}
        \captionof{figure}{And-pointer labeling of a tree}
        \label{fig:tree_apn}
    \end{minipage}
\end{figure}

In Figure \ref{fig:sparse_diagram} we give an overview of the relation of $\gccac$ and the pointer numbers to other sets of graph classes defined in terms of graph class properties and graph parameters. A graph parameter in this figure is interpreted as the set of graph classes that are bounded by it. For example, a well-known result by Robertson and Seymour states that a graph class $\mathcal{C}$ has bounded tree-width iff there exists a planar graph $G$ such that no graph in $\mathcal{C}$ has $G$ as minor, i.e.~$\mathcal{C} \in \minorfree(\gcplanar)$. The remainder of this section explains the relations of the classes shown in Figure \ref{fig:sparse_diagram} going from top to bottom.

\begin{figure}
    \centering
    \begin{tikzpicture}[shorten >=1pt,auto,node distance=1.2cm,
main node/.style={circle,draw}]
\usetikzlibrary{shapes.misc}

\newcommand*{\movey}{1.4}%
\newcommand*{\movex}{2}%
\newcommand*{\movexsi}{0.02}%

\node (gp) at (-1*\movex,0) {$\gccp$};
\node (gac) at (-1*\movex,-1*\movey) {$\gccac$};

\node[align=center] (gsh) at (\movex,-0.5*\movey) {{ $\cogcsmallh$}};

\node (gacsh) at (0*\movex,-2*\movey) {{ \hspace{2em} $\gccac \cap \cogcsmallh$ } };

\node (gopn) at (0*\movex,-3*\movey) {{ Or-Pointer Number}};
\node (gapn) at (2*\movex,-3*\movey) {{ And-Pointer Number}};
\node (gcw) at (-2.7*\movex,-3*\movey) {{ Clique-Width }};

\node[align=center] (gus) at (-1*\movex,-4.5*\movey) {
    $\cogcsh$\\
    Degeneracy\\
    Bij.~Or-Pointer No.    
};

\node[align=center] (gpmc) at (-1*\movex,-6*\movey) {
   $\cogcpmc$\\
    $\minorfree(\setofgraphs)$
};

\node[align=center] (gtw) at (-1*\movex,-7*\movey) {
    $\minorfree(\gcplanar)$\\
    Tree-Width
};


\node[align=center] (gmd) at (1*\movex,-6*\movey+0.6*\movey) {
    Max.~Degree\\
    Bij.~And-Pointer No.
};

\node[align=center] (gti) at (2.9*\movex,-6*\movey+0.6*\movey) {
    {$\cogcth$}\\
    Twin Index
};

\node[align=center] (gecc) at (2.9*\movex,-7*\movey+0.6*\movey) {
    Intersection Number
};

\path[-]
(gac) edge (gp)
(gacsh) edge (gac)
(gacsh) edge (gsh)
(gopn) edge (gacsh)
(gapn) edge (gacsh)
(gcw) edge (gacsh)

(gus) edge (gopn)
(gpmc) edge (gus)
(gtw) edge (gpmc)

(gmd) edge (gus)
(gmd) edge (gapn)

(gti) edge (gapn)
(gti) edge (gopn)

(gecc) edge (gti)

(gtw) edge[bend left=45,bend angle=20] (gcw);

\node (sitw) at (-1*\movex,-7*\movey+0.43*\movey) {\strictinclusion};
\node (sipmc) at (-1*\movex,-6*\movey+0.68*\movey)  {\strictinclusion};
\node (sish) at (-1*\movex+0.6*\movex,-4*\movey+0.38*\movey) {\strictinclusiond};

\node (simd) at (1*\movex-0.9*\movex,-6*\movey+1.02*\movey) {\strictinclusion};

\node (simdapn) at (1*\movex+0.3*\movex,-6*\movey+1.2*\movey+0.15) {\strictinclusiond};
\node (siecc) at (2.9*\movex,-7*\movey+1.05*\movey) {\strictinclusion};
\node (sithopn) at (1*\movex+1.2*\movex,-6*\movey+1.2*\movey) {\strictinclusion};
\node (sithapn) at (1*\movex+1.64*\movex,-6*\movey+1.2*\movey+0.19) {\strictinclusion};

\node (sitwcw) at (-2.43*\movex,-5.8*\movey) {\strictinclusion};
\node (sicw) at (-1.35*\movex,-2.5*\movey) {\strictinclusiond};

\node (siapngac) at (0.74*\movex+0.54,-2.5*\movey) {\strictinclusion};

\node (siac) at (-0.55*\movex,-1.44*\movey) {\strictinclusion};

\end{tikzpicture}
    \caption{Various sets of graph classes and their relation to labeling schemes}
    \label{fig:sparse_diagram}                
\end{figure}
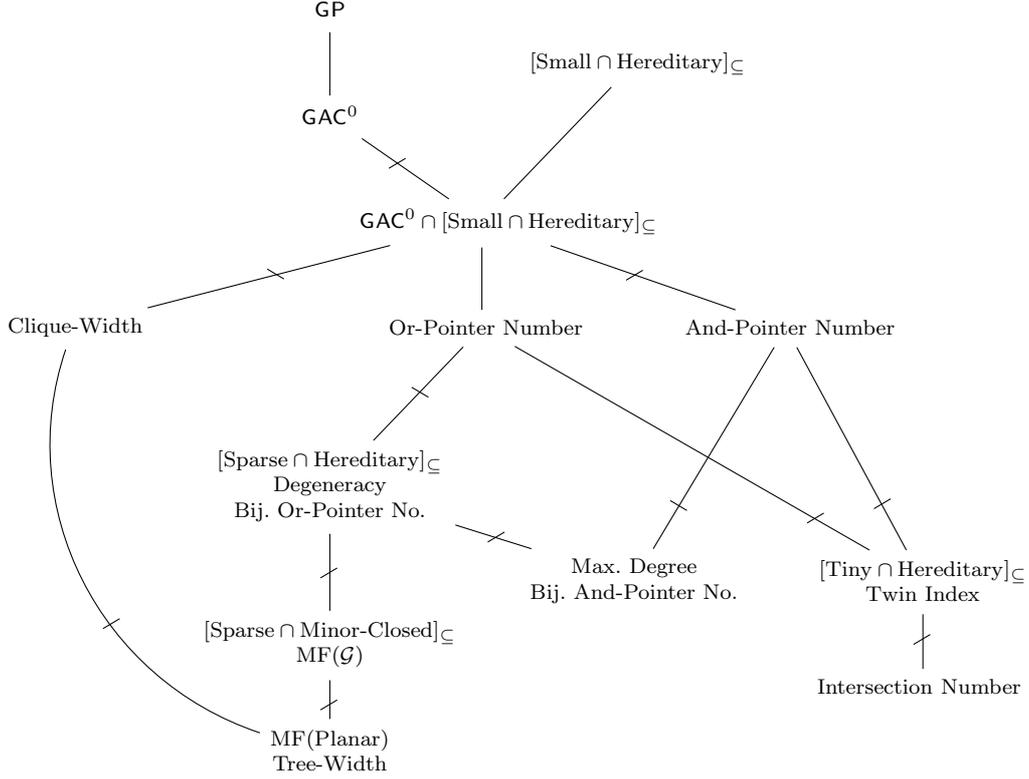  

The question of whether $\cogcsmallh$ is a subset of $\gccp$ is equivalent to the implicit graph conjecture. Therefore it is also unknown whether $\cogcsmallh$ is a subset of $\gccac$ since this would imply the implicit graph conjecture. However, it can be shown that the converse does not hold, i.e.~$\gccac$ is not a subset of $\cogcsmallh$. From a graph-theoretical point of view being hereditary is the weakest form of uniformity condition that can be imposed on a graph class in order for it to have some meaningful structure. Therefore  there probably is no elegant graph-theoretical characterization of $\gccac$. This might also be one of the reasons why it is so difficult to analyze $\gccac$ and its supersets. 

\begin{theorem}
    $\gccac \not\subseteq  \cogcsmallh$.
    \label{thm:acsmallh}
\end{theorem}
\begin{proof}
    For a graph class $\mathcal{C}$ recall that $[\mathcal{C}]_{\subseteq}$ denote its closure under induced subgraphs. If a graph class $\mathcal{C}$ is in $\cogcsmallh$ then $[\mathcal{C}]_{\subseteq}$ is in $\cogcsmallh$ as well. 
    To prove the above statement we show that (I) $\gccr \not\subseteq  \cogcsmallh$ and (II) for every graph class $\mathcal{C} \in \gccr$ there exists a graph class $\mathcal{D}$ in $\gccac$ with $[\mathcal{C}]_{\subseteq} \subseteq [\mathcal{D}]_{\subseteq}$.
    
    For (I) we construct a labeling scheme $S=(F,1)$ where $F$ is computable and $[\gr{S}]_{\subseteq}$ is the class of all graphs. Let $f \colon \N \rightarrow \setofgraphs$ be a computable function such that: 
    \begin{enumerate}
        \item the image of $f$ equals the set of all graphs $\setofgraphs$
        \item $|V(f(n))| \leq n$ for all $n \in \N$
        \item if $n$ is not a power of two then $f(n)$ is the graph with a single vertex
    \end{enumerate}
    For $n \geq 2$ let $G_n$ be the lexicographically smallest labeled graph that is isomorphic to $f(n)$ and $V(G_n) \subseteq \{0,1\}^{\log n}$. We define the label decoder $F$ as
    $$(x,y) \in F \Leftrightarrow (x,y) \in E(G_{2^m})$$
    for all $x,y \in \{0,1\}^m$ and $m \in \N$. For every graph $G$ with at least two vertices there exists an $m \in \N$ such that $G$ is isomorphic to $f(2^m)$ and therefore $G$ occurs as induced subgraph of the graph with $2^m$ vertices in $\gr{S}$.  
    
    We show that (II) holds due to a padding argument. Given a labeling scheme $S=(F,c)$ in $\gccr$ and a function $p \colon \N \rightarrow \N$ we define the labeling scheme $S_p=(F_p,c)$ as
    $$ (x,y) \in F \Leftrightarrow (xx',yy') \in F_p$$
    for all $m \in \N$, $x,y \in \{0,1\}^{cm}$, $x',y' \in \{0,1\}^{cp(m)}$. To see that $[\gr{S}]_{\subseteq} \subseteq [\gr{S_p}]_{\subseteq}$ consider a graph $G$ that is in $\gr{S}$ via a labeling $\ell \colon V(G) \rightarrow \{0,1\}^{c \log n}$. We show that there is a graph $G_0$ in $\gr{S_p}$ on $n_0 = 2^{\log n + p(\log n)}$ vertices such that $G$ is an induced subgraph of $G_0$. Let us assume that $V(G) \subseteq V(G_0)$.
    The partial labeling $\ell' \colon V(G_0) \rightarrow \{0,1\}^{c (\log n + p(\log n))}$ with $\ell'(u) = \ell(u) 0^{c p(\log n)}$ for all $u \in V(G)$ shows that $G$ is an induced subgraph of $G_0$. It remains to argue that for every computable label decoder $F$ there exists a padding function $p$ such that $F_p$ is computable by a family of logspace-uniform $\AC^0$-circuits. Let $t \colon \N \rightarrow \N$ be the runtime of a Turing machine which computes $F$. The idea is to choose $p$ sufficiently large in terms of $t$ such that the logspace transducer which computes the circuit family can precompute the satisfying assignments for the circuit and then compile them into a DNF. Notice that membership in $F_p$ only depends on a small part of the input bits and therefore the DNF is only polynomial w.r.t.~the input size and thus can be directly encoded into the circuit.
\end{proof}

The or- and and-pointer numbers are hereditary because deleting vertices does not increase them. Also, the labeling schemes behind these numbers can be computed in $\gccac$ and therefore the pointer numbers are contained in the intersection of these two classes. The and-pointer number is strictly contained in the intersection because planar graphs have unbounded and-pointer number but bounded or-pointer number.

\begin{fact}[{\cite[p.~165 f.]{spinrad}}]
	Every graph class with bounded clique-width is in $\gccac$. 
	\label{fact:cwgac}
\end{fact}
\begin{proof}
    A subset of vertices $S$ of a graph $G$ is called a $k$-module if it can be partitioned into at most $k$ parts $S_1,\dots,S_k$ such that $S_i$ is a module in $G[( V \setminus S)\cup S_i]$, i.e.~the vertices in $S_i$ are indistinguishable to vertices of $V \setminus S$ for $i \in [k]$. A $k$-module $S$ of $G$ is called balanced if $S$ contains at least one third and at most two thirds of the vertices of $G$. 
    Spinrad asserts that every graph with at least three vertices and clique-width $k$ has a balanced $k$-module. 
    Therefore given a graph $G$ with clique-width $k$ one can construct a binary tree $T(G)$ by recursively finding a balanced $k$-module $S$ and putting the vertices of $S$ in the left node and the vertices of $V(G) \setminus S$ in the right node. The root node contains every vertex of $G$. Since every child node in $T(G)$ has at most two thirds of the vertices of its parent node it follows that the tree has depth $\mathcal{O}(\log n)$. This tree can be used to construct the following labeling scheme. Given a graph $G$ with clique width $k$ a vertex $v$ of $G$ is labeled as follows. There is a path in $T(G)$ from the root node to the leaf node which contains $v$. Let this path be $x_1,\dots,x_c$ where $x_1$ is the root node and $x_c$ is the leaf. For each $1 \leq i < c$ the following information is stored in the label of $v$. For $x_i$ let $S^i$ be the balanced $k$-module stored in the left child node of $x_i$ which can be partitioned into $S^i_1,\dots,S^i_k$. The first bit of $v$ for the $i$-th level denotes whether $v$ is in the left child of $x_i$. If $v$ is in the left child of $x_i$ then this means $v$ is in the balanced $k$-module $S^i$ and one also stores the index $j$ such that $v \in S^i_j$ for $1 \leq j \leq k$. If $v$ is in the right child of $x_i$ then one stores the subset of $X$ of $\{1,\dots,k\}$ such that $v$ is adjacent to the modules $S^i_j$ for all $j \in X$. Each level only requires a constant number of bits. To check whether two vertices $u,v$ of $G$ are adjacent one has to find the first level of $T(G)$ such that $u$ and $v$ are placed in different nodes. Assume that $u$ is in the left subtree and $v$ in the right one. Then $u$ and $v$ are adjacent iff the index $j$ of the part of the balanced $k$-module that $u$ is contained in is part of the subset $X$ of $v$ for this level. It is not difficult to construct a constant-depth circuit which computes this label decoder. 
\end{proof}

Every uniformly sparse graph class has bounded degeneracy and thus bounded or-pointer number due to Fact \ref{fact:bopnus}. Since the family of complete graphs has bounded or-pointer number but is not (uniformly) sparse this inclusion is strict. 

\begin{fact}
    Every graph class with bounded twin index has bounded or- and and-pointer number.  
\end{fact}
\begin{proof}
    Let $\mathcal{C}$ be a graph class with bounded twin index $k$. 
    This means a graph $G$ in $\mathcal{C}$ has at most $k$ twin classes $V_1,\dots,V_k$. 
    The following labeling of the vertices in $G$ shows that $G$ has or- and and-pointer number at most $k$. Given a vertex $v$ in $G$ let its id be the index of the twin class, i.e.~$\ell_{\mathrm{id}}(v) = i$ such that $v \in V_i$. Then let $\ell(v)$ be the subset of $[k]$ such that $j \in \ell(v)$ iff the twin class $V_j$ is adjacent to the twin class of $v$ or $j=\ell_{\mathrm{id}}(v)$ and $V_j$ is a clique. 
\end{proof}

The class of square grid graphs has unbounded twin index but bounded degree. Therefore the twin index is strictly contained in the pointer numbers.

\begin{fact}
    A graph class is in $\cogcth$ iff it has bounded twin index.
    \label{fact:thfu}
\end{fact}
\begin{proof}
    ``$\Rightarrow$'': This direction is proved in \cite{scheinerman}.

    ``$\Leftarrow$'':  Let $\mathcal{C}$ be the set of graphs with twin index at most $k$ for a $k \in \N$. Clearly, $\mathcal{C}$ is hereditary. It remains to show that $\mathcal{C}$ is tiny. Recall that a graph class $\mathcal{C}$ is tiny if there exists a $c < \frac{1}{2}$ such that $|\mathcal{C}_n| \leq n^{cn}$ for all sufficiently large $n$.  A graph $G$ on $n$ vertices in $\mathcal{C}$ with twin index $1 \leq i \leq k$ is determined by the following choices. Choose an unordered partition of $n$ into $i$ parts $p_1, p_2,  \dots,  p_i \in [n]$, which means $p_1 + \dots + p_i =n$ and $p_1 \leq p_2 \leq \dots \leq p_i$. Let $P_{n,i}$ denote the number of such partitions. This partition tells us that the first $p_1$ vertices of $G$ are a twin class, the next $p_2$ vertices of $G$ are a twin class and so on. For every twin class one has to choose whether it is a clique or an independent set ($2^i$ possibilities). It remains to choose how the $i$ twin classes interact ($|\setofgraphs_i| \leq 2^{i^2}$ possibilities).      
    $$ \sum_{i=1}^k  2^{i^2 + i} \cdot  P_{n,i}     
    \leq \sum_{i=1}^k  2^{i^2+i} \cdot 2^n \leq k \cdot  2^{k^2+k} \cdot 2^n \leq k \cdot 2^{n +k^2} \leq 2^{\frac{1}{3} n \log n} = n^{\frac{1}{3}n}  $$    
    These inequalities are meant to hold for fixed $k$ and sufficiently large $n$.  
    The first inequality holds due to the fact that the partition number $P_n = \sum_{i=1}^n P_{n,i}$ is in $\mathcal{O}(2^n)$, which follows from Hardy and Ramanujan's asymptotic formula for $P_n$.
\end{proof}

It is interesting to note that  a set of graph classes defined in terms of graph-theoretical properties (tiny and hereditary) can be characterized in terms of a certain class of labeling schemes (labeling schemes with constant label length \cite{veryshort}).

To see that the intersection number is bounded by the twin index observe that a graph with intersection number $k$ can have at most $2^k$ different twin classes because the twin class of a vertex is determined by a subset of $k$ elements. To see that this inclusion is strict consider the class of complete bipartite graphs which have twin index two but unbounded intersection number. 

\section{Reductions Between Graph Classes}
\label{sec:reductions}
The concept of reduction is vital to complexity theory as it enables one to formally compare the complexity of problems as opposed to just treating each problem separately. In our context we want to say a graph class $\mathcal{C}$ reduces to a graph class $\mathcal{D}$ if the adjacency of graphs in $\mathcal{C}$ can be expressed in terms of adjacency of graphs in $\mathcal{D}$. It should satisfy the following closure property: if $\mathcal{D}$ has a labeling scheme and $\mathcal{C}$ is reducible to $\mathcal{D}$ then $\mathcal{C}$ has a labeling scheme as well. This closure should also hold when the complexity of label decoders is restricted, i.e.~if $\mathcal{D}$ is in $\gccp$ and $\mathcal{C}$ reduces to $\mathcal{D}$ then $\mathcal{C}$ is in $\gccp$. Such a reduction notion makes it possible to compare graph classes for which no labeling schemes are known to see whether there might be a common obstacle to designing a labeling scheme. We introduce two kinds of reduction that satisfy this closure property.
Before we formally define and examine them let us first explain the intuition behind them.

For the first reduction type, called algebraic reduction, the idea is to express the adjacency of a graph $G$ on vertex set $V$ in terms of graphs $H_1,\dots,H_k$ which also have vertex set $V$ and a $k$-ary boolean function $f$. Two vertices $u,v$ in $G$ are adjacent iff $f(x_1,\dots,x_k) = 1$ where $x_i$ denotes whether $u$ and $v$ are adjacent in $H_i$. For example, every axis-aligned rectangle intersection graph can be expressed as the conjunction of two interval graphs as shown in Figure \ref{fig:intvbox}, i.e.~two boxes intersect iff both of their corresponding intervals intersect.
Therefore we say axis-aligned rectangle intersection graphs reduce to interval graphs.
We call this type of reduction algebraic because it is build upon an interpretation of boolean functions as functions on graph classes. The resulting algebra on graph classes inherits some of the properties of its boolean ancestor.
For instance, the negation of a graph class $\mathcal{C}$ equals its edge-complement co-$\mathcal{C}$ in this interpretation. Thus negating a graph class twice is an involution, i.e.~$\neg \neg \mathcal{C} = \mathcal{C}$. Additionally, this algebra gives a unifying terminology for concepts such as arboricity, thickness and boxicity.  We remark that a similar but in a sense less general notion called locally bounded coverings of graphs has been used in \cite{lozin} and \cite{atminas}. They say a set of graphs $H_1,\dots,H_k$ is a covering of a graph $G$ if $V(G) = \cup_{i=1}^k V(H_i)$ and $E(G) = \cup_{i=1}^k E(H_i)$. If we assume that $V(H_i) = V(G)$ then in our terminology  $G$ is the disjunction of $H_1,\dots,H_k$. They also use this as a tool to prove the existence of implicit representations \cite[Lem.~4]{atminas}.     

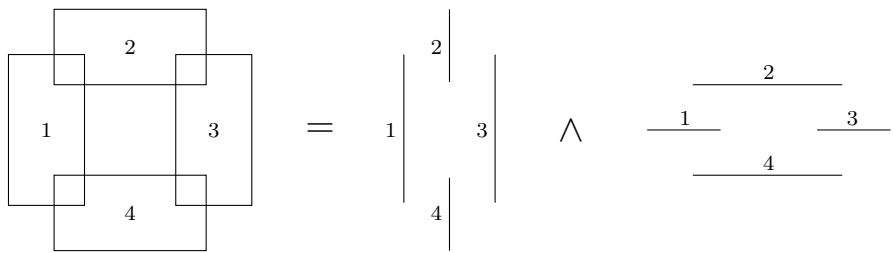
\begin{figure}[b]
    \centering
    \begin{tikzpicture}[shorten >=1pt,auto,node distance=1.2cm,
main node/.style={circle,draw}]
\usetikzlibrary{shapes.misc}

\newcommand*{\xa}{0}%
\newcommand*{\ya}{0}%

\newcommand*{\xb}{4.2}%
\newcommand*{\yb}{0}%

\newcommand*{\xc}{8.4}%
\newcommand*{\yc}{0}%

\newcommand*{\ibo}{0.6}%

\newcommand*{\boxlong}{2}%
\newcommand*{\boxshort}{1}%

\draw (\xa-\ibo-\boxshort,\ya+0.5*\boxlong) rectangle (\xa-\ibo,\ya-0.5*\boxlong);
\draw (\xa+\ibo,\ya+0.5*\boxlong) rectangle (\xa+\ibo+\boxshort,\ya-0.5*\boxlong);

\draw (\xa-0.5*\boxlong,\ya+\ibo+\boxshort) rectangle (\xa+0.5*\boxlong,\ya+\ibo);
\draw (\xa-0.5*\boxlong,\ya-\ibo) rectangle (\xa+0.5*\boxlong,\ya-\ibo-\boxshort);

\node (a) at (\xa-\ibo-0.5*\boxshort,\ya) {1};
\node (c) at (\xa+\ibo+0.5*\boxshort,\ya) {3};
\node (b) at (\xa,\ya+\ibo+0.5*\boxshort) {2};
\node (d) at (\xa,\ya-\ibo-0.5*\boxshort) {4};

\node (eqn) at (\xa+2.5,\ya) {\Large =};

\draw (\xb-\ibo,\yb+0.5*\boxlong) -- (\xb-\ibo,\yb-0.5*\boxlong);
\draw (\xb+\ibo,\yb+0.5*\boxlong) -- (\xb+\ibo,\yb-0.5*\boxlong);
\draw (\xb,\yb+\ibo+\boxshort) -- (\xb,\yb+\ibo);
\draw (\xb,\yb-\ibo-\boxshort) -- (\xb,\yb-\ibo);
\node (aa) at (\xb-\ibo-0.17,\yb) {1};
\node (ca) at (\xb+\ibo-0.17,\yb) {3};
\node (ba) at (\xb-0.17,\yb+\ibo+0.5*\boxshort) {2};
\node (da) at (\xb-0.17,\yb-\ibo-0.5*\boxshort) {4};

\node (eqn) at (\xb+1.6,\ya) {\Large $\wedge$};

\draw (\xc-\ibo-\boxshort,\yc) -- (\xc-\ibo,\yc);
\draw (\xc+\ibo+\boxshort,\yc) -- (\xc+\ibo,\yc);
\draw (\xc-0.5*\boxlong,\yc+\ibo) -- (\xc+0.5*\boxlong,\yc+\ibo);
\draw (\xc-0.5*\boxlong,\yc-\ibo) -- (\xc+0.5*\boxlong,\yc-\ibo);

\node (ab) at (\xc-\ibo-0.5*\boxshort,\yc+0.16) {1};
\node (ab) at (\xc+\ibo+0.5*\boxshort,\yc+0.16) {3};
\node (ab) at (\xc,\yc+\ibo+0.16) {2};
\node (ab) at (\xc,\yc-\ibo+0.16) {4};


\end{tikzpicture}
    \caption{Axis-aligned rectangle intersection graph as conjunction of two interval graphs}
    \label{fig:intvbox}                
\end{figure} 

The second kind of reduction is called subgraph reduction. While it is technically more tedious to define than the algebraic variant, the underlying intuition is just as simple. Instead of expressing the adjacency of a graph $G$ using a sequence of graphs as before, we do this in terms of a single, larger graph $H$. Informally, every vertex of $G$ is assigned to a constant-sized subgraph of $H$ and two vertices $u,v$ of $G$ are adjacent if their combined (labeled) subgraph in $H$ satisfies some condition. To illustrate this let us consider how the adjacency of $k$-interval graphs can be expressed in terms of interval graphs in this sense. Let $G$ be a $k$-interval graph with $n$ vertices. This means there exists an interval model $M(G)$ of $G$ with $kn$ intervals. Let $H$ be the interval graph with $kn$ vertices that is induced by the intervals of $M(G)$. Then assign each vertex $u$ of $G$ the $k$ vertices  $u_1,\dots,u_k$ in $H$ that correspond to its $k$ intervals. Two vertices $u,v$ in $G$ are adjacent iff there exist $i,j \in [k]$ such that $u_i$ and $v_j$ are adjacent in $H$. It is not clear whether $k$-interval graphs can also be reduced to interval graphs in the algebraic sense. We show that algebraic reductions are a special case of subgraph reductions under certain circumstances. 

An important application of reductions is to demonstrate that certain problems are representative (complete) for certain complexity classes. The notion of completeness is also applicable to our setting and it is natural to ask what a complete graph class for $\gccp$ looks like. We show that no hereditary graph class can be complete for $\gccac$ or any superset thereof with respect to algebraic or subgraph reductions. In fact, it might very well be the case that there do not even exist complete graph classes for $\gccp$ and $\gccac$ at all. 
On the upside, in the next section we introduce classes of labeling schemes defined in terms of first-order logic for which completeness results under both types of reductions hold.

\subsection{Algebraic Reductions}
\label{sec:ar}
\begin{definition}
    Let $f$ be a $k$-ary boolean function and $H_1,\dots,H_k$ are graphs with the same vertex set $V$ and $k \geq 0$. We define $f(H_1,\dots,H_k)$ to be the graph with vertex set $V$ and an edge $(u,v)$ iff $f( x_1, \dots, x_k) =1$ where $x_i = \llbracket (u,v) \in E(H_i) \rrbracket$ for all $u,  v \in V$.
\end{definition}

The constant boolean functions $0$ and $1$ define the empty and complete graph, respectively. The negation of a graph is its edge-complement.

\begin{definition}
    Let $f$ be a $k$-ary boolean function and $\mathcal{C}_1,\dots,\mathcal{C}_k$ are graph classes and $k \geq 0$. We define $f(\mathcal{C}_1,\dots,\mathcal{C}_k)$ to be the graph class that contains every graph $G$ such that there exist $(H_1,\dots,H_k)  \in \mathcal{C}_1 \times \dots \times \mathcal{C}_k$ with $(u,v) \in E(G)$ iff $(u,v) \in E(f(H_1,\dots,H_k))$ for all $u \neq v \in V$  assuming that $G$ and $H_1,\dots,H_k$ all have vertex set $V$.
    \label{def:gciobf}
\end{definition} 

If one exclusively considers graphs without self-loops then the condition `$(u,v) \in E(G)$ iff $(u,v) \in E(f(H_1,\dots,H_k))$ for all $u \neq v \in V$' can be replaced with `$G=f(H_1,\dots,H_k)$'. 
We explain the importance of this subtle detail when we consider closure under reductions. 

Under the graph class interpretation the constant boolean functions $0$ and $1$ define the class of empty and complete graphs, respectively. The negation of a graph class $\mathcal{C}$ is the edge-complement co-$\mathcal{C}$. 
A graph $G$ has arboricity at most $k$ iff $G \in \bigvee_{i=1}^k \gcforest$. Similarly, $G$ has thickness at most $k$ iff  $G \in \bigvee_{i=1}^k \gcplanar$ and boxicity at most $k$ iff $G \in \bigwedge_{i=1}^k \gcinterval$.

Suppose you are given two boolean formulas $F_1,F_2$ with $k$ variables. Due to the previous definition we can naturally interpret the formula $F_i$ as a function $f_i$ which maps $k$ graph classes to a graph class (each subformula of $F_i$ evaluates to a graph class). The following statement shows that $f_1 = f_2$ iff $F_1$ and $F_2$ are logically equivalent.

\begin{lemma}[Compositional Equivalence]
Given boolean functions $f,g,h_1,\dots,h_l$ where $f,h_1,\dots,h_l$ have arity $k$ and $g$ has arity $l$
 such that $f$ is the composition of $g,h_1,\dots,h_l$, i.e.~$f(\vec{x})$ $= g(h_1(\vec{x}),\dots,h_l(\vec{x}))$ for all $\vec{x} \in \{0,1\}^k$. Then for all sequences of graph classes $\vec{\mathcal{C}} = (\mathcal{C}_1,\dots, \mathcal{C}_k)  $ it holds that
$$ f(\vec{\mathcal{C}}) = g(h_1(\vec{\mathcal{C}}) , \dots, h_l(\vec{\mathcal{C}}) ) $$  
\label{lem:compequiv}   
\end{lemma}
\begin{proof}
    We show that $f(\vec{H}) = g(h_1(\vec{H}),\dots,h_l(\vec{H}))$ for all sequences of graphs $\vec{H} = (H_1,\dots,H_k)$. From that it directly follows that $ f(\vec{\mathcal{C}}) = g(h_1(\vec{\mathcal{C}}) , \dots, h_l(\vec{\mathcal{C}}) ) $. 
	Let $G_f = f(\vec{H})$ and $G_g = g(h_1(\vec{H}),\dots,h_l(\vec{H}))$. It holds that
	$$ (u,v) \in E(G_f) \Leftrightarrow f(\vec{x}) = 1 \Leftrightarrow g(h_1(\vec{x}),\dots,h_l(\vec{x})) = 1 \Leftrightarrow (u,v) \in E(G_g) $$
	with $\vec{x} = (x_1,\dots,x_k)$ and $x_i = \llbracket (u,v) \in E(H_i) \rrbracket$ for $i \in [k]$.
\end{proof}

Given a boolean function $f$ of arbitrary arity it is trivial to construct a unary boolean function $g$ such that $f(x,\dots,x) = g(x)$ for all $x \in \{0,1\}$. This is not possible under the graph class interpretation. For instance, $\gcforest \vee \gcforest \neq f(\gcforest)$ for any of the four unary boolean functions $f$. Therefore disjunction is not idempotent and one can show that neither is conjunction. However, this algebra does inherit some of the properties of its boolean ancestor. For instance, disjunction and conjunction are associative and commutative, negation is an involution and De Morgan's laws apply. More generally, all laws of boolean algebra where every variable occurs at most once on each side apply to this algebra as well.  Boolean algebra is a special case of this algebra on graph classes if one restricts the universe to the class of complete graphs $\set{K_n}{n \in \N}$ and the edge-complement of it.

\begin{definition}[Algebraic Reduction]
    Let $F$ be a set of boolean functions and $\mathcal{C}, \mathcal{D}$ are graph classes. We say $\mathcal{C} \leq_F \mathcal{D}$ if $\mathcal{C} \subseteq f(\mathcal{D},\dots,\mathcal{D})$ for some $f \in F$. 
    We write $[\mathcal{D}]_F$ to denote the set of graph classes reducible to $\mathcal{D}$ w.r.t.~$\leq_F$.        
\end{definition}
 
In order for $\leq_F$ to be reflexive and transitive the following additional requirement has to be made.
A set of boolean functions is called a boolean clone if it is closed under composition and contains all projection functions $\pi_k^i(x_1,\dots,x_k) = x_i$ for $1 \leq i \leq k$ and $k \in \N$. For a set of boolean functions $F$ let us write $[F]_{\cloneclosure}$ to denote the closure of $F$ and all projections functions under composition. Alternatively, one can think of $[F]_{\cloneclosure}$ as all boolean functions that can be expressed as boolean formulas which use functions from $F$ as connectives. For example,  $[\neg,\wedge]_{\cloneclosure}$ is the set of all boolean functions $\clonebf$. A list of all boolean clones can be found in \cite[p.~145]{lau}.

\begin{lemma}
    If $F$ is a boolean clone then $\leq_F$ is reflexive and transitive.
\end{lemma}
\begin{proof}
	Reflexivity directly follows from the identity function which is a projection function and therefore contained in $F$. For transitivity let $\mathcal{C} \leq_F \mathcal{D}$ via a $k$-ary $f \in F$ and $\mathcal{D} \leq_F \mathcal{E}$ via an $l$-ary $g \in F$. We show that $\mathcal{C} \leq_F \mathcal{E}$ via the $kl$-ary boolean function $h(\vec{x_1},\dots,\vec{x_k}) = f(g(\vec{x_1}),\dots,g(\vec{x_k}))$ where $\vec{x_i}$ denotes a sequence of $l$ variables. Since $h$ is the composition of $f$ and $g$ it follows that $h$ is in $F$. Consider a graph $G$ in $\mathcal{C}$. There are graphs $H_1,\dots,H_k \in \mathcal{D}$ such that $G = f(H_1,\dots,H_k)$. Since $H_i \in \mathcal{D}$ there are graphs $I^i_1,\dots,I^i_l \in \mathcal{E}$ such that $H_i = g(I^i_1,\dots, I^i_l)$ for all $i \in [k]$. Therefore $G = h(\vec{I^1},\dots,\vec{I^k}) =  f(g(\vec{I^1}),\dots,g(\vec{I^k}))$ with $\vec{I^i} = (I^i_1,\dots,I^i_l)$ for $i \in [k]$. 
\end{proof}

We say a set of graph classes $\mathbb{A}$ is closed under a $k$-ary boolean function $f$ if for all $\mathcal{C}_1,\dots,\mathcal{C}_k \in \mathbb{A}$ it holds that $f(\mathcal{C}_1,\dots,\mathcal{C}_k) \in \mathbb{A}$. We say $\mathbb{A}$ is closed under $\leq_F$ for a set of boolean functions $F$ if $\mathcal{C} \leq_F \mathcal{D}$ and $\mathcal{D} \in \mathbb{A}$ implies $\mathcal{C} \in \mathbb{A}$. 

\begin{lemma}
	Let $F$ be a boolean clone and let $B$ be a set of boolean functions with $F = [B]_{\cloneclosure}$. 
    If a set of graph classes $\mathbb{A}$ is closed under subsets and $f$ for every $f \in B$ then  $\mathbb{A}$ is closed under $\leq_{F}$. 
    \label{lem:cloneregclos}
\end{lemma}
\begin{proof}
    Let $\mathbb{A}$ be closed under subsets and every boolean function in $B$. We need to argue that $\mathbb{A}$ is closed under $f$ for every $f \in F$. From that and the closure under subsets it follows that $\mathbb{A}$ is closed under $\leq_F$. 
    
    Observe that if $\mathbb{A}$ is closed under some boolean functions then it is also closed under the composition of these functions for the following reason. Let $f,g,h_1,\dots,h_l$ be boolean functions where $f,h_1,\dots,h_l$ have arity $k$ and $g$ has arity $l$ and $f$ is the composition of $g$ with $h_1,\dots,h_l$. Let $\mathbb{A}$ be closed under $g$ and $h_1,\dots,h_l$. Given $\vec{\mathcal{D}} \in \mathbb{A}^k$ it holds that $f(\vec{\mathcal{D}}) = g(h_1(\vec{\mathcal{D}}),\dots,h_l(\vec{\mathcal{D}}))$ due to the compositional equivalence from Lemma \ref{lem:compequiv}. Since $\mathbb{A}$ is closed under $h_i$ it follows that $\mathcal{\mathcal{D}}_i := h_i(\vec{\mathcal{D}})$ is in $\mathbb{A}$ for all $i \in [l]$. Therefore  $g(\mathcal{\mathcal{D}}_1,\dots,\mathcal{\mathcal{D}}_l)$ is in $\mathbb{A}$ as well and hence $\mathbb{A}$ is closed under $f$.     
    Since every function in $F$ can be expressed as composition of functions from $B$ and projections and $\mathbb{A}$ is closed under every function from $B$ and projections it follows that it is closed under every function from $F$. 
\end{proof}

If $\mathbb{A}$ is closed under union then the implication in Lemma~\ref{lem:cloneregclos} becomes an equivalence. 

\begin{corollary}
    A set of graph classes is closed under $\bfreduction$ if it is closed under subsets, negation and conjunction. Moreover, for a set of graph classes $\mathbb{A}$ that is closed under union it holds that $\mathbb{A}$ is closed under $\bfreduction$ iff it is closed under subsets, negation and conjunction.   
    \label{corol:bfcompl}
\end{corollary}

Before we show closure under $\bfreduction$ for certain classes, let us go back to that subtle detail concerning self-loops in Definition~\ref{def:gciobf}. Let us say two graphs $G,H$ on the same vertex set $V$ are equivalent, in symbols $G \equiv H$, if $(u,v) \in E(G) \Leftrightarrow (u,v) \in E(H)$ for all \emph{distinct} $u$ and $v$ in $V$. This means that self-loops are ignored when comparing $G$ and $H$.
For a graph class $\mathcal{C}$ let $[\mathcal{C}]_{\equiv}$ be the set of graphs that are equivalent to a graph in $\mathcal{C}$. We say a graph class $\mathcal{C}$ is closed under $\equiv$ if $\mathcal{C} = [\mathcal{C}]_{\equiv}$. Given $k \in \N$, a $k$-ary boolean function $f$ and graph classes $\mathcal{C}_1,\dots,\mathcal{C}_k$ it holds that $f(\mathcal{C}_1,\dots,\mathcal{C}_k)$ is closed under $\equiv$. Intuitively, this means a reduction does not need to preserve information about self-loops. In contrast, in a labeling scheme self-loops must be correctly described by the label decoder $F$ and a labeling $\ell$, i.e.~$(u,u) \in E(G)$ iff $(\ell(u),\ell(u)) \in F$. This becomes relevant when we consider whether a set of graph classes $\mathbb{A}$ is closed under reductions. 
However, we can safely ignore the self-loop issue because for every set of graph classes $\mathbb{A}$ for which we show closure under reductions it can be easily seen that $\mathcal{C} \in \mathbb{A}$ implies $[\mathcal{C}]_{\equiv} \in \mathbb{A}$. Practically, this means our constructed labeling schemes do not need to correctly capture self-loops. 

\begin{lemma}
	The classes $\gccac, \gccl, \gccp, \gccexp, \gccr$ and $\gccall$ are closed under $\bfreduction$.
\end{lemma}
\begin{proof}
    If a set of languages $\ccex$ is closed under complement then $\gccex$ is closed under negation.
    To see that the above mentioned classes are closed under conjunction, consider the following construction. Given two labeling schemes $S_1 = (F_1,c_1)$ and $S_2 = (F_2,c_2)$ let the labeling scheme $S_3 = (F_3,c_1+c_2)$ with $(x_1x_2,y_1y_2) \in F_3 \Leftrightarrow (x_1,y_1) \in F_1 \wedge (x_2,y_2) \in F_2$ and $\frac{|x_i|}{|x|} = \frac{|y_i|}{|y|} = \frac{c_i}{c}$ for $i \in [2]$. It holds that $\gr{S_1} \wedge \gr{S_2} \subseteq \gr{S_3}$ and it is simple to see that this construction works for all of the above complexity classes. 
\end{proof}

\begin{lemma}
    The class $\cogcsmallh$ is closed under $\bfreduction$.
\end{lemma}
\begin{proof}
   If a graph class is small and hereditary then so is its edge-complement and therefore we have closure under negation. It remains to show that this set of graph classes is closed under conjunction. Let $\mathcal{C},\mathcal{D}$ be small and hereditary graph classes. A graph in $\mathcal{C} \wedge \mathcal{D}$ is determined by choosing a graph in $\mathcal{C}$ and $\mathcal{D}$ and therefore   
   $ |(\mathcal{C} \wedge \mathcal{D})_n| \leq |\mathcal{C}_n| \cdot |\mathcal{D}_n| \in n^{\mathcal{O}(n)} $ since $\mathcal{C}$ and $\mathcal{D}$ are small. Given a graph $G \in \mathcal{C} \wedge \mathcal{D}$. Let $G = H_1 \wedge H_2$ with $H_1 \in \mathcal{C}, H_2 \in \mathcal{D}$. It holds that $G[V'] = H_1[V'] \wedge H_2[V']$ for all vertex subsets $V'$ of $G$. Since $\mathcal{C}$ and $\mathcal{D}$ are hereditary it follows that every induced subgraph of $G$ is in $\mathcal{C} \wedge \mathcal{D}$. 
\end{proof}

Let us say a graph class $\mathcal{C}$ is $\bfreduction$-complete for a set of graph classes $\mathbb{A}$  if $\mathcal{C}$ is in $\mathbb{A}$ and $\mathcal{D} \bfreduction \mathcal{C}$ for every $\mathcal{D} \in \mathbb{A}$. Alternatively, one can also say $\mathcal{C}$ is complete for $\mathbb{A}$ if $[\mathcal{C}]_{\bfreduction} = \mathbb{A}$. Observe that a directed graph class is not $\bfreduction$-reducible to an undirected one. Since classes such as $\gccp$ and $\gccac$ contain directed graph classes it trivially follows that no undirected graph class can be complete for them.
If we assume that only undirected graph classes can be small and hereditary then it trivially holds that $\gccac$ and its supersets do not have a $\bfreduction$-complete small and hereditary graph class. To make this more interesting we can either drop this assumption or we can restrict $\gccac$ to undirected graph classes. Irregardless of the choice, the following statement holds.

\begin{fact}
    There exists no small, hereditary  graph class that is $\bfreduction$-complete for $\gccac$, or a superset thereof.  
\end{fact}
\begin{proof}
    Assume there exists a small, hereditary graph class $\mathcal{C}$ that is complete for $\gccac$ with respect to $\bfreduction$. Since $\mathcal{C}$ is in  $\cogcsmallh$ and  $\cogcsmallh$ is closed under $\bfreduction$ it follows that every class reducible to $\mathcal{C}$ must be in $\cogcsmallh$ and thus $\gccac \subseteq \cogcsmallh$. This contradicts Theorem \ref{thm:acsmallh}.
\end{proof}
 
As a closing example, we explain why every uniformly sparse graph class is $\bfreduction$-reducible to interval graphs. A graph class is uniformly sparse iff it has bounded arboricity. Stated differently, $\mathcal{C}$ is uniformly sparse iff there exists a $k \in \N$ such that $\mathcal{C} \subseteq \bigvee_{i=1}^k \gcforest$. It holds that every forest has boxicity at most 2, i.e.~$\gcforest \subseteq \gcinterval \wedge \gcinterval$. Therefore if  $\mathcal{C}$ is uniformly sparse then there exists a $k \in \N$ such that $\mathcal{C} \subseteq \bigvee_{i=1}^k (\gcinterval \wedge \gcinterval)$. Notice that this reduction does not require negation and therefore is a $\leq_{\mathrm{M}}$-reduction  where $\mathrm{M}$ denotes the monotone clone $[\wedge,\vee]_{\cloneclosure}$.

\subsection{Subgraph Reductions}

Given $k \geq 0$, a $k^2$-ary boolean function $f$ and a $(k \times k)$-matrix $A$ over $\{0,1\}$. We write $f(A)$ to denote the value of $f$ when plugging in the entries of $A$ from left to right and top to bottom. We say $f$ is diagonal if the value of $f$ only depends on the $k$ entries on the main diagonal of $A$. Given $k,l \geq 1$, a $k^2$-ary boolean function $f$ and an $l^2$-ary boolean function $g$. We define the composition of $f$ with $g$ to be the $(kl)^2$-ary boolean function 
$$(f \circ g ) 
\left( \begin{matrix}
B_{1,1} &  \dots &  B_{1,k} \\
 \vdots &  \ddots 	& \vdots \\
B_{k,1} &  	\dots		& B_{k,k}
\end{matrix} \right) 
= 
f 
\left( \begin{matrix}
g(B_{1,1}) &  \dots &  g(B_{1,k}) \\
\vdots &  \ddots 	& \vdots \\
g(B_{k,1}) &  	\dots		& g(B_{k,k})
\end{matrix} \right)
$$
where $B_{i,j}$ is a $(l \times l)$-matrix for $i,j \in [k]$. 

\begin{definition}
	Given graphs $G,H$, $k \in \N$ and a $k^2$-ary boolean function $f$. We say $G$ has an $(H,f)$-representation if there exists an $\ell \colon V(G) \rightarrow V(H)^k$ such that for all $u \neq v \in V(G)$ 
    $$(u,v) \in E(G) \Leftrightarrow f(A_{uv}^{\ell}) = 1$$
    with $A_{uv}^{\ell} = (\llbracket (\ell(u)_i, \ell(v)_j) \in E(H)  \rrbracket)_{i,j \in [k]} $.
\end{definition}

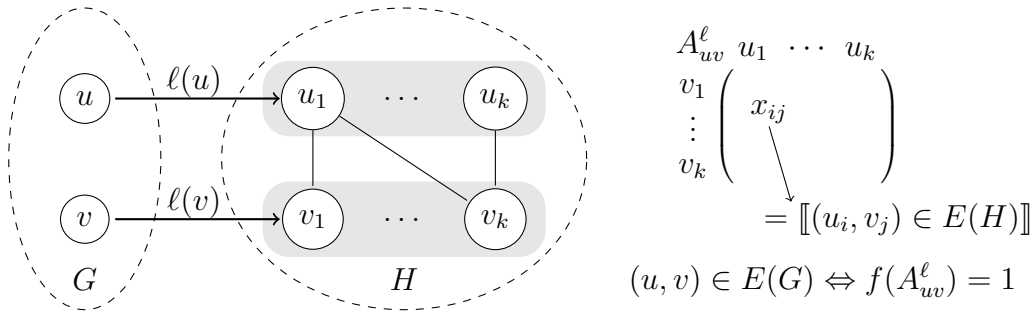
\begin{figure}[b]
    \centering
    \begin{tikzpicture}[shorten >=1pt,auto,node distance=1.2cm,
main node/.style={circle,draw}]
\usetikzlibrary{shapes.misc}

\newcommand*{\xa}{0}%
\newcommand*{\ya}{0}%
\newcommand*{\xb}{4.2}%
\newcommand*{\xc}{9.7}%

\newcommand*{\movea}{0.8}%
\newcommand*{\moveb}{0.63}%
\newcommand*{\movec}{0.4}%

\newcommand*{\nd}{1.2}%

\node[main node] (un) at (\xa,\ya+\movea) {\large $u$};
\node[main node] (vn) at (\xa,\ya-\movea) {\large $v$};
\draw[dashed] (\xa,\ya) ellipse (1cm and 2cm);
\node (gn) at (\xa,\ya-2*\movea) {\large $G$};

\fill[black!10,rounded corners=10pt] (\xb-1.55*\nd,\ya+\movea+0.5) rectangle (\xb+1.55*\nd,\ya+\movea-0.5);
\fill[black!10,rounded corners=10pt] (\xb-1.55*\nd,\ya-\movea+0.5) rectangle (\xb+1.55*\nd,\ya-\movea-0.5);
\node[main node,fill=white] (u1n) at (\xb-1*\nd,\ya+\movea) {\large$u_1$};
\node[main node,fill=white] (v1n) at (\xb-1*\nd,\ya-\movea) {\large$v_1$};
\node (u3n) at (\xb+0*\nd,\ya+\movea) {\large$\dots$};
\node (v3n) at (\xb+0*\nd,\ya-\movea) {\large$\dots$};
\node[main node,fill=white] (u4n) at (\xb+1*\nd,\ya+\movea) {\large$u_k$};
\node[main node,fill=white] (v4n) at (\xb+1*\nd,\ya-\movea) {\large$v_k$};
\draw[dashed] (\xb,\ya) ellipse (2.4cm and 2cm);
\node (hn) at (\xb,\ya-2*\movea) {\large $H$};

\path[-]
(u1n) edge (v1n)
(u1n) edge (v4n)
(u4n) edge (v4n)
;

\draw[thick,->] (\xa+0.4,\ya+\movea) -- (\xb-1*\nd-0.4,\ya+\movea);
\node (hn) at (\xa+1.8*\movea,\ya+\movea+0.23) {\large $\ell(u)$};

\draw[thick,->] (\xa+0.4,\ya-\movea) -- (\xb-1*\nd-0.4,\ya-\movea);
\node (hn) at (\xa+1.8*\movea,\ya-\movea+0.23) {\large $\ell(v)$};

\node (hn) at (\xc-1.2-\movec,\ya+\movea+0.1+\moveb) {\large $A_{uv}^\ell$};
\node (hn) at (\xc+0.2-\movec,\ya+\movea+\moveb) {\large$u_1 \: \: \cdots \: \: u_k$};
\node (hn) at (\xc+0.2-\movec,\ya-0.2+\moveb) {$\left( 
    {\color{white} 
\begin{matrix}
a_1 & a_2 & a 345  \\
b_1 & b_2 & b 4 \\
c_1 & c_2 & c 3  \\
d_1 & d_2 & d 3 \\
c_1 & c_2 & c 3  
\end{matrix} }
  \right)$};

\node (hn) at (\xc-1.3-\movec,\ya+\movea-0.5+\moveb) {\large$v_1$};
\node (hn) at (\xc-1.3-\movec,\ya+\movea-0.93+\moveb) {\large$\vdots$};
\node (hn) at (\xc-1.3-\movec,\ya+\movea-0.93-0.63+\moveb) {\large$v_k$};

\node (hn) at (\xc-\movec-0.3,\ya+\moveb) {\large$x_{ij}$};
\draw[->] (\xc-\movec-0.3,\ya+\moveb-0.2) -- (\xc+1-1.4,\ya-1*\movea+0.2);

\node (hn) at (\xc+1,\ya-1*\movea) {\large$= \llbracket (u_i,v_j) \in E(H) \rrbracket$};

\node (hn) at (\xc,\ya-2*\movea) {\large$(u,v) \in E(G) \Leftrightarrow f(A_{uv}^\ell) = 1$};

\end{tikzpicture}
    \caption{Schematic of an $(H,f)$-representation of a graph $G$ with labeling $\ell$}    
    \label{fig:sgreduction}                
\end{figure}

\begin{definition}[Subgraph Reduction]
	Given graph classes $\mathcal{C},\mathcal{D}$. We say $\mathcal{C} \sgreduction \mathcal{D}$  if there exist $c,k \in \N$ and a $k^2$-ary boolean function $f$ such that for all $n \in \N$ and $G \in \mathcal{C}_n$ there exists an $H \in \mathcal{D}_{n^c}$ such that $G$ has an $(H,f)$-representation. We say $\mathcal{C} \sgdreduction \mathcal{D}$ if this holds for a diagonal $f$. We write $[\mathcal{C}]_{\sgclosure}$ to denote the set of graph classes $\sgreduction$-reducible to $\mathcal{C}$.
\end{definition}

In the case of intersection graph classes the subgraph reduction can be simplified. Given two families of sets $X$ and $Y$ let $\mathcal{C}_X$ and $\mathcal{C}_Y$ denote the intersection graph classes that they induce. It holds that $\mathcal{C}_Y \sgreduction \mathcal{C}_X$ if there exists a $k \in \N$, a $k^2$-ary boolean function $f$ and a labeling $\ell \colon Y \rightarrow X^k$ such that for all $u \neq v \in Y$ it holds that $u \cap v \neq \emptyset \Leftrightarrow f(A_{uv}^\ell) = 1$ with $A_{u,v}^{\ell} = (\llbracket \ell(u)_i \cap \ell(v)_j \neq \emptyset \rrbracket)_{i,j \in [k]}$. For instance, let $X$ be the set of intervals on the real line and $Y = \set{\cup_{i=1}^k x_i }{ x_1,\dots,x_k \in X}$ for some $k \in \N$. This means $\mathcal{C}_X$ is the set of interval graphs and $\mathcal{C}_Y$ is the set of $k$-interval graphs. To reduce $\mathcal{C}_Y$ to $\mathcal{C}_X$ we can choose the labeling $\ell(y) = (x_1,\dots,x_k)$ with $y \in Y$ and $x_1,\dots,x_k \in X$ such that $y = \cup_{i=1}^k x_i$ and $f$ is the boolean function that returns one if at least one of its arguments is one. 

Let $\mathcal{C},\mathcal{D}$ be two graph classes such that $\mathcal{C} \sgreduction \mathcal{D}$ via $c,k \in \N$ and a $k^2$-ary boolean function $f$. In the definition of subgraph reductions it is required that every graph on $n$ vertices in $\mathcal{C}$ must have an $(H,f)$-representation for a graph $H \in \mathcal{D}$ on exactly $n^c$ vertices. The following lemma shows that if $\mathcal{D}$ satisfies some fairly weak conditions then this size restriction becomes obsolete.

Recall that a graph class $\mathcal{C}$ is called inflatable if for every graph $G$ on $n$ vertices in $\mathcal{C}$ and all $m > n$ there exists a graph on $m$ vertices in $\mathcal{C}$ which contains $G$ as induced subgraph.

\begin{lemma}
    Let $\mathcal{D}$ be a hereditary and inflatable graph class. It holds for all graph classes $\mathcal{C}$ that $\mathcal{C} \sgreduction \mathcal{D}$ iff there exists a $k \in \N$ and a $k^2$-boolean function $f$ such that every graph in $\mathcal{C}$ has an $(H,f)$-representation for some $H \in \mathcal{D}$.
    \label{lem:sgwoc}
\end{lemma}
\begin{proof}
    The direction ``$\Rightarrow$'' is clear. For the other direction consider a graph class $\mathcal{C}$ and let $f$ be a $k^2$-ary boolean function such that every graph in $\mathcal{C}$ has an $(H,f)$-representation for some $H \in \mathcal{D}$. We show that $\mathcal{C} \sgreduction \mathcal{D}$ via $c=k,k$ and $f$. Let $G$ be a graph in $\mathcal{C}$ with $n$ vertices. We show that there exists a graph $H''$ on $n^k$ vertices in $\mathcal{D}$ such that $G$ has an $(H'',f)$-representation. There exists a graph $H$ in $\mathcal{D}$ such that $G$ has an $(H,f)$-representation via a labeling $\ell \colon V(G) \rightarrow V(H)^k$. Let $V'$ be the set of vertices of $H$ that occur in the image of $\ell$. Let $H'$ be the induced subgraph of $H$ on $V'$. It holds that $H'$ has at most $kn$ vertices and $G$ has an $(H',f)$-representation via $\ell$. Since $\mathcal{D}$ is hereditary it holds that $H'$ is in it. Let $H''$ be a graph with $n^k$ vertices in $\mathcal{D}$ which contains $H'$ as induced subgraph. Since $\mathcal{D}$ is inflatable such a graph must exist. It follows that $G$ has an $(H'',f)$-representation via $\ell$. 
\end{proof}

\begin{lemma}
	$\sgreduction$ and $\sgdreduction$ are reflexive and transitive. 
\end{lemma}
\begin{proof}
    For reflexivity it suffices to show that $\mathcal{C} \sgdreduction \mathcal{C}$. This holds because every $G \in \mathcal{C}$ has a $(G,f)$-representation with $f(x) = x$ ($c = k = 1$) and $f$ is diagonal.
    
    For transitivity assume that $\mathcal{C} \sgreduction \mathcal{D}$ via $c,k \in \N$ and a $k^2$-ary boolean function $f$  and $\mathcal{D} \sgreduction \mathcal{E}$ via $d,l  \in \N$ and an $l^2$-ary boolean function $g$. We show that $\mathcal{C} \sgreduction \mathcal{E}$ via $cd,kl$ and the $(kl)^2$-ary boolean function $f \circ g$. Given $G \in \mathcal{C}_n$ we need to show that there exists an $I \in \mathcal{E}_{n^{cd}}$ such that $G$ has an $(I,f \circ g)$-representation. Since $G \in \mathcal{C}_n$ there exist an $H \in \mathcal{D}_{n^c}$ such that $G$ has an $(H,f)$-representation via a labeling $\ell \colon V(G) \rightarrow V(H)^k$. Also, there exists an $I \in \mathcal{E}_{n^{cd}}$ such that $H$ has an $(I,g)$-representation via a labeling $\ell' \colon V(H) \rightarrow V(I)^l$. Then it can be verified that $G$ has an $(I,f \circ g)$-representation due to the following labeling $\ell'' \colon V(G) \rightarrow V(I)^{kl}$. For $u \in V(G)$ let $\ell(u) = (u_1,\dots,u_k)$ and let $\ell'(u_i) = (u_{i,1},\dots,u_{i,l})$ for $i \in [k]$. We define $\ell''(u)$ as $(u_{1,1},\dots,u_{1,l},u_{2,1},\dots,u_{2,l},\dots,u_{k,1},\dots,u_{k,l})$.
    The same argument shows that $\sgdreduction$ is transitive because the composition of two diagonal boolean functions yields a diagonal boolean function.  
\end{proof}

\begin{lemma}
    The class $\gccac$ is closed under $\sgreduction$.
    \label{lem:sgacclosure}
\end{lemma}
\begin{proof}
    Let $\mathcal{C} \sgreduction \mathcal{D}$ and let $\mathcal{D} \in \gccac$. We need to show that $\mathcal{C} \in \gccac$. 
Let $\mathcal{D} \in \gccac$ via the labeling scheme $S=(F,c)$ and $\mathcal{C} \sgreduction \mathcal{D}$ via $d,k \in \N$ and a $k^2$-ary boolean function $f$. Then we claim that the labeling scheme $S' = (F',cdk)$ with 
$$ (x_1 \dots x_k,y_1 \dots y_k) \in F' \Leftrightarrow 
f \left( 
\begin{matrix}
\llbracket (x_1,y_1) \in F \rrbracket &   \llbracket (x_1,y_2) \in F \rrbracket & \dots \\
\vdots &  \ddots 	&  \\
\llbracket (x_k,y_1) \in F \rrbracket &  			& \llbracket (x_k,y_k) \in F \rrbracket
\end{matrix}
\right)  
= 1 $$   
for all $x_i,y_i \in \{0,1\}^{cdm}$ with $m \in \N$ and $i \in [k]$
represents $\mathcal{C}$ and $F'$ can be computed in $\AC^0$ since $F$ can be computed in $\AC^0$. We show how to label a given $G \in \mathcal{C}_n$. Since $\mathcal{C} \sgreduction \mathcal{D}$ there exist an $H \in \mathcal{D}_{n^d}$ such that $G$ has an $(H,f)$-representation via the labeling $\ell_G \colon V(G) \rightarrow V(H)^k$. Since $\mathcal{D}$ is represented by $S$ there also exists a labeling $\ell_H \colon V(H) \rightarrow \{0,1\}^{c \log n^d}$ such that $H$ is represented by $S$ via $\ell_H$. Then a vertex $u$ of $G$ can be labeled with $\ell_H(u_1) \dots \ell_H(u_k)$  where $u_i$ is the $i$-th component of $\ell_G(u)$.
\end{proof}

\begin{corollary}
    The classes $\gccl, \gccp, \gccexp, \gccr$ and $\gccall$ are closed under $\sgreduction$.
\end{corollary}
\begin{proof}
    Observe that for all these complexity classes the same argument as the one given for $\gccac$ in Lemma \ref{lem:sgacclosure} works. 
\end{proof}

\begin{lemma}
        Let $\mathcal{C},\mathcal{D}$ be graph classes and $\mathcal{D}$ is in $\cogcsmallh$. If $\mathcal{C} \sgreduction \mathcal{D}$ holds then $\mathcal{C}$ is small.
\end{lemma}
\begin{proof}
    Let $\mathcal{C} \sgreduction \mathcal{D}$ via $c,k \in \N$ and a $k^2$-ary boolean function $f$ and $\mathcal{D}$ is small and hereditary. A graph $G$ with $n$ vertices in $\mathcal{C}$ is determined by a graph $H$ with $n^c$ vertices in $\mathcal{D}$ and a labeling $\ell \colon V(G) \rightarrow V(H)^k$. Observe that the adjacency relation of $G$ only depends on an induced subgraph $H'$ of $H$ with at most $kn$ vertices (every vertex of $H$ that occurs in the image of $\ell$). Since $\mathcal{D}$ is small and hereditary the number of options for $H'$ is limited by 
    $$ \sum_{i=1}^{kn} |\mathcal{D}_i| \leq \sum_{i=1}^{kn} i^{\mathcal{O}(i)} \leq kn \cdot  kn^{\mathcal{O}(kn)} \leq n^{\mathcal{O}(n)}  $$ 
    The number of different labelings $\ell \colon V(G) \rightarrow V(H')^k$ is limited by ${(kn)}^{kn} \leq n^{\mathcal{O}(n)}$. Therefore $\mathcal{C}$ is small.        
\end{proof}

\begin{lemma}
    The class $\cogcsmallh$ is closed under $\sgreduction$.
\end{lemma}
\begin{proof}
    Let $\mathcal{C},\mathcal{D}$ be graph classes such that $\mathcal{C} \sgreduction \mathcal{D}$ via $c,k \in \N$ and a $k^2$-ary boolean function $f$ and $\mathcal{D}$ is in $\cogcsmallh$. For the sake of contradiction let us assume that $\mathcal{C}$ is not a subset of a small and hereditary graph class. Let $\mathcal{C}'$ be the closure of $\mathcal{C}$ under induced subgraphs. The class $\mathcal{C}'$ cannot be small since this would imply that $\mathcal{C}$ is in $\cogcsmallh$. We argue that there exists a finite graph class $\mathcal{E}$ such that $\mathcal{C}' \setminus \mathcal{E} \sgreduction \mathcal{D}$ via $c,k$ and $f$. From the previous lemma it follows that $\mathcal{C}' \setminus \mathcal{E}$ is small. Since $\mathcal{E}$ is only finite it follows that $\mathcal{C}'$ is also small, which is a contradiction. Let $n_0$ be the smallest number such that $kn_0 \leq n_0^c$ and $\mathcal{E}$ is the class of all graphs with at most $n_0$ vertices. Let $G' \in \mathcal{C}' \setminus \mathcal{E}$ with $n'$ vertices. We need to show that there exists a graph $H' \in \mathcal{D}$ on ${(n')}^c$ vertices such that $G'$ has an $(H',f)$-representation. There exists a $G \in \mathcal{C}$ on $n$ vertices such that $G'$ is an induced subgraph of $G$ and an $H \in \mathcal{D}_{n^c}$ such that $G$ has an $(H,f)$-representation via a labeling $\ell \colon V(G) \rightarrow V(H)^k$. Therefore $G'$ has an $(H,f)$-representation via the labeling $\ell' \colon V(G') \rightarrow V(H)^k$  defined as the restriction of $\ell$ to vertices from $G'$. Since $G'$ is an induced subgraph of $G$ it holds that $n' \leq n$ and therefore ${(n')}^c \leq n^c$. Additionally, since $G' \notin \mathcal{E}$ it holds that $kn' \leq {(n')}^c$. To construct the desired graph $H' \in \mathcal{D}$ with ${(n')}^c$ vertices such that $G'$ has an $(H',f)$-representation one can delete vertices from $H$ until only ${(n')}^c$ vertices remain but every vertex in the image of $\ell'$ is still in $H'$.
\end{proof}

\begin{corollary}
    There exists no small, hereditary graph class that is complete for $\gccac$ (or a superset thereof) with respect to $\sgreduction$.
\end{corollary}
\begin{proof}
    $\gccac$ is no subset of $\cogcsmallh$ due to Theorem \ref{thm:acsmallh}. If a small and hereditary graph class is complete for $\gccac$ w.r.t.~$\sgreduction$ then this would imply that $\gccac$  is a subset of $\cogcsmallh$ since $\cogcsmallh$ is closed under $\sgreduction$.
\end{proof}

Recall that a graph class $\mathcal{C}$ is called self-universal if for every finite subset of $\mathcal{C}$ there exists a graph in $\mathcal{C}$ that contains every graph of that finite subset as induced subgraph. 

\begin{lemma}
    Let $\mathcal{D}$ be a self-universal and inflatable graph class. Then it holds that $\mathcal{C} \bfreduction \mathcal{D}$ implies $\mathcal{C} \sgdreduction \mathcal{D}$ for all graph classes $\mathcal{C}$. 
    \label{lem:bftosg}
\end{lemma}
\begin{proof}
Since $\mathcal{C} \bfreduction \mathcal{D}$ there exists a $k$-ary boolean function $f$ for some $k \in \N$ such that $\mathcal{C} \subseteq f(\mathcal{D},\dots,\mathcal{D})$. Let $f'$ be defined as the $k^2$-ary diagonal boolean function $f'(A) = f(A_{1,1},\dots,A_{k,k})$. To show that $\mathcal{C} \sgdreduction \mathcal{D}$ we argue that for all $n \in \N$ every graph $G \in \mathcal{C}_n$ has an $(H,f')$-representation for some $H \in \mathcal{D}_{n^k}$. It holds that $G = f(H_1,\dots,H_k)$ for some $H_1,\dots,H_k \in \mathcal{D}_n$ and $G,H_1,\dots,H_k$ all have the same vertex set $V$. 	
Since $\mathcal{D}$ is self-universal and inflatable there exists a graph $H \in \mathcal{D}_{n^k}$ which contains $H_1,\dots,H_k$ as induced subgraphs.	 
Let $\pi_i \colon V \rightarrow V(H)$ be the witness that shows that $H_i$ is an induced subgraph of $H$ for $i \in [k]$.  Then $G$ has an $(H,f)$-representation via the labeling $\ell(u) = (\pi_1(u),\dots,\pi_k(u))$.    
\end{proof}

\begin{fact}
    Let $\mathcal{D}$ be an intersection graph class. For all graph classes $\mathcal{C}$ it holds that $\mathcal{C} \bfreduction \mathcal{D}$ iff $\mathcal{C} \sgdreduction \mathcal{D}$.
\end{fact}
\begin{proof}
    ``$\Rightarrow$'': Every intersection graph class is self-universal and inflatable. Therefore this direction follows from the previous lemma.    
    
	``$\Leftarrow$'': Let $\mathcal{C} \sgdreduction \mathcal{D}$ via $c,k \in \N$ and a $k^2$-ary diagonal boolean function $f$. Let $g$ be the $k$-ary boolean function which underlies $f$, i.e.~$f(A) = g(A_{1,1},\dots,A_{k,k})$. We claim that $\mathcal{C} \bfreduction \mathcal{D}$ via $g$. This means for every $G \in \mathcal{C}$ we need to show that there exist $H_1,\dots,H_k \in \mathcal{D}$ on vertex set $V(G)$ such that $G = g(H_1,\dots,H_k)$. Given $G \in \mathcal{C}_n$ there exist an $H \in \mathcal{D}_{n^c}$ and a labeling $\ell \colon V(G) \rightarrow V(H)^k$ such that $(u,v) \in E(G)$ iff $g(x_1,\dots,x_k) = 1$ with $x_i := \llbracket (\ell(u)_i,\ell(v)_i) \in E(H)  \rrbracket$. Since $\mathcal{D}$ is an intersection graph class we can assume without loss of generality that no vertex of $H$ occurs more than once in the image of $\ell$ (otherwise we could just clone that vertex; every vertex in $H$ has a self-loop since $H$ is an intersection graph). 
    Let $H_i$ be the induced subgraph of $H$ that has $\set{ \ell(u)_i }{ u \in V(G)}$ as vertex set for $i \in [k]$. Since $\mathcal{D}$ is hereditary it follows that $H_i$ is in $\mathcal{D}$ for all $i \in [k]$.	
\end{proof}

\section{Logical Labeling Schemes}
\label{sec:logic}
Many of the graph classes which are known to have an implicit representation can be represented by a labeling scheme where a vertex label is interpreted as a fixed number of non-negative, polynomially bounded integers and the label decoding algorithm performs basic arithmetic on these integers and compares the results. First-order logic provides a fairly natural way to describe this class of labeling schemes.
For example, consider the labeling scheme for interval graphs that we have seen in the introduction. A vertex of an interval graph is labeled with two numbers which represent the endpoints of its interval. The label decoder for this labeling scheme can be expressed as FO formula $\varphi(x_1,x_2,y_1,y_2) \triangleq \neg (x_2 < y_1 \vee y_2 < x_1)$. Given two vertices $u,v$ of an interval graph on $n$ vertices with labels $u_1,u_2,v_1,v_2 \in [2n]$, the intervals $[u_1,u_2]$ and $[v_1,v_2]$ intersect iff $\varphi(u_1,u_2,v_1,v_2)$ holds, assuming that $u_1 \leq u_2$ and $v_1 \leq v_2$. Our motivation for introducing logical labeling schemes is that it seems probable that lower bounds for hereditary graph classes can be proved against certain fragments of such logical labeling schemes.

In the first part of this section we compare the expressiveness of logical labeling schemes for various fragments and relate them to classes such as $\gccac$.
In the second part we show that if arithmetic (addition and multiplication) is disallowed then the resulting sets of graph classes have various complete graph classes, such as trees and interval graphs.
In the third part we consider a generalization of quantifier-free logical labeling schemes called polynomial-boolean systems (PBS) where the size restriction is dropped, i.e.~the numbers of a vertex label are not polynomially bounded anymore. As a consequence, such a system cannot be regarded as a labeling scheme.  
PBS are interesting for two reasons. First, they contain many of the candidates for the implicit graph conjecture. Secondly, we show that quantifier-free logical labeling schemes can be separated from certain subclasses by proving that these subclasses do not coincide with the potentially much larger set of graph classes representable by polynomial-boolean systems.

\subsection{Definition and Basic Properties}

\begin{definition}
    A (quantifier-free, atomic) logical labeling scheme is a tuple $S=(\varphi,c)$ with a (quantifier-free, atomic) formula $\varphi \in \FO_{2k}$ and $c,k \in \N$. 
    A graph $G$ is in $\gr{S}$  if there exists a labeling $\ell \colon V(G) \rightarrow {[n^c]}_0^k$ such that $(u,v) \in E(G) \Leftrightarrow \mathcal{N}_{n^c},(\ell(u),\ell(v)) \models \varphi$ for all $u, v \in V(G)$. 
    \label{def:lls}
\end{definition}
\begin{definition}
    Let $\sigma \subseteq \{ \ltp, \addp, \mulp \}$. A graph class $\mathcal{C}$ is in $\ccg\FO(\sigma)$ if there exists a logical labeling scheme $(\varphi,c)$ with $\varphi \in \FO_{2k}(\sigma)$ and $c,k \in \N$ such that $\mathcal{C} \subseteq \gr{\varphi,c}$. Let $\ccg\FO_{\mathrm{qf}}(\sigma)$ denote the quantifier-free analogue.
\end{definition}

We say a logical labeling scheme $S=(\varphi,c)$ is in $\ccg\FO_{(\mathrm{qf})}(\sigma)$ if $\varphi$ is a (quantifier-free) formula in $\FO(\sigma)$.

\begin{lemma}
    Let $\sigma \subseteq \{\ltp,\addp,\mulp\}$ and let $\ccex$ be a complexity class that is closed under $\AC^0$ many-one reductions.
    If the bounded model checking problem for every (quantifier-free) formula in $\FO(\sigma)$ can be decided in $\ccex$ then $\ccg\FO_{(\mathrm{qf})}(\sigma) \subseteq \gccex$. 
\end{lemma}
\begin{proof}
    We show that the above statement holds for $\sigma = \{\ltp, \addp, \mulp \}$ and formulas with quantifiers. From this the result for all restricted classes of formulas follows.

    Let us assume that the bounded model checking problem for every formula in $\FO$ can be decided in $\ccex$. Consider a logical labeling scheme $S=(\varphi,c)$ with $2k$ variables. Let $G$ be a graph that is in $\gr{S}$ via a labeling $\ell \colon V(G) \rightarrow [n^c]_0^k$. We can translate the labeling $\ell$ into a binary encoded one $\ell' \colon V(G) \rightarrow \{0,1\}^{kc'\log n}$ where $c' = c + 1$ and a block of $c' \log n$ bits represents a number in $[n^c]_0$. For all $u, v \in V(G)$ it holds that $(u,v) \in E(G) \Leftrightarrow \mathcal{N}_{n^c},(\ell(u),\ell(v)) \models \varphi \Leftrightarrow (\ell'(u),\ell'(v),\mathrm{bin}(n^c))$ is a positive instance of the bounded model checking problem for $\varphi$. This almost gives us a labeling scheme which shows that $\gr{S}$ is in $\gccex$. The problem, however, is that $\mathrm{bin}(n^c)$ is not part of the labeling and    
    the formal definition of $\gccex$ does not allow us to use any input but the labels. To solve this we append $\mathrm{bin}(n^c)$ to the labeling. Consider a labeling $\ell'' \colon V(G) \rightarrow \{0,1\}^{(k+1)c'\log n}$ such that $\ell''(u) = \ell'(u)\mathrm{bin}(n^c)$. The last $c'\log n$ bits of $\ell''(u)$ encode $n^c$. 
    Now, the following labeling scheme $S'=(F,(k+1)c')$ shows that $\gr{S}$ is in $\gccex$. The label decoder is defined as
    $(x_1 \dots x_{k+1},y_1 \dots y_{k+1}) \in F$ iff $(x_1,\dots,x_k,y_1,\dots,y_k,x_{k+1})$ is a positive instance of the bounded model checking problem for $\varphi$
    for all $x_i,y_i \in \{0,1\}^{c' m}$, $i \in [k+1]$ and $m \in \N$. The label decoder $F$ can be decided in $\ccex$ since it is closed under $\AC^0$ many-one reductions and thus we can transform the input of the label decoder into an instance of the bounded model checking problem. 
\end{proof}

\begin{theorem}
    $\ccg\FO_{\mathrm{qf}}(\ltp,\addp) \subsetneq \gccac$, $\gccfo \subsetneq \gcctc$ and $\gccfoq \subseteq \gccph$.
\end{theorem}
\begin{proof}
    Due to the previous lemma we can prove these inclusions by showing that the bounded model checking problem for formulas in $\FO_{\mathrm{qf}}(\ltp,\addp)$, $\FO_{\mathrm{qf}}$ and $\FO$ can be solved in $\AC^0$, $\TC^0$ and $\PH$, respectively. First, let us explain why the bounded model checking problem for formulas in $\FO_{\mathrm{qf}}$ can be solved in $\TC^0$. The naive approach to model-check a  quantifier-free formula $\varphi$ is as follows: evaluate the terms of $\varphi$ (expressions involving addition and multiplication of the free variables (the input)), then evaluate the atomic formulas which means comparing numbers and finally compute the underlying boolean function of $\varphi$. Since the order relation, addition and multiplication can be computed in $\TC^0$ (see \cite{vollmer}) and $\varphi$ is fixed this naive approach can be realized by a family of $\TC^0$-circuits (for addition and multiplication one has to additionally handle the overflow condition). Since order and addition can be computed in $\AC^0$ it follows that the bounded model checking problem for formulas in $\FO_{\mathrm{qf}}(\ltp,\addp)$ is in $\AC^0$. In the case of formulas with quantifiers we can use the non-determinism of $\PH$ to evaluate them. Observe that the number of bits that need to be guessed are only linear w.r.t.~the input size.  Therefore we can conclude that $\ccg\FO_{\mathrm{qf}}(\ltp,\addp) \subseteq \gccac$, $\gccfo \subseteq \gcctc$ and $\gccfoq \subseteq \gccph$. The strictness of the first two inclusions is a consequence of the fact that $\gccac \not\subseteq \cogcsmallh$ (see Theorem \ref{thm:acsmallh}) and $\gccfo \subseteq \cogcsmallh$ which is proved later (see Theorem~\ref{thm:pbssmallh} and Corollary~\ref{corol:pbs}).      
\end{proof}

\begin{lemma}
	$\ccg\FO(\sigma)$ and $\ccg\FO_{\mathrm{qf}}(\sigma)$ are closed under $\bfreduction$ for all $\sigma \subseteq \{<,+,\times \}$.
	\label{lem:gfobf}
\end{lemma}
\begin{proof}
    Let $\mathcal{C} \bfreduction \mathcal{D}$ via a $k$-ary boolean function $f$ and $\mathcal{D}$ is in $\ccg\FO_{(\mathrm{qf})}(\sigma)$ via a logical labeling scheme $S=(\varphi,c)$ with $2l$ variables. 
    We construct a logical labeling scheme $S' = (\psi,c)$ with $2kl$ variables which shows that $\mathcal{C}$ is in $\ccg\FO_{(\mathrm{qf})}(\sigma)$. Let $\psi$ have variables $x_{i,j},y_{i,j}$ for $i \in [k]$ and $j \in [l]$ and let us write $\vec{x_i}$ for $x_{i,1},\dots,x_{i,l}$. Let $\psi(\vec{x_1},\dots,\vec{x_k},\vec{y_1},\dots,\vec{y_k})$ be defined as $f( \varphi(\vec{x_1},\vec{y_1}), \dots , \varphi(\vec{x_k},\vec{y_k}) )$. We claim that $\mathcal{C} \subseteq \gr{S'}$. Consider a graph $G \in \mathcal{C}$ with vertex set $V$. There exist $k$ graphs $H_1,\dots,H_k\in \mathcal{D}$ with vertex set $V$ such that $G = f(H_1,\dots,H_k)$. Since $H_i$ is in $\mathcal{D}$ it is also in $\gr{S}$ via a labeling $\ell_i \colon V \rightarrow [n^c]_0^l$ for $i \in [k]$. Consider the labeling $\ell \colon V \rightarrow [n^c]_0^{kl}$ with $\ell(u) = (\ell_1(u),\dots,\ell_k(u))$ for all $u \in V$. It is easy to verify that $G$ is in $\gr {S'}$ via $\ell$. 
\end{proof}

\begin{lemma}
    $\ccg\FO(\sigma)$ and $\ccg\FO_{\mathrm{qf}}(\sigma)$ are closed under $\sgreduction$ for all $\sigma \subseteq \{<,+,\times \}$.
\end{lemma}
\begin{proof}
    Let $\mathcal{C} \sgreduction \mathcal{D}$ via $c,k \in \N$ and a $k^2$-ary boolean function $f$ and $\mathcal{D}$ is in $\ccg\FO_{(\mathrm{qf})}(\sigma)$ via the labeling scheme $S=(\varphi,d)$ and $\varphi$ has $2l$ variables.  Let $\psi$ be a formula with $2kl$ variables $\vec{x}_1,\dots,\vec{x}_k,\vec{y}_1,\dots,\vec{y}_k$ where $\vec{x}_i$ and $\vec{y}_i$ are sequences of $l$ variables. We define $\psi$ as $f(A)$ where $A$ is a $(k \times k)$-matrix with $A_{ij} = \varphi(\vec{x}_i,\vec{y}_j)$ for $i,j \in [k]$. We claim that $\mathcal{C} \subseteq \gr {\psi,cd}$.  Given a graph $G \in \mathcal{C}$ on $n$ vertices there exists a graph $H \in \mathcal{D}$ on $n^c$ vertices such that $G$ has an $(H,f)$-representation via a labeling $\ell_G \colon V(G) \rightarrow V(H)^k$. Also, $H \in \gr {S}$ via a labeling $\ell_H \colon V(H) \rightarrow {[n^{cd}]}_0^l$. 
    Let $\ell \colon V(G) \rightarrow {[n^{cd}]}^{kl}_0$ be defined as follows. Given $u \in V(G)$ let $\ell_G(u) = (u_1,\dots,u_k)$ and $\ell_H(u_i) = (u_{i,1},\dots,u_{i,l})$ for $i \in [k]$. We define $\ell(u)$ as $(u_{1,1},\dots,u_{1,l},u_{2,1},\dots,u_{2,l},\dots,u_{k,1},\dots,u_{k,l})$.
    It can be verified that $G$ is in $\gr {\psi,cd}$ via the labeling $\ell$. No new atoms or quantifiers are introduced in $\psi$ compared to $\varphi$ and thus it remains  in the same class of formulas.
\end{proof}

In a logical labeling scheme $S=(\varphi,c)$ the label length $c$ determines the size of the universe ($\mathcal{N}_{n^c}$) which is used to interpret $\varphi$ for graphs on $n$ vertices. For quantifier-free formulas  the universe size only affects at what point the overflow condition of addition and multiplication applies. Since $\varphi$ is known a priori one can always choose $c$ sufficiently large in order to ensure that no overflow occurs for all labelings of a predetermined size. Is it possible to exploit the overflow condition to express a graph class that would not be expressible without it? We show that for certain fragments this is not the case. Therefore for these fragments the formula can be assumed to be always interpreted over $\mathcal{N}$ irregardless of the number of vertices of the graph.

\begin{definition}
    Given a logical labeling scheme $S=(\varphi,c)$. A graph $G$ is in $\gr[\infty]{S}$ if there exists a labeling $\ell \colon V(G) \rightarrow {[n^c]}_0^k$ such that $(u,v) \in E(G) \Leftrightarrow \mathcal{N},(\ell(u),\ell(v)) \models \varphi$ for all $u, v \in V(G)$.   
\end{definition}

\begin{lemma}
	Let $\sigma = \emptyset$, or $\sigma \subseteq \{\ltp,\addp,\mulp\}$ and `$\ltp$' is in $\sigma$.
    A graph class $\mathcal{C}$ is in $\ccg\FO_{\mathrm{qf}}(\sigma)$ iff there exists a logical labeling scheme $S$ in $\ccg\FO_{\mathrm{qf}}(\sigma)$ such that $\mathcal{C} \subseteq \gr[\infty]{S}$. 
    \label{lem:inftyinterpretation}
\end{lemma}
\begin{proof}
In the case that $\sigma = \emptyset$ it is easy to check that for every logical labeling scheme $S$ in $\ccg\FO_{\mathrm{qf}}(=)$ it holds that $\gr{S} = \gr[\infty]{S}$ and thus the above claim holds. Therefore let us consider the case where `$\ltp$' is in $\sigma$.   	
	
``$\Rightarrow$'': Let $\mathcal{C}$ be a graph class that is in $\ccg\FO_{\mathrm{qf}}(\sigma)$ via a logical labeling scheme $S=(\varphi,c)$, i.e.~$\mathcal{C} \subseteq \gr{S}$. We construct a logical labeling scheme $S'=(\psi,c)$ such that $\mathcal{C} \subseteq \gr[\infty]{S'}$ and $S'$ is in $\ccg\FO_{\mathrm{qf}}(\sigma)$. We assume w.l.o.g.~that we have access to the constants $c_0=0$ and $c_1=n^c$ in $\psi$. 
The constants can be realized by adding two variables to each vertex which are promised to contain the value of the constants for the considered labelings. 
We build $\psi$ from $\varphi$ such that the overflow checks are incorporated into the propositional part of $\psi$. To do this we replace each atomic subformula $A$ of $\varphi$ by a guarded one $A'$. 
We demonstrate how to do this based on the following example. Let $A(x_1,x_2,y_1,y_2)$  be the atomic formula $\mulp(\addp(x_1,y_2),x_2) < \addp(x_2,y_1)$. We convert $A$ into $A'$ by 
checking whether an overflow occurs at each subterm bottom-up. Let $a \rightarrow b$ denote propositional implication which is shorthand for $\neg a \vee b$. 
Then $A'$ is the following formula (order of operation is implied by indentation and reading a propositional formula of the form $\varphi \rightarrow \alpha \wedge \neg \varphi \rightarrow \beta$ as ``if $\varphi$ then $\alpha$ else $\beta$'').
\newcommand{\formtab}[1]{\hspace{#1cm}}
\begin{align*}
  &					c_1 < +(x_1,y_2) \rightarrow
\\&\formtab{1} 			c_1 < \times(c_0,x_2) \rightarrow
\\&\formtab{2} 				c_1 < +(x_2,y_1) \rightarrow
\\&\formtab{3} 					c_0 < c_0 
\\&\formtab{2} 				\wedge \neg c_1 < +(x_2,y_1) \rightarrow
\\&\formtab{3}	 				c_0 < +(x_1,y_1) 
\\&\formtab{1}			\wedge \neg c_1 < \times(c_0,x_2) \rightarrow
\\&\formtab{2}				c_1 < +(x_2,y_1) \rightarrow
\\&\formtab{3}					\times(c_0,x_2) < c_0 
\\&\formtab{2}				\wedge \neg c_1 < +(x_2,y_1) \rightarrow
\\&\formtab{3}					\times(c_0,x_2) < +(x_2,y_1) 
\\&					\wedge \neg c_1 < +(x_1,y_2) \rightarrow
\\&\formtab{1} 			c_1 < \times(+(x_1,y_2),x_2) \rightarrow  
\\&\formtab{2} 			\vdots
\end{align*}
The correctness of this transformation follows from showing that $\mathcal{N}_{n^c}, \vec{a} \models A$ iff $\mathcal{N}, \vec{a} \models A'$ for all $\vec{a} \in [n^c]_0^4$. 

``$\Leftarrow$'': Let $\mathcal{C}$ be a graph class and $S=(\varphi,c)$ is a logical labeling scheme in $\ccg\FO_{\mathrm{qf}}(\sigma)$ such that $\mathcal{C} \subseteq \gr[\infty]{S}$. The maximal value that results from evaluating any term in $\varphi$ must be polynomially bounded, i.e.~there exists a $d \in \N$ such that the largest value produced while evaluating $\varphi$ for a graph with $n$ vertices does not exceed $n^{cd}$. Therefore $\gr[\infty]{\varphi,c} \subseteq \gr{\varphi,cd}$ and $\mathcal{C} \in  \ccg\FO_{\mathrm{qf}}(\sigma)$.
\end{proof}

\begin{fact}
    Let $\sigma = \emptyset$, or $\sigma \subseteq \{\ltp,\addp,\mulp\}$ and `$\ltp$' is in $\sigma$.
    $\ccg\FO(\sigma)$ and $\ccg\FO_{\mathrm{qf}}(\sigma)$ are closed under union.
    \label{fact:gfounion}
\end{fact}
\begin{proof}
    First, we argue why showing closure under union for $\ccg\FO_{(\mathrm{qf})}(\sigma)$ reduces to proving that ($\star$) for every labeling scheme $S=(\varphi,c)$ in $\ccg\FO_{(\mathrm{qf})}(\sigma)$ there exists a labeling scheme $S'=(\varphi',c+1)$ with $\gr{S} \subseteq \gr{S'}$. Let $\mathcal{C},\mathcal{D}$ be in $\ccg\FO_{(\mathrm{qf})}(\sigma)$ via labeling schemes $S_1 = (\varphi_1,c_1)$ and $S_2=(\varphi_2,c_2)$. Due to ($\star$) we can assume w.l.o.g.~that $c_1 = c_2 = c$. Let $2k_i$ be the number of free variables of $\varphi_i$ for $i \in [2]$.  Furthermore, we assume w.l.o.g.~that $\gr{S_1}$ and $\gr{S_2}$ contain all empty graphs $\overline{K_n}$ on $n$ vertices.    
    Let $S = (\psi,c)$ with $\psi(\vec{x_1},\vec{x_2},\vec{y_1},\vec{y_2}) \triangleq \varphi_1(\vec{x_1},\vec{y_1}) \vee \varphi_2(\vec{x_2},\vec{y_2})$ where $\vec{x_i},\vec{y_i}$ are sequences of $k_i$ variables for $i \in [2]$. It holds that $S$ is in  $\ccg\FO_{\mathrm{qf}}(\sigma)$. To see that $\mathcal{C} \cup \mathcal{D} \subseteq \gr{S}$ holds consider a graph $G$ on $n$ vertices in $\mathcal{C}$. We can combine a labeling which shows that $G$ is in $\gr{S_1}$ with a labeling that shows that $\overline{K_n}$ is in $\gr{S_2}$ to get a labeling which shows that $G$ is in $\gr{S}$ because $G \vee \overline{K_n} = G$. The correctness relies on the fact that the labeling schemes $S,S_1,S_2$ are all interpreted over the same universe $\mathcal{N}_{n^c}$ for all graphs with $n$ vertices and $n \in \N$. 
    
    Next, let us explain why $(\star)$ holds. If $\sigma = \emptyset$ then  $(\star)$ obviously holds for $\gccfoeqqf$. Since $\gccfoeqqf = \gccfoeq$ as we shall see later (Fact~\ref{fact:foeqqf}) this also applies to $\gccfoeq$. Therefore we consider the case where `$\ltp$' is in $\sigma$. Let $S=(\varphi,c)$ be a labeling scheme with $2k$ variables in $\ccg\FO_{(\mathrm{qf})}(\sigma)$. 
    We assume that $\varphi$ is in prenex normal form and it has $q \geq 0$ quantified variables, i.e.~$\varphi(\vec{x}) \triangleq Q_1 z_1 \dots Q_q z_q \psi(\vec{x},\vec{z})$ with $Q_i \in \{ \forall,\exists \}$ for $i \in [q]$ and $\vec{z} = (z_1,\dots,z_q)$ and $\psi$ is a quantifier-free formula.
    Let $\psi'$ be the formula that is obtained from $\psi$ by incorporating the overflow checks into the propositional part (the same construction that is used in the proof of Lemma~\ref{lem:inftyinterpretation}). Let $\varphi'(\vec{x})$ be a FO formula which is equivalent to $(Q_1 z_1 \leq c_1) \dots (Q_q z_q \leq c_1) : \psi(\vec{x},\vec{z})$ where $c_1$ is a constant representing the value $n^c$. It holds that $\gr{S} \subseteq \gr{\varphi',c+1}$. 
\end{proof}

\begin{theorem}[Algebraic Interpretation]
	Let $\sigma = \emptyset$, or $\sigma \subseteq \{\ltp,\addp,\mulp\}$ and `$\ltp$' is in $\sigma$.
	A graph class $\mathcal{C}$ is in $\ccg\FO_{\mathrm{qf}}(\sigma)$ iff there exist an $a \in \N$, atomic labeling schemes $S_1,\dots,S_a$ in $\ccg\FO_{\mathrm{qf}}(\sigma)$ and an $a$-ary boolean function $f$ such that $\mathcal{C} \subseteq f(\gr[\infty]{S_1},\dots, \gr[\infty]{S_a})$.
	\label{thm:gfo_algebra}
\end{theorem}
\begin{proof}
	``$\Rightarrow$'': Let $\mathcal{C}$ be in $\ccg\FO_{\mathrm{qf}}(\sigma)$. Due to Lemma \ref{lem:inftyinterpretation} there exists a a logical labeling scheme $S=(\varphi,c)$ in  $\ccg\FO_{\mathrm{qf}}(\sigma)$ such that $\mathcal{C} \subseteq \gr[\infty]{S}$.
    Let $A_1,\dots,A_a$ be the atoms of $\varphi$ and $f$ is the underlying $a$-ary boolean function of $\varphi$. 
	Let $\varphi$ have $2ak$ variables $x_{i,j},y_{i,j}$ with $i \in [a]$ and $j \in [k]$. Furthermore, w.l.o.g.~let the set of variables used in $A_i$ be a subset of $\set{ x_{i,j},y_{i,j} }{ j \in [k]}$ for $i \in [a]$. This means that the variables that occur in $A_i$ and $A_j$ are disjoint for all $i \neq j \in [a]$.
	We claim that $\mathcal{C} \subseteq f(\gr[\infty]{A_1,c},\dots,\gr[\infty]{A_a,c})$. 
	For a graph $G \in \mathcal{C}$ there exist labelings $\ell_i \colon V(G) \rightarrow {[n^c]}_0^{k}$ for each $i \in [a]$ such that 
	$$ (u,v) \in E(G) \Leftrightarrow f(x_1,\dots,x_a) = 1 \text { with } x_i := \llbracket \mathcal{N}, (\ell_i(u),\ell_i(v)) \models A_i \rrbracket  $$
	for all $u, v \in V(G)$. Let $H_i$ be the graph with the same vertex set as $G$ and there is an edge $(u,v) \in E(H_i)$ iff $\mathcal{N}, (\ell_i(u),\ell_i(v)) \models A_i$ for all $i \in [a]$. It holds that $G = f(H_1,\dots,H_a)$ and $H_i \in \gr[\infty]{A_i,c}$ via $\ell_i$.
	
	``$\Leftarrow$'': Since $\ccg\FO_{\mathrm{qf}}(\sigma)$ is closed under union (Fact \ref{fact:gfounion}) it holds that $\mathcal{D} = \bigcup\limits_{i=1}^a \gr[\infty]{S_i}$ is in $\ccg\FO_{\mathrm{qf}}(\sigma)$. Additionally, $\mathcal{C} \bfreduction \mathcal{D}$ via the $a$-ary boolean function $f$  because $\mathcal{C} \subseteq f(\mathcal{D},\dots,\mathcal{D})$. It follows that $\mathcal{C}$ is in $\ccg\FO_{\mathrm{qf}}(\sigma)$ due to closure under $\bfreduction$. 	
\end{proof}

\begin{theorem}
    $\gccfoltqf = \ccg \FO_{\mathrm{qf}}(\ltp,\addp)  = \ccg \FO_{\mathrm{qf}}(\ltp,\mulp) $.
    \label{thm:adduseless}
\end{theorem}
\begin{proof}
    To prove that $\ccg \FO_{\mathrm{qf}}(\ltp,\alpha)$ is a subset of $\gccfoltqf$ for $\alpha \in \{\addp,\mulp\}$ we argue that it suffices to show that for every atomic labeling scheme $S$ in $\ccg \FO_{\mathrm{qf}}(\ltp,\alpha)$ it holds that $\gr[\infty]{S} \in \gccfoltqf$. Given a graph class $\mathcal{C} \in \ccg \FO_{\mathrm{qf}}(\ltp,\alpha)$, there exist atomic labeling schemes $S_1,\dots,S_a$ in $\ccg \FO_{\mathrm{qf}}(\ltp,\alpha)$ and an $a$-ary boolean function $f$ such that $\mathcal{C} \subseteq f(\gr[\infty]{S_1},\dots, \gr[\infty]{S_a})$ because of Theorem~\ref{thm:gfo_algebra}. By assumption it holds that $\gr[\infty]{S_1},\dots, \gr[\infty]{S_a}$ are in $\gccfoltqf$ and therefore $\mathcal{D} = \bigcup_{i=1}^k \gr[\infty]{S_i}$ is in $\gccfoltqf$  due to closure under union. Then $\mathcal{C} \bfreduction \mathcal{D}$ via $f$ and due to closure under $\bfreduction$ it follows that $\mathcal{C} \in \gccfoltqf$.

    Let $S = (\varphi,c)$ be an atomic labeling scheme in $\ccg \FO_{\mathrm{qf}}(\ltp,\alpha)$. We argue that $\gr[\infty]{S}$ is in $\gccfoltqf$.
    Using $\gr[\infty]{S}$ instead of $\gr{S}$ allows us to assume that addition and multiplication are associative.  
     Let $\varphi$ have variables $x_1,\dots,x_k,y_1,\dots,y_k$. The idea is to rearrange the (in)equation such that the variables $x_1,\dots,x_k$ are on one side of the (in)equation and $y_1,\dots,y_k$ are on the other side. This allows us to precompute the required values in the labeling of the new labeling scheme which does not use $\alpha$. Let us show how this works in detail when $\alpha$ is `$\addp$' and $\varphi$ uses `$\ltp$'. In that case $\varphi$ is a linear inequation and can be written as
    $$ \sum\limits_{i=1}^k a_i x_i + b_i y_i < \sum\limits_{i=1}^k c_i x_i + d_i y_i $$
    for certain $a_i,b_i,c_i,d_i \in \N_0$ for $i \in [k]$. This can be rewritten as:
    $$ \underbrace{\sum\limits_{i=1}^k (a_i - c_i) x_i}_{l_n(x_1,\dots,x_k)} < \underbrace{\sum\limits_{i=1}^k (d_i - b_i) y_i}_{r_n(y_1,\dots,y_k)} $$    
    For $n \in \N$ let $l_n,r_n$ be the functions induced by the left-hand and right-hand expression with signature $l_n,r_n \colon {[n^c]}_0^k \rightarrow \mathbb{R}$.  
    Let $E_n$ be the union of the image of $l_n$ and the image of $r_n$.      
    Let $E_n = \{e_1,\dots,e_{z_n}\}$ for some $z_n \in \N$ and $e_i < e_j$ for all $i < j$ with $i,j \in [z_n]$. 
    It holds for all $n \in \N$, $\vec{a},\vec{b} \in [n^c]_0^k$ and $e_i = l_n(\vec{a}), e_j = r_n(\vec{b})$ that
    $$\mathcal{N},({\vec{a},\vec{b}}) \models \varphi \Leftrightarrow l_n(\vec{a}) < r_n(\vec{b}) \Leftrightarrow  e_i < e_j  \Leftrightarrow i < j  $$   
    We claim that for the labeling scheme $S'=(\psi,c')$ where $\psi(x_1,x_2,y_1,y_2)$ is $x_1 < y_2$ and $c' \in \N$ is chosen sufficiently large, it holds that $\gr[\infty]{S} \subseteq \gr[\infty]{S'}$. Consider a graph $G$ on $n$ vertices that is in $\gr[\infty]{S}$ via a labeling $\ell \colon V(G) \rightarrow {[n^c]}_0^k$. We construct a labeling $\ell' \colon V(G) \rightarrow {[n^{c'}]}_0^2$ which shows that $G$ is in $\gr[\infty]{S'}$. For $u \in V(G)$ let $\ell'(u) = (i,j)$ with $e_i = l_n(\ell(u))$ and $e_j = r_n(\ell(v))$. For all $u, v \in V(G)$ it holds that    
    \begin{align*}
        (u,v) \in E(G) & \Leftrightarrow \mathcal{N},(\ell(u),\ell(v)) \models \varphi 
        \\ & \Leftrightarrow l_n(\ell(u)) < r_n(\ell(v))  
        \\ & \Leftrightarrow \ell'(u)_1 < \ell'(v)_2
        \\ & \Leftrightarrow \mathcal{N},(\ell'(u),\ell'(v)) \models \psi 
    \end{align*}
\end{proof}

\begin{fact}
   $\ccg\FO_{\mathrm{qf}}(=) = \gccfoeq$.  
   \label{fact:foeqqf}
\end{fact}
\begin{proof}
    Let $\mathcal{C}$ be in $\gccfoeq$ via a labeling scheme $S=(\varphi,c)$ with $2k$ variables, i.e.~$\mathcal{C} \subseteq \gr{S}$. Observe that $\mathcal{N}_n, \vec{a} \models \varphi$ iff $\mathcal{N}, \vec{a} \models \varphi$ for all $n > r$ and $\vec{a} \in [n]_0^{2k}$ where $r$ denotes the number of free and quantified variables in $\varphi$. 
    This means $\mathcal{C}_{> r} \subseteq \gr[\infty]{S}$.  
    Let $\psi$ be a quantifier-free formula in $\FO_{2k}(=)$ such that $\mathcal{N},\vec{a} \models \varphi$ iff $\mathcal{N},\vec{a} \models \psi$ for all $\vec{a} \in \N_0^{2k}$. The existence of $\psi$ can be proved by quantifier elimination. It follows that $\mathcal{C}_{> r}  \subseteq \gr[\infty]{\psi,c}$ and therefore $\mathcal{C}_{> r} \in \gccfoeqqf$. Since $\gccfoeqqf$ is closed under union and every finite graph class is in $\gccfoeqqf$ it follows that  $\mathcal{C}$ is in $\gccfoeqqf$.    
\end{proof}

\begin{fact}
    $\ccg\FO_{\mathrm{qf}}(<) = \gccfolt$.  
\end{fact}
\begin{proof}
    Observe that $\FO(\ltp)$ has no quantifier-elimination in the sense that there is no quantifier-free formula in $\FO(\ltp)$ which is equivalent to $\exists z \: x < z \wedge z < y$ where $x,y$ are free variables. Instead, we show that $(\star)$ for every formula $\varphi$ in $\FO_k(\ltp)$ there exists a quantifier-free formula $\psi$ in $\FO_k(\ltp,\addp)$ such that $\varphi$ and $\psi$ are equivalent w.r.t.~$\mathcal{N}_n$ for all $n \in \N$. It immediately follows that $\gccfolt \subseteq \ccg\FO_{\mathrm{qf}}(\ltp,\addp)$. Since $\gccfoltqf = \ccg \FO_{\mathrm{qf}}(\ltp,\addp)$ (Theorem~\ref{thm:adduseless}) it holds that $\gccfolt = \gccfoltqf$.
    
    Now, let us argue why $(\star)$ holds. 
    For every formula in $\FO(\ltp)$ it can be assumed w.l.o.g.~that it contains no negation since $\neg x = y$ is equivalent to $x < y \vee y < x$ and $\neg x < y$ is equivalent to $x = y \vee y < x$. To prove that every formula in $\FO(\ltp)$ has a quantifier-free equivalent in $\FO(\ltp,\addp)$ it suffices to show that every formula $\varphi$ of the form $\exists z \: C$  where $C$ is a conjunction of atoms from $\FO(\ltp)$ has a quantifier-free equivalent $\psi$ in $\FO(\ltp,\addp)$ (see \cite[p.~310]{smor}). We assume that $\psi$ can use the constants $c_0,c_1,c_m$ which represent 0, 1 and the maximal value in the universe, respectively.
    If $C$ is unsatisfiable then a quantifier-free equivalent of $\varphi$ is the negation of a tautology. Therefore we assume that $C$ is satisfiable. The conjunctive clause $C$ can be seen as a  directed acyclic graph $D_C$. The equality atoms in $C$ induce a partition of the variables in $C$; let the vertex set of $D_C$ be that partition. For two vertices $U,V$ in $D_C$ there is an edge $(U,V)$  if there exist variables $x \in U, y \in V$ such that $x < y$ is a literal in $C$. Let $Z$ be the vertex of $D_C$ which contains the quantified variable $z$. Assume that $Z$ contains another variable $x\neq z$. In that case a quantifier-free equivalent $\psi$ of $\varphi$ can be obtained by renaming every occurrence of $z$ in $C$ to $x$ and removing the quantifier. If $Z$ contains only $z$ we can proceed as follows. We assume that $z$ occurs in at least one literal of $C$ since otherwise it could be trivially removed. This implies that $Z$ is not an isolated vertex in $D_C$. If $Z$ has in-degree zero then $z$ can be replaced by the constant $c_0$. Similarly, if $Z$ has out-degree zero then $z$ can be replaced by the constant $c_m$. If $Z$ has neither in-degree nor out-degree zero then $\psi$ can be constructed from $\varphi$ as follows. For all in-neighbors $X$ of $Z$, out-neighbors $Y$ of $Z$ and variables $x \in X, y \in Y$ append `$\wedge \: x + c_1 < y \wedge x \neq c_m$' to $\psi$. Then remove  the quantifier and every atom containing $z$ from $\psi$. The atom $x + c_1 < y$ ensures that the difference between $x$ and $y$ is at least two, which was previously expressed by saying that there exists a value $z$ between $x$ and $y$. A problem occurs when $x + c_1$ evaluates to zero because $x$ is assigned the maximal value of the universe due to the overflow condition. To prevent this we add the atom $x \neq c_m$.     
    More formally, it can be checked that $\mathcal{N}_n, \vec{a} \models \varphi$ iff $\mathcal{N}_n, \vec{a} \models \psi$ for all $n \in \N$.     
\end{proof}

We remark that quantifier-free labeling schemes in $\gccfolt$ are solely determined by their formula in the following sense. Given such a formula $\varphi$ with $2k$ variables it holds that $\gr{\varphi,k} = \cup_{i \in \N} \gr{\varphi,i}$ \cite[Lem.~20]{chandoo}. 

\subsection{Complete Graph Classes}

\begin{corollary}
    Let $\sigma = \emptyset$, or $\sigma \subseteq \{\ltp,\addp,\mulp\}$ and `$\ltp$' is in $\sigma$.
    A graph class $\mathcal{D}$ is $\bfreduction$-complete for $\ccg \FO_{\mathrm{qf}}(\sigma)$ iff $\mathcal{D}$ is in $\ccg \FO_{\mathrm{qf}}(\sigma)$ and    
    $\gr[\infty]{S} \bfreduction \mathcal{D}$ holds for all atomic labeling schemes $S$ in $\ccg \FO_{\mathrm{qf}}(\sigma)$.
    \label{corol:complete}
\end{corollary}
\begin{proof}
    ``$\Rightarrow$'': If $\mathcal{D}$ is $\bfreduction$-complete for $\ccg \FO_{\mathrm{qf}}(\sigma)$ then every graph class in $\ccg \FO_{\mathrm{qf}}(\sigma)$ is $\bfreduction$-reducible to $\mathcal{D}$. From Lemma \ref{lem:inftyinterpretation} it follows that $\gr[\infty]{S}$ is in $\ccg \FO_{\mathrm{qf}}(\sigma)$ for every atomic labeling scheme $S$ in $\ccg \FO_{\mathrm{qf}}(\sigma)$.
    
    ``$\Leftarrow$'': Suppose that  $\gr[\infty]{S} \bfreduction \mathcal{D}$ holds for all atomic labeling schemes $S$ in $\ccg \FO_{\mathrm{qf}}(\sigma)$. From Theorem \ref{thm:gfo_algebra} it follows that if a graph class $\mathcal{C}$ is in  $\ccg \FO_{\mathrm{qf}}(\sigma)$ then there exist atomic labeling schemes $S_1,\dots,S_a$ and an $a$-ary boolean function $f$ in  $\ccg \FO_{\mathrm{qf}}(\sigma)$ such that $\mathcal{C} \subseteq f(\gr[\infty]{S_1},\dots,\gr[\infty]{S_a})$. There exist boolean functions $g_1,\dots,g_a$ such that $\gr[\infty]{S_i} \subseteq g_i(\mathcal{D},\dots,\mathcal{D})$ for all $i \in [a]$. Therefore $\mathcal{C} \subseteq f(g_1(\mathcal{D},\dots,\mathcal{D}),\dots,g_a(\mathcal{D},\dots,\mathcal{D}))$ and thus $\mathcal{C} \bfreduction \mathcal{D}$. 
\end{proof}

\begin{definition}
    A directed graph $G$ is dichotomic if for all $u, v \in V(G)$ and $\alpha \in \{\mathrm{in},\mathrm{out}\}$ it holds that $N_{\alpha}(u) \cap N_{\alpha}(v) = \emptyset$ or $N_{\alpha}(u) = N_{\alpha}(v)$.    
\end{definition}

Observe that every directed forest is dichotomic. Every vertex in a forest has in-degree at most one and therefore $N_{\mathrm{in}}(u) = N_{\mathrm{in}}(v)$ or $N_{\mathrm{in}}(u) \cap N_{\mathrm{in}}(v) = \emptyset$ for all $u, v \in V(G)$. Additionally, the out-neighborhoods of every distinct pair of vertices are disjoint because every node has a unique parent.

\begin{lemma}
    There exists an atomic labeling scheme $S$ in $\gccfoeq$ such that $\gr{S}$ is exactly the class of dichotomic graphs. 
    \label{lem:dichchar}    
\end{lemma}
\begin{proof}
    Let $S=(\varphi,1)$ with $\varphi(x_1,x_2,y_1,y_2) \triangleq x_1 = y_2$.   
     
    First, we argue that every dichotomic graph is in $\gr{S}$. Given a dichotomic graph $G$ with $n$ vertices. Let $\sim$ be the equivalence relation on $V(G)$ such that $u \sim v$ if $u$ and $v$ have identical out-neighborhoods. Let $V_1,\dots,V_k$ be the equivalence classes of $\sim$.
    We write $V'_i$ to denote the out-neighbors of the vertices in $V_i$ for $i \in [k]$. 
    Let $V'_0$ be the set of vertices which have in-degree zero.      
    It holds that $V'_0,V'_1,\dots,V'_k$ is a partition of $V(G)$ (with possibly some empty sets) since $G$ is dichotomic. 
    The following labeling $\ell \colon V(G) \rightarrow [n]_0^2$ shows that $G$ is in $\gr{S}$. For $u \in V(G)$ let $\ell(u) = (u_1,u_2)$ with $u_1,u_2 \in [k]_0$ such that $u \in V_{u_1}$ and $u \in V'_{u_2}$. Since $k \leq n$ this is a valid labeling. 
    
    For the other direction let $G$ be a graph that is in $\gr{S}$ via the labeling $\ell  \colon V(G) \rightarrow [n]_0^2$. Consider two vertices $u,v$ of $G$.    
    Let $\ell(u) = (u_1,u_2)$ and $\ell(v) = (v_1,v_2)$. If $u_1 = v_1$ then they have identical out-neighborhoods. If $u_1 \neq v_1$ then they have disjoint out-neighborhoods. The same applies to the in-neighborhoods and $u_2,v_2$. Therefore $G$ is dichotomic.
\end{proof}

\begin{theorem}
    Dichotomic graphs are $\bfreduction$-complete for $\gccfoeq$.
    \label{thm:dgbfcompl}
\end{theorem}
\begin{proof}
    From the previous lemma it follows that dichotomic graphs are in $\gccfoeq$. For the hardness we have to argue that for every atomic labeling scheme $S$ in $\gccfoeq$ it holds that $\gr[\infty]{S}$ is $\bfreduction$-reducible to dichotomic graphs (see Corollary \ref{corol:complete}). Let $S=(\varphi,c)$ be an atomic labeling scheme in $\gccfoeq$ with $2k$ variables $x_1,\dots,x_k,y_1,\dots,y_k$ for some $c,k \in \N$. The formula $\varphi$ must be one of the following:
    \begin{enumerate}
        \item $x_a = x_b$ for some $a,b \in [k]$
		\item $y_a = y_b$ for some $a,b \in [k]$    	
        \item $x_a = y_b$ for some $a,b \in [k]$
    \end{enumerate}
	It is simple to check that every graph in $\gr{S}$ is dichotomic for the first two cases. It remains to deal with the third case. Given a graph $G$ that is in $\gr{S}$ via a labeling $\ell \colon V(G) \rightarrow [n^c]_0^k$. We construct a labeling $\ell' \colon V(G) \rightarrow [n]_0^2$ such that $(u,v) \in E(G)$ iff $\ell'(u)_1 = \ell'(v)_2$ for all $u, v \in V(G)$. A graph is dichotomic iff it has such a labeling $\ell'$	(see the proof of Lemma \ref{lem:dichchar}).
	Let $V(G) = \{v_1,\dots,v_n\}$ and $\ell(v_i) = (v_i^1,\dots,v_i^k)$ for $i \in [n]$. Observe that only the $a$-th and $b$-th component of the labeling $\ell$ are relevant because the other components are never considered. For a set $Z \subseteq \N$ and $z \in Z$ let $\mathrm{ord}(z,Z)$ denote the number of elements smaller than $z$ in $Z$ plus one, e.g.~$\mathrm{ord}(0,\{0,3,4\}) = 1$.
	Let $A = \{v_1^a,\dots,v_n^a \}$. Given $i \in [n]$ we define $\ell'(v_i)$ as $( \mathrm{ord}(v_i^a,A),v_i')$ with $v_i' = 0$ if $v_i^b$ is not in $A$ and $v_i' = \mathrm{ord}(v_i^b,A)$ otherwise. 
     Correctness follows from the fact that $v_i^a = v_j^b$ iff $\ell'(v_i)_1 = \ell'(v_j)_2$ for all $i,j \in [n]$ and that only numbers between $0$ and $n$ are used. 
\end{proof}

\begin{corollary}
    Dichotomic graphs are $\sgreduction$-complete for $\gccfoeq$. 
    \label{corol:dgsgcompl}
\end{corollary}
\begin{proof}
   Lemma~\ref{lem:bftosg} states that for every self-universal and inflatable graph class $\mathcal{D}$ it holds that $\mathcal{C} \bfreduction \mathcal{D}$ implies $\mathcal{C} \sgreduction \mathcal{D}$. 
   Since dichotomic graphs can be characterized as $\gr{S}$ for a (logical) labeling scheme $S$ it follows that they are self-universal and inflatable. 
\end{proof}

\begin{theorem}
    Path graphs are $\sgreduction$-complete for $\gccfoeq$. 
\end{theorem}
\begin{proof}
    Dichotomic graphs are $\sgreduction$-reducible to path graphs via $c=3, k=4$ and $f(A) = a_{1,4} \wedge a_{2,3}$ for $A=(a_{i,j})_{i,j \in [4]}$.
    For $n \in \N$ let $P_n$ be the undirected path graph with $n$ vertices. We assume that $P_n$ has vertex set $\{0,\dots,n-1\}$ and two vertices are adjacent if their absolute difference is one. We need to show that every dichotomic graph $G$ on $n$ vertices has a $(P_{n^3},f)$-representation via some labeling $\ell \colon V(G) \rightarrow [n^3-1]_0^4$. Since $G$ is dichotomic there exists a labeling $\ell' \colon V(G) \rightarrow {[n]}_0^2$ such that $(u,v) \in E(G)$ iff $\ell'(u)_1 = \ell'(v)_2$ for all $u, v \in V(G)$ (see proof of Lemma \ref{lem:dichchar}). 
    For $u \in V(G)$ let $\ell'(u) = (u_1,u_2)$; we define $\ell(u)$ as $(2u_1,2u_1+1,2u_2,2u_2+1)$. The maximal value in the image of $\ell$ is $2n+1$ which is smaller than $n^3-1$ for all $n \geq 2$.     
    For two vertices $u \neq v \in V(G)$ with $\ell'(u)=(u_1,u_2), \ell'(v) = (v_1,v_2)$ it holds that $(u,v) \in E(G)$ iff $u_1 = v_2$ iff $f(A_{uv}^{\ell}) = 1$. Therefore $G$ has a $(P_{n^3},f)$-representation via $\ell$. 
\end{proof}

\begin{definition}
    A directed graph $G$ is a linear neighborhood graph if for all $u, v \in V(G)$ and $\alpha \in \{\mathrm{in},\mathrm{out}\}$ it holds that $N_{\alpha}(u) \subseteq N_{\alpha}(v)$ or $N_{\alpha}(v) \subseteq N_{\alpha}(u)$.
\end{definition}

\begin{lemma}
    There exists an atomic labeling scheme $S$ in $\gccfolt$ such that $\gr{S}$ is exactly the class of linear neighborhood graphs. 
    \label{lem:lngchar}
\end{lemma}
\begin{proof}
    Let $S=(\varphi,1)$ with $\varphi(x_1,x_2,y_1,y_2) \triangleq x_1 < y_2$.
    
    First, we show that every linear neighborhood graph is in $\gr{S}$. Given a linear neighborhood graph $G$ with $n$ vertices. Let $\sim$ be the equivalence relation on $V(G)$ such that $u \sim v$ if $u$ and $v$ have identical in-neighborhoods. 
    Let $V_0$ be the set of vertices with in-degree zero. Let $V_1,\dots,V_k$ be the equivalence classes of $\sim$ except $V_0$ such that $N_{\mathrm{in}}(V_i) \subseteq N_{\mathrm{in}}(V_j)$ for all $1 \leq i < j  \leq k$. Observe that $V_0,\dots,V_k$ is a partition of $V(G)$.    
    The following labeling $\ell \colon V(G) \rightarrow [n]_0^2$ shows that $G$ is in $\gr{S}$. For $u \in V(G)$ let $\ell(u) = (u_1,u_2)$ with $u \in V_{u_2}$ and $u_1$ is the minimal value such that $u \in N_{\mathrm{in}}(V_{u_1+1})$ ($u_1=k$ if this minimum does not exist) for $u_1,u_2 \in \{0,\dots,k\}$. To see that this is correct let us consider an edge $(u,v) \in E(G)$ and $\ell(u) = (u_1,u_2), \ell(v) = (v_1,v_2)$. It holds that $u \in N_{\mathrm{in}}(v) = N_{\mathrm{in}}(V_{v_2})$. Since $u \in N_{\mathrm{in}}(V_{v_2})$ it follows that $u_1 + 1  \leq v_2$ and thus $u_1 < v_2$. Next, consider a non-edge $(u,v) \notin E(G)$. 
    It holds that $u \notin N_{\mathrm{in}}(v) = N_{\mathrm{in}}(V_{v_2})$. Therefore $u_1 + 1 \geq v_2$ and thus $u_1 \not< v_2$. 
    
    For the other direction let $G$ be a graph that is in $\gr{S}$ via a labeling $\ell \colon V(G) \rightarrow [n]_0^2$. We argue that $G$ is a linear neighborhood graph. Given two vertices $u, v \in V(G)$ and $\ell(u) = (u_1,u_2), \ell(v) = (v_1,v_2)$. If $u_1 \leq v_1$ then $N_{\mathrm{out}}(v) \subseteq N_{\mathrm{out}}(u)$. If $u_1 \geq v_1$ then $N_{\mathrm{out}}(u) \subseteq N_{\mathrm{out}}(v)$. The same holds for $u_2,v_2$ and the in-neighborhoods of $u$ and $v$. Therefore $G$ is a linear neighborhood graph.
\end{proof}

\begin{theorem}
    Linear neighborhood graphs are $\bfreduction$-complete for $\gccfolt$.
    \label{thm:lngfolt}
\end{theorem}
\begin{proof}
    Membership follows from the previous lemma. For the hardness we have to show for every atomic labeling scheme $S=(\varphi,c)$ in $\gccfolt$ that $\gr{S}$ is $\bfreduction$-reducible to linear neighborhood graphs. 
    If $\varphi$ uses equality then it can be rewritten using order because $x = y$ iff $\neg (x < y \vee y < x)$.    
     Therefore we assume that $\varphi$ uses order. Let $\varphi$ have variables $x_1,x_2,y_1,y_2$. Using more than two variables per vertex is useless in an atomic labeling scheme without functional symbols as we have seen in the proof of Theorem~\ref{thm:dgbfcompl}. If  $\varphi \triangleq x_i < x_j$ (or $y_i < y_j$) for $i,j \in [2]$ then it is trivial to see that $\gr{S}$ is a subset of linear neighborhood graphs. We assume w.l.o.g.~that $\varphi  \triangleq  x_1 < y_2$.  We show that $\gr{\varphi,c} \subseteq \gr{\varphi,1}$ for all $c \in \N$. Since $\gr{\varphi,1}$ are linear neighborhood graphs this concludes the hardness. Let $G$ be a graph with $n$ vertices that is in $\gr{\varphi,c}$ via a labeling $\ell \colon V(G) \rightarrow [n^c]_0^2$. We argue that there is a labeling $\ell' \colon V(G) \rightarrow [n]_0^2$ which shows that $G$ is in $\gr{\varphi,1}$. Let $V(G) = \{v_1,\dots,v_n\}$ and $\ell(v_i) = (v_i^1,v_i^2)$ for $i \in [n]$. Observe that only the relative order of the labels is relevant. Therefore the labels $v_1^2,\dots,v_n^2$ can be mapped to new labels $\bar{v}_1^2,\dots,\bar{v}_n^2 \subseteq \{1,\dots,n\}$ such that order is preserved, i.e.~$v_i^2 < v_j^2$ iff $\bar{v}_i^2 < \bar{v}_j^2$ for all $i,j \in [n]$. Similarly, the labels  $v_1^1,\dots,v_n^1$ can be mapped to new labels $\bar{v}_1^1,\dots,\bar{v}_n^1 \subseteq \{0,1,\dots,n\}$ such that $v_i^1 < v_j^2$ iff $\bar{v}_i^1 < \bar{v}_j^2$ for all $i,j \in [n]$.    
\end{proof}

\begin{corollary}
   Linear neighborhood graphs are $\sgreduction$-complete for $\gccfolt$. 
\end{corollary}
\begin{proof}
    Same argument as in the proof of Corollary \ref{corol:dgsgcompl}.
\end{proof}

\begin{theorem}
    The transitive closure of directed paths is $\sgreduction$-complete for $\gccfolt$.
\end{theorem}
\begin{proof}
Let $D_n$ denote the transitive closure of the directed path on $n$ vertices. Let us assume that $D_n$ has $\{0,\dots,n-1\}$ as vertex set and $(u,v) \in E(D_n)$ if $u < v$. It is clear from the definition that this graph class is in $\gccfolt$.

We show that linear neighborhood graphs are $\sgreduction$-reducible to this class via $c=2$, $k=2$ and $f(A) = a_{1,2} $ for $A=(a_{i,j})_{i,j \in [2]}$. Let $G$ be a linear neighborhood graph with $n$ vertices. Then there exists a labeling $\ell \colon V(G) \rightarrow [n]_0^2$ such that $(u,v) \in E(G)$ iff $\ell(u)_1 < \ell(v)_2$. It holds that $G$ has a $(D_{n^2},f)$-representation via the same labeling $\ell$.   
\end{proof}

\begin{theorem}
    Interval graphs are $\sgreduction$-complete for $\gccfolt$. 
\end{theorem}
\begin{proof}
    The transitive closure of directed paths is $\sgreduction$-reducible to interval graphs via $c=2$, $k=2$, $f(A)= a_{2,1} \wedge \neg a_{1,2}$ for $A=(a_{i,j})_{i,j \in [2]}$. We show that for all $n \in \N$ there exists an interval graph $H$ on $n^2$ vertices such that $D_n$ has an $(H,f)$-representation.
    Let $\mathcal{I}$ denote the set of intervals on the real line and $V(D_n) = \{0,1,\dots,n-1\}$. 
    The following function $\ell \colon V(D_n) \rightarrow \mathcal{I}^2$ is a labeling for $D_n$. For $u \in  V(D_n)$ let $\ell(u) = ([0,u],[u,u])$. The image of $\ell$ defines an interval graph $H'$ with $2n \leq n^2$ vertices. Let $H$ be an interval graph with $n^2$ vertices which contains $H'$ as induced subgraph. 
    For two vertices $u \neq v \in V(D_n)$ it holds that
    $$ (u,v) \in E(D_n) \Leftrightarrow u < v \Leftrightarrow [u,u] \cap [0,v] \neq \emptyset \wedge [0,u] \cap [v,v] = \emptyset \Leftrightarrow f(A_{uv}^\ell) = 1$$
    and therefore $D_n$ has an $(H,f)$-representation via $\ell$. 
\end{proof}

\subsection{Polynomial-Boolean Systems}
In the beginning, we defined a labeling scheme independently of a model of computation. The label decoder was just a binary relation over words. In the case of logical labeling schemes we neglected this separation by identifying label decoders with logical formulas. It would have been more hygienic to say that a logical formula $\varphi$ with $2k$ variables computes (or represents) a label decoder $F_\varphi \subseteq \N_0^{2k}$. 
However, a subtle difference between logical labeling schemes and classical ones is that the label length $c$ in a logical labeling scheme also influences how the formula is interpreted  whereas in classical labeling schemes the label length does not affect how an algorithm which computes the label decoder is executed. In the case of quantifier-free logical labeling schemes this dependence can be removed as we have shown in Lemma~\ref{lem:inftyinterpretation}. In this section we consider a generalization of $\gccfo$ where the restriction on the label length is dropped. A quantifier-free logical labeling scheme can be seen as a boolean combination of polynomial inequations. We formalize this by what we call polynomial-boolean systems. We consider a polynomial to be an expression over a set of variables that only involves addition and multiplication. We also consider the constant zero to be a polynomial.

\begin{definition}
    A polynomial-boolean system (PBS) with $k$ variables is a tuple $R=((p_1,\dots,p_l),f)$ where $p_1,\dots,p_l$ are polynomials with $k$ variables  and $f$ is an $l^2$-ary boolean function and $k,l \in \N$.
    Given $\mathbb{X} \in \{\N_0,\mathbb{Q},\mathbb{R} \}$ the PBS $R$ induces a $k$-ary relation  $F_R^\mathbb{X}$ over $\mathbb{X}$ which is defined as
    $$(a_1,\dots,a_k) \in F_R^{\mathbb{X}} \Leftrightarrow f(x_{1,1},\dots,x_{l,l}) = 1
    \text{ with } x_{i,j} = \llbracket p_i(a_1,\dots,a_k) < p_j(a_1,\dots,a_k) \rrbracket
    $$      
    for all $a_1,\dots,a_k \in \mathbb{X}$ and $i,j \in [l]$.
\end{definition}

\begin{definition}
    Let $\mathbb{X} \in \{\N_0, \mathbb{Q},  \mathbb{R} \}$.
    For $k \in \N$ and a relation $F \subseteq \mathbb{X}^{2k}$ let $\gr{F}$ be the following set of graphs. A graph $G$ is in $\gr{F}$ if there exists a labeling $\ell \colon V(G) \rightarrow \mathbb{X}^{k}$ such that $(u,v) \in E(G) \Leftrightarrow (\ell(u),\ell(v)) \in F$ for all $u, v \in V(G)$.     
    A graph class $\mathcal{C}$ is in $\gccpbs(\mathbb{X})$ if there exists a PBS $R$ such that $\mathcal{C} \subseteq \gr{F_R^{\mathbb{X}}}$. 
\end{definition}

\begin{fact}
	$k$d-line segment graphs, $k$-ball graphs and $k$-dot product graphs are in $\gccpbs(\mathbb{Q})$ for all $k \in \N$. 
    \label{fact:expbs}
\end{fact}
\begin{proof}
    It is intuitively clear from the definitions of these graph classes that they lie in $\gccpbs(\mathbb{R})$. For example, in a line segment graph each vertex can be assigned four real numbers which represent the two endpoints of the line segment of that vertex. It remains to verify that a boolean combination of the results of polynomial inequations suffices to determine whether two line segments intersect.           
    To see that these graph classes are in $\gccpbs(\mathbb{Q})$ we make the following observation. If a graph class $\mathcal{C}$ is in $\gccpbs(\mathbb{R})$ via a PBS $R$ and for every graph $G$ in $\mathcal{C}$ there exists a rational labeling $\ell$ of $G$  that shows that $G$ is in $\gr{F_R^{\mathbb{R}}}$ then $\mathcal{C}$ is in $\gccpbs(\mathbb{Q})$ via $R$. For $k$-dot product graphs it is shown in \cite[Proposition 3]{fiduccia} that rational labelings suffice. For $k$d-line segment graphs and $k$-ball graphs a perturbation argument shows that rational coordinates and radii suffice as well.
\end{proof}

For $k \geq 2$ it is unknown whether the graph classes mentioned in the previous fact even have a labeling scheme. 

\begin{lemma}
    Let $\mathbb{X} \in \{\N_0, \mathbb{Q},  \mathbb{R} \}$. $\gccpbs(\mathbb{X})$ is closed under $\bfreduction$ and $\sgreduction$.
\end{lemma}
\begin{proof}
    For $\bfreduction$ it suffices to check that $\gccpbs(\mathbb{X})$ is closed under susbets, negation and conjunction due to Corollary \ref{corol:bfcompl}.    
    For $\sgreduction$ consider the following argument. Let $\mathcal{C} \sgreduction \mathcal{D}$ via $c,k \in \N$ and a $k^2$-ary boolean function $f$. Let $\mathcal{D} \in \gccpbs(\mathbb{X})$ via a PBS $R=((p_1,\dots,p_l),g)$ with $2m$ variables and $l,m \in \N$. To avoid technical clutter we just outline what a PBS $R'$ must look like such that $\mathcal{C} \subseteq \gr{F_{R'}^{\mathbb{X}}}$ (and thus $\mathcal{C} \in \gccpbs(\mathbb{X})$).    
    For a graph $G \in \mathcal{C}$ with $n$ vertices there exists a graph $H \in \mathcal{D}$ with $n^c$ vertices such that $G$ has an $(H,f)$-representation via a labeling $\ell_G \colon V(G) \rightarrow V(H)^k$. Let $\ell_H \colon V(H) \rightarrow \mathbb{X}^m$ be a labeling which shows that $H$ is in $\gr{F_R^{\mathbb{X}}}$. By combining $\ell_G$ and $\ell_H$ we get a labeling $\ell \colon V(G) \rightarrow \mathbb{X}^{km}$. Intuitively, the labeling $\ell$ provides us with all the information required to determine adjacency in $G$. More specifically, one can construct a PBS $R'$ with $2km$ variables from $R$ and $f$ such that $G \in \gr{F_{R'}^{\mathbb{X}}}$ via $\ell$.
\end{proof}

\begin{theorem}
    $\gccpbs(\mathbb{N}_0) = \gccpbs(\mathbb{Q})$.
    \label{thm:gnq}
\end{theorem}
\begin{proof}
    It is clear that $\gccpbs(\mathbb{N}_0) \subseteq \gccpbs(\mathbb{Q})$.
    Let $\mathbb{Q}_+ = \set{n \in \mathbb{Q}}{n \geq 0}$.
    For the other direction we show that $\gccpbs(\mathbb{Q}) \subseteq \gccpbs(\mathbb{Q}_+)$ and $\gccpbs(\mathbb{Q}_+) \subseteq \gccpbs(\mathbb{N}_0)$.
    
    Let $\mathcal{C} \in \gccpbs(\mathbb{Q})$ via a PBS $R=((p_1,\dots,p_l),f)$ with $2k$ variables. We outline a PBS $R'$ such that $\mathcal{C}$ is in $\gccpbs(\mathbb{Q}_+)$ via $R'$. This construction relies on the following observation. 
    Given $a \in \mathbb{Q}$ let $|a|$ denote its absolute value and $\sign(a)$ equals $-1$ if $a$ is negative and $1$ otherwise.       
    For $n \in \N$ and  a vector $\vec{a} \in \mathbb{Q}^n$ let $|\vec{a}| = (|a_1|,\dots,|a_n|)$ and $\sign(\vec{a}) = (\sign(a_1),\dots,\sign(a_n))$.
    For all polynomial functions $p,q \colon \mathbb{Q}^n \rightarrow \mathbb{Q}$ and sign patterns $\vec{s} \in \{-1,1\}^n$ there exist polynomial functions $p',q' \colon \mathbb{Q}_+^n \rightarrow \mathbb{Q}_+$ such that for all $\vec{a} \in \mathbb{Q}^n$ with $\sign(\vec{a}) = \vec{s}$ it holds that $p(\vec{a}) < q(\vec{a})$ iff $p'(|\vec{a}|) < q'(|\vec{a}|)$.     
     For example, consider the polynomials $p(x,y,z) = x^2y^3z + y $ and $q(x,y,z) = z$ and the sign pattern $(-1,1,-1)$ for $(x,y,z)$. If we only consider inputs with this sign pattern then it holds that $p(x,y,z) < q(x,y,z)$ iff $\underbrace{|y| + |z|}_{p'} < \underbrace{|x|^2|y|^3|z|}_{q'}$. For each variable in $R$ we have two variables in $R'$. ($\star$) The first one is used to store the absolute value of the original variable and the second one encodes the sign.       
      Let $G$ be a graph that is in $\gr{F_R^{\mathbb{Q}}}$ via a labeling $\ell \colon V(G) \rightarrow \mathbb{Q}^k$. Then we derive a labeling $\ell' \colon V(G) \rightarrow \mathbb{Q}_+^{2k}$  from $\ell$ as follows. Given $u \in V(G)$ let $\ell(u)=(u_1,\dots,u_k)$. We set $\ell'(u) = (|u_1|,u_1',\dots,|u_k|,u_k')$ where $u_i' = |u_i|$ if $u_i$ is negative and any other non-negative value if $u_i$ is positive. This allows us to infer the sign pattern and absolute values of the original labeling $\ell$ from $\ell'$.      
      The PBS $R'$ is constructed such that $G \in \gr{F_{R'}^{\mathbb{Q}_+}}$ via $\ell'$. Suppose we are given two vertices $u, v \in V(G)$. Then the adjacency of $u$ and $v$ depends on the results of $p_i(\ell(u),\ell(v)) < p_j(\ell(u),\ell(v))$ for $i,j \in [l]$. We emulate the inequation $p_i(\ell(u),\ell(v)) < p_j(\ell(u),\ell(v))$ in $R'$ by  $p'(|\ell(u)|,|\ell(v)|) < q'(|\ell(u)|,|\ell(v)|)$ where $p'$ and $q'$ depend on $p_i,p_j$ and the sign pattern of $\ell(u),\ell(v)$. This means for every pair $i,j \in [l]$ and every sign pattern $s \in \{-1,1\}^{2k}$ there is a pair of polynomials in $R'$ and additionally $R'$ has the $2k$ identity polynomials to decode the signs.       
    
    To see that $\gccpbs(\mathbb{Q}_+) \subseteq \gccpbs(\mathbb{N}_0)$ it suffices to make the following observation. Given two polynomial functions $p,q \colon \mathbb{Q}_+^k \rightarrow \mathbb{Q}_+$ there exist two polynomial functions $p',q' \colon \N_0^{2k} \rightarrow \N_0$ such that for all $\vec{a} = (\frac{a_1}{b_1},\dots,\frac{a_k}{b_k}) \in \mathbb{Q}_+^k$ it holds that $p(\vec{a}) < q(\vec{a})$ iff $p'(a_1,b_1,\dots,a_k,b_k) < q'(a_1,b_1,\dots,a_k,b_k)$. The functions $p'$ and $q'$ can be obtained from the inequation $p < q$ by multiplying with the denominators. Therefore a PBS $R$ with $2k$ variables can be translated into a PBS $R'$ with $4k$ variables such that $\gr{F_R^{\mathbb{Q}_+}} \subseteq \gr{F_{R'}^{\N_0}}$.    
\end{proof}

\begin{theorem}
    $\gccpbs(\mathbb{R}) \subseteq \cogcsmallh$.
    \label{thm:pbssmallh}
\end{theorem}
\begin{proof}
    Let $R=((p_1,\dots,p_l),f)$ be a PBS with $2k$ variables. We show that $\gr{F_R^{\mathbb{R}}}$ is small and hereditary. From that it follows that $\gccpbs(\mathbb{R})$ is a subset of $\cogcsmallh$. 
    Let $G$ be a graph that is in $\gr{F_R^{\mathbb{R}}}$ via a labeling $\ell \colon V(G) \rightarrow \mathbb{R}^k$. An induced subgraph of $G$ on vertex set $V' \subseteq V(G)$ is in $\gr{F_R^{\mathbb{R}}}$ via the labeling $\ell$ restricted to $V'$. Thus $\gr{F_R^{\mathbb{R}}}$ is hereditary.
     
     It remains to argue that $\gr{F_R^{\mathbb{R}}}$ is small.  We do so by applying Warren's theorem \cite[p.~55]{spinrad}, which can be stated as follows. Let $\mathcal{E}=(E_1,\dots,E_m)$ be a sequence of polynomial inequations over variables $x_1,\dots,x_n$. More specifically, the inequations are assumed to be of the form $p(x_1,\dots,x_n) < q(x_1,\dots,x_n)$ where $p,q$ are polynomials. Also, let $d$ denote the maximum degree that occurs in any of these inequations. The sequence $\mathcal{E}$ can be understood as a function  from $\mathbb{R}^n$ to $\{0,1\}^m$ in the following sense. Given $\vec{a} \in \mathbb{R}^n$ let $\mathcal{E}(\vec{a}) = (e_1,\dots,e_m)$ with $e_i = 1$ iff the inequation $E_i(\vec{a})$ holds. An element of the image of $\mathcal{E}$ is called a sign pattern.     
     Warren's theorem states that the cardinality of the image of $\mathcal{E}$ (or equivalently, the number of sign patterns of $\mathcal{E}$)  is at most $\left({\frac{cdm}{n}}\right)^n$ where $c$ is some constant. 
     
     We show that the number of graphs on $n$ vertices in $\gr{F_R^{\mathbb{R}}}$ is bounded by the number of sign patterns of a certain sequence of equations $\mathcal{E}_{R,n}$.  Consider a graph $G$ on $n$ vertices that is in $\gr{F_R^{\mathbb{R}}}$ via a labeling $\ell \colon V(G) \rightarrow \mathbb{R}^k$. 
     The presence of the edge $(u,v)$ in $G$ is determined by the result of $l^2$ polynomial inequations. 
     Therefore $G$ is determined by the result of a sequence of $l^2n^2$ polynomial inequations. These inequations use  $kn$ variables $x_u^i$ with $u \in V(G)$ and $i \in [k]$. Let $d$ denote the maximum degree over the polynomials $p_1,\dots,p_l$. This means $\mathcal{E}_{R,n}$ has $kn$ variables, $l^2n^2$ equations and maximum degree $d$. Thus a graph on $n$ vertices in $\gr{F_R^{\mathbb{R}}}$ is determined by a sign pattern of $\mathcal{E}_{R,n}$. As a consequence there are at most $\left({\frac{cdl^2n^2}{kn}}\right)^{kn} \in n^{\mathcal{O}(n)}$ graphs on $n$ vertices in $\gr{F_R^{\mathbb{R}}}$ ($c,d,k,l$ are constants). 
\end{proof}

In logical labeling schemes the numbers of the vertex labels have to be polynomially bounded. The following definition adds a size restriction on the labelings to polynomial-boolean systems. Let $F$ be a $2k$-ary relation over $\N_0$ and let $s$ be a function from $\N$ to $\N$. We say
a graph $G$ with $n$ vertices is in $\gr{F,s}$ if there is a labeling $\ell$ of $G$ which shows that $G$ is in $\gr{F}$ and the largest number in the image of $\ell$ does not exceed $s(n)$. In the second definition such a size restriction is applied to relations over $\mathbb{Q}$.

\begin{definition}
    Let $s \colon \N \rightarrow \N$ be a total function, $k \in \N$ and $F \subseteq \N_0^{2k}$. We say a graph $G$ is in $\gr{F,s}$ if there exists a labeling $\ell \colon V(G) \rightarrow [s(n)]_0^k$ such that $(u,v) \in E(G) \Leftrightarrow (\ell(u),\ell(v)) \in F$ for all $u, v \in V(G)$.
    Let $S$ be a set of total functions from $\N$ to $\N$. We say a graph class $\mathcal{C}$ is in $\gccbpbs(\mathbb{N}_0,S)$ if there exists a PBS $R$ and $s \in S$ such that $\mathcal{C} \subseteq \gr{F_R^{\mathbb{N}_0},s}$. 
\end{definition}

\begin{definition}
    For $m \in \N$ let $\mathbb{Q}_m = \set{ \frac{sa}{b} }{ a , b \in [m], s \in \{-1,0,1\} }$. 
    Let $s \colon \N \rightarrow \N$ be a total function, $k \in \N$ and $F \subseteq \mathbb{Q}^{2k}$.
    We say a graph $G$ is in $\gr{F,s}$ if there exists a labeling $\ell \colon V(G) \rightarrow {(\mathbb{Q}_{s(n)})}^k$ such that 
    $(u,v) \in E(G) \Leftrightarrow (\ell(u),\ell(v)) \in F$  for all $u, v \in V(G)$.    
    Let $S$ be a set of total functions from $\N$ to $\N$.
    A graph class $\mathcal{C}$ is in $\gccbpbs(\mathbb{Q},S)$ if there exists a PBS $R$ and $s \in S$ such that $\mathcal{C} \subseteq \gr{F_R^{\mathbb{Q}},s}$.
\end{definition}

\begin{theorem}
	Let $S$ be a set of total functions from $\N$ to $\N$.
    $\gccbpbs(\mathbb{N}_0,S) = \gccbpbs(\mathbb{Q},S)$.
\end{theorem}
\begin{proof}
    Let $\mathcal{C}$ be in $\gccbpbs(\mathbb{Q},S)$ via a PBS $R$ and $s \in S$, i.e.~$\mathcal{C} \subseteq \gr{ F_R^{\mathbb{Q}}, s }$. First, we construct a PBS $R'$ such that $\gr{ F_R^{\mathbb{Q}}, s } \subseteq \gr{ F_{R'}^{\mathbb{Q}_+}, s }$. This is the same construction that is described in the proof of Theorem~\ref{thm:gnq}. One has to additionally check that the size restriction on the labeling is not violated. 
    More precisely, if a graph $G$ on $n$ vertices is in $\gr{ F_R^{\mathbb{Q}}, s }$ via a labeling $\ell \colon V(G) \rightarrow \mathbb{Q}_{s(n)}$ then the labeling $\ell'$ derived from $\ell$ to show that $G$ is in $\gr{ F_{R'}^{\mathbb{Q}_+}, s }$ is only allowed to contain values from ${(\mathbb{Q}_+)}_{s(n)} = \set{ \frac{a}{b} }{ a,b \in [s(n)] } \cup \{0\}$ in its image. This is the case because $\ell'$ only requires the absolute values that occur in $\ell$ and at least two different values to indicate the sign, see $(\star)$ in the proof of Theorem~\ref{thm:gnq}. 
            
    In a second step, we construct a PBS $R''$ such that $\gr{ F_{R'}^{\mathbb{Q}_+}, s } \subseteq \gr{ F_{R''}^{\mathbb{N}_0}, s }$. We apply the same construction as in the proof of Theorem~\ref{thm:gnq}. Let us consider the size restriction on the labeling. Consider a graph $G$ on $n$ vertices which is in $\gr{ F_{R'}^{\mathbb{Q}_+}, s }$ via a labeling $\ell'$. The values that occur in the image of $\ell'$ must be a subset of $(\mathbb{Q}_+)_{s(n)}$. There is a labeling $\ell''$ which shows that $G$ is in $\gr{ F_{R''}^{\mathbb{N}_0}, s }$ since the values in the image of $\ell''$ are the numerators and denominators required to express the values in the image of $\ell'$. Due to the definition of $(\mathbb{Q}_+)_{s(n)}$ these numerators and denominators are all in $[s(n)]_0$. 
\end{proof}

In \cite{mcdiarmid} it is shown that disk and line segment graphs are in $\gccbpbs(\mathbb{N}_0,\exp^2(\mathcal{O}(n)))$. In \cite{kang} it is shown that $k$-dot product graphs are in $\gccbpbs(\mathbb{N}_0,\exp^2(\mathcal{O}(n)))$. 
What is the `smallest' class of functions $S$ such that $\gccbpbs(\mathbb{N}_0,S) = \gccpbs(\mathbb{N}_0)$? The following statement guarantees the existence of such an $S$. 

\begin{fact}
	Let $\mathrm{Tot}$ be the set of all total functions from $\N$ to $\N$. $\gccbpbs(\mathbb{N}_0,\mathrm{Tot}) = \gccpbs(\N_0)$. 
\end{fact}
\begin{proof}
	Let $\mathcal{C}$ be in $\gccpbs(\N_0)$ via a PBS $R$ with $2k$ variables, i.e.~$\mathcal{C} \subseteq \gr{F_R^{\N_0}}$. We argue that there exists a total function $s \colon \N \rightarrow \N$ which depends on $R$ such that $\mathcal{C} \subseteq \gr{F_R^{\N_0},s}$. From that the claim follows. 
	For a graph $G$ in $\gr{F_R^{\N_0}}$ let $\mathcal{L}(G)$ denote the set of labelings $\ell \colon V(G) \rightarrow \N_0^k$ which show that $G$ is in $\gr{F_R^{\N_0}}$. For a labeling $\ell$ let $\Ima^*(\ell) \subseteq \N_0$ be the set of values that occur in the image of $\ell$. Let $x_G$ be defined as $\min\limits_{\ell \in \mathcal{L}(G)} \{ \max( \Ima^*(\ell) ) \}$, i.e.~the largest value that occurs in a `smallest' labeling of $G$. Then $s(n)$ can be defined as $\max \set{ x_G }{ G \text{ has $n$ vertices and is in } \gr{F_R^{\N_0}} }$.     
\end{proof}

\begin{theorem}
    $\gccfo = \gccbpbs(\mathbb{N}_0,n^{\mathcal{O}(1)})$.
    \label{thm:gccfobpbsn}
\end{theorem}
\begin{proof}
    ``$\subseteq$'': Let $\mathcal{C}$ be a graph class in $\gccfo$. From Lemma \ref{lem:inftyinterpretation} it follows that there exists a logical labeling scheme $S=(\varphi,c)$ in $\gccfo$ such that $\mathcal{C} \subseteq \gr[\infty]{S}$. This means the interpretation of each term in $\varphi$ is identical to a polynomial function over $\N_0$. Therefore $S$  directly translates to a PBS $R$ over $\N_0$ where the polynomial functions are given by the terms of $\varphi$ and the boolean function of $R$ is given by the boolean function underlying $\varphi$. It follows that $\gr[\infty]{S} \subseteq \gr{F_R,n^c}$ and thus $\mathcal{C} \in \gccbpbs(\N_0,n^{\mathcal{O}(1)})$. 
    
    ``$\supseteq$'': Let $\mathcal{C}$ be a graph class that is in $\gccbpbs(\N_0,n^{\mathcal{O}(1)})$ via a PBS $R$ with $2k$ variables and $c,k \in \N$. The PBS $R$ can be translated into a quantifier-free formula $\varphi$ with $2k$ variables in a straightforward fashion. It holds that $\gr{F_R^{\N_0},n^c}$ is a subset of $\gr[\infty]{\varphi,c}$  and thus $\mathcal{C} \in \gccfo$. 
\end{proof}

\begin{corollary}    
    $\gccfo  = \gccbpbs(\mathbb{N}_0,n^{\mathcal{O}(1)})  = \gccbpbs(\mathbb{Q},n^{\mathcal{O}(1)}) \subseteq \gccbpbs(\mathbb{N}_0,\mathrm{Tot}) = \gccpbs(\mathbb{N}_0) = \gccpbs(\mathbb{Q}) \subseteq \gccpbs(\mathbb{R}) \subseteq \cogcsmallh$.
    \label{corol:pbs}
\end{corollary}

Next, we show that if $\gccfolt$ can be separated from $\gccpbs(\mathbb{N}_0)$ then this separation can be amplified to $\gccfolt \neq \gccfo$. To prove this we show that if $\gccfo$ has a complete graph class w.r.t~$\bfreduction$ which is hereditary then $\gccpbs(\N_0)$ collapses to $\gccfo$. A similar statement holds w.r.t.~$\sgreduction$-reductions.  

Let us say a set of graph classes $\mathbb{A}$ is closed under hereditary closure if for all $\mathcal{C}$ in $\mathbb{A}$ it holds that its hereditary closure $[\mathcal{C}]_{\subseteq}$ is in $\mathbb{A}$. If a set of graph classes $\mathbb{A}$ is closed under hereditary closure and subsets then a graph class $\mathcal{C}$ is in $\mathbb{A}$ iff $[\mathcal{C}]_{\subseteq}$ is in $\mathbb{A}$. As a consequence it suffices to consider only hereditary graph classes when studying sets of graph classes that are closed under hereditary closure.  

\begin{theorem}
    If $\gccfo$ is closed under hereditary closure then $\gccfo = \gccpbs(\N_0)$.
    \label{thm:gccfohcpbs}
\end{theorem}
\begin{proof}
    Assume that $\gccfo$ is closed under hereditary closure. We show that $\gccpbs(\N_0)$ is a subset of $\gccbpbs(\N_0,n^{\mathcal{O}(1)})$.
    Let $\mathcal{C}$ be a graph class that is in $\gccpbs(\N_0)$ via a PBS $R$ with $2k$ variables, i.e.~$\mathcal{C} \subseteq \gr{F_R^{\N_0}}$. 
    It holds that $\mathcal{D} := \gr{F_R^{\N_0},n}$ is in $\gccbpbs(\N_0,n^{\mathcal{O}(1)})$. 
    We claim that every graph in $\mathcal{C}$ occurs as induced subgraph of some graph in $\mathcal{D}$ and therefore $\mathcal{C} \subseteq [\mathcal{D}]_{\subseteq}$. 
    Since $\gccbpbs(\N_0,n^{\mathcal{O}(1)})=\gccfo$ is closed under hereditary closure by assumption it follows that $ [\mathcal{D}]_{\subseteq}$ (and thus $\mathcal{C}$) is in $\gccfo$. Let $G$ be a graph with $n$ vertices that is in $\mathcal{C}$. There exists a labeling $\ell \colon V(G) \rightarrow \N_0^k$ such that $G$ is in $\gr{F_R^{\N_0}}$ via $\ell$. Let $z \in \N_0$ be the maximal value that occurs in the image of $\ell$. Let $H$ be some graph with $z + n$ vertices and $V(G) \subseteq V(H)$. The labeling $\ell$ is a partial labeling of $H$ which shows that $G$ is an induced subgraph of $H$. If one extends the labeling $\ell$ such that $\ell(u) = (0,\dots,0)$ for all $u \in V(H) \setminus V(G)$ then this shows that $H$ is in $\gr{F_R^{\N_0},n} = \mathcal{D}$.  
\end{proof}

Stated differently, $\gccpbs(\N_0)$ can be characterized as the set of graph classes obtained by closing $\gccfo$ under hereditary closure. 

\begin{lemma}
    Let $\mathbb{A}$ be a set of graph classes closed under $\bfreduction$. If there exists a hereditary graph class that is $\bfreduction$-complete for $\mathbb{A}$ then $\mathbb{A}$ is closed under hereditary closure. 
    \label{lem:hcbfh}
\end{lemma}
\begin{proof}
	Let $\mathcal{C}$ be a hereditary graph class that is $\bfreduction$-complete for $\mathbb{A}$. Let $\mathcal{D}$ be a graph class in $\mathbb{A}$. Since $\mathcal{C}$ is complete for $\mathbb{A}$ there exists a $k$-ary boolean function $f$ such that $\mathcal{D} \subseteq f(\mathcal{C},\dots,\mathcal{C})$. We show that every graph which occurs as induced subgraph of some graph in $\mathcal{D}$ is also in $f(\mathcal{C},\dots,\mathcal{C})$, i.e.~$[\mathcal{D}]_{\subseteq} \subseteq f(\mathcal{C},\dots,\mathcal{C})$. Let $G$ be a graph in $\mathcal{D}$ and let $G'$ be an induced subgraph of $G$ on vertex set $V' \subseteq V(G)$. 
	There exist graphs $H_1,\dots,H_k \in \mathcal{C}$ on vertex set $V(G)$ such that $G = f(H_1,\dots,H_k)$. It follows that $G' = f(H_1',\dots,H_k')$ where $H_i'$ is the induced subgraph of $H_i$ on vertex set $V'$ for $i \in [k]$. Since $\mathcal{C}$ is hereditary it  contains $H_1',\dots,H_k'$ and thus $G' \in f(\mathcal{C},\dots,\mathcal{C})$. Therefore $[\mathcal{D}]_{\subseteq} \subseteq f(\mathcal{C},\dots,\mathcal{C})$. Stated differently, $[\mathcal{D}]_{\subseteq}$ is $\bfreduction$-reducible to $\mathcal{C}$ via $f$ and thus must be in $\mathbb{A}$.   
\end{proof}

\begin{lemma}
    Let $\mathbb{A}$ be a set of graph classes closed under $\sgreduction$. If there exists a hereditary and inflatable graph class that is $\sgreduction$-complete for $\mathbb{A}$ then $\mathbb{A}$ is closed under hereditary closure. 
\end{lemma}
\begin{proof}
    Let $\mathcal{C}$ be a hereditary and inflatable graph class that is $\sgreduction$-complete for $\mathbb{A}$.
    We argue that if a graph class $\mathcal{D}$ is $\sgreduction$-reducible to $\mathcal{C}$ then $[\mathcal{D}]_{\subseteq}$ is also  $\sgreduction$-reducible to $\mathcal{C}$. From that the above statement follows. Due to Lemma \ref{lem:sgwoc} it holds that  $\mathcal{D}$ is $\sgreduction$-reducible to $\mathcal{C}$ iff there exist a $k \in \N$ and a $k^2$-ary boolean function $f$ such that for all graphs $G$ in $\mathcal{D}$ there exists a graph $H$ in $\mathcal{C}$ such that $G$ has an $(H,f)$-representation. Observe that if $G$ has an $(H,f)$-representation  then every induced subgraph of $G$  has an $(H,f)$-representation as well. Therefore $[\mathcal{D}]_{\subseteq}$ is  $\sgreduction$-reducible to $\mathcal{C}$ via $f$. 
\end{proof}

\begin{corollary}
    If $\gccfo$ has a $\bfreduction$-complete graph class that is hereditary then $\gccfo = \gccpbs(\N_0)$.  If $\gccfo$ has an $\sgreduction$-complete graph class that is hereditary and inflatable then $\gccfo = \gccpbs(\N_0)$.  
    \label{corol:pbscol}
\end{corollary}

Since $\gccfolt$ has a hereditary graph class which is $\bfreduction$-complete, namely linear neighborhood graphs, it follows that $\gccfolt$ must be a strict subset of $\gccfo$ unless $\gccfolt = \gccpbs(\mathbb{N}_0)$.

\section{Algorithmic Properties}
\label{sec:algo}
Consider the following question: is deciding the existence of a Hamiltonian cycle W[1]-hard when parameterized by $\gccp$? This seems to be an ill-defined question because unlike, for example, tree-width the class $\gccp$ does not resemble a parameter at all. We argue that the incapability of recognizing this as a  well-defined question  is caused by a flawed understanding of what constitutes a parameterization in parameterized complexity. In fact, classes such as $\gccfoeq, \gccac, \gccp$ and $\gccr$ are all sensible parameterizations. After substantiating this claim, we use upper and lower bounds from the literature to identify  algorithmic research questions where the parameterization is a class of labeling schemes.

A parameter $\kappa$ in parameterized complexity is a total function which maps words over some alphabet $\Sigma$ to natural numbers, i.e.~$\kappa \colon \Sigma^* \rightarrow \N$. A parameterized problem is a tuple $(L,\kappa)$ where $L$ is a language and $\kappa$ is a parameter, both over the same alphabet. 
This formalization of parameterized problems is used in the introductory textbook \cite{flumgrohe}\footnote{They make the additional requirement that a parameter must be polynomial-time computable, which we shall ignore here.}. 
A graph parameter is a total function mapping unlabeled graphs to natural numbers and therefore can be regarded as special case of a parameter. For a parameter $\kappa$ and $c \in \N$ let $\kappa_c$ denote the set of words $w$ with $\kappa(w) \leq c$. For example, $\text{tree-width}_c$ is the set of graphs with tree-width at most $c$.

To compare two parameters $\kappa, \tau$ over the same alphabet the following notion of boundedness is used: $\kappa$ is upper bounded by $\tau$, in symbols $\kappa \preccurlyeq \tau$, if there exists a function $f \colon \N \rightarrow \N$ such that $\kappa(w) \leq f(\tau(w))$ holds for all words $w$. If $\kappa \preccurlyeq \tau$ holds then $(L,\tau)$ is fpt-reducible to $(L,\kappa)$ for every language $L$. The precise definition of fpt-reducibility is not relevant here. It suffices to know that it can be seen as an analogon of polynomial-time many-one reducibility in classical complexity. Intuitively, designing a good algorithm for a problem using $\tau$ as parameter is not harder than doing this for $\kappa$. 
If $\kappa \preccurlyeq \tau$ and $\tau \preccurlyeq \kappa$ holds we say that $\kappa$ and $\tau$ are equivalent. From the previous implication it follows that $(L,\kappa)$ and $(L,\tau)$ are fpt-equivalent for all languages $L$ whenever $\kappa$ and $\tau$ are equivalent. As a consequence, it does not make a difference whether $\kappa$ or $\tau$ is considered when assessing the complexity of a parameterized problem and thus it is more accurate to define a parameterized problem as a tuple $(L,\mathbb{K})$ where $\mathbb{K}$ denotes an equivalence class of parameters. In that sense parameterized complexity is no more about parameters than graph theory is about adjacency matrices. However, while it is self-evident that adjacency matrices represent graphs it is not so obvious what is represented by parameters. A different notion of boundedness helps to answer this question. Let us say a language $L$ is bounded by a parameter $\kappa$ if there exists a $c \in \N$ such that $L \subseteq \kappa_c$. Let us write $\mathbb{K}(\kappa)$ to denote the set of languages that are bounded by $\kappa$. We say $\kappa$ is a subset of $\tau$, in symbols $\kappa \subseteq \tau$, if $\mathbb{K}(\kappa) \subseteq \mathbb{K}(\tau)$. Stated differently, every language that is bounded by $\kappa$ is bounded by $\tau$ as well. For example, the maximum degree is a subset of the clique number but not vice versa. In fact, the `$\subseteq$'-relation is just the inverse relation of `$\preccurlyeq$'.

\begin{fact}        
    Let $\kappa,\tau$ be parameters over the same alphabet. It holds that $\kappa \preccurlyeq \tau$ iff $\tau \subseteq \kappa$. 
\end{fact}
\begin{proof}        
    ``$\Rightarrow$'': Let $f \colon \N \rightarrow \N$ be a monotone function such that $\kappa(w) \leq f(\tau(w))$ for all words $w$. We show inductively that for every $i \in \N$ it holds that $\tau_i \subseteq \kappa_{f(i)}$. For the base case $i = 1$ it holds that $w \in \tau_1$ iff $\tau(w) = 1$. It follows that $\kappa(w) \leq f(1)$ and therefore $w \in \kappa_{f(1)}$. For the inductive step $i \rightarrow i + 1$ it must be the case that $w$ is either in $\tau_{i+1} \setminus \tau_i$ or in $\tau_i$. If $w$ is in $\tau_i$ then by induction hypothesis it holds that $w \in \kappa_{f(i)}$. Since $f$ is monotone it follows that $w \in \kappa_{f(i+1)}$ as well. For the other case it holds that $\tau(w) = i+1$ and therefore $\kappa(w) \leq f(i+1)$ which means $w \in \kappa_{f(i+1)}$.      
    
    ``$\Leftarrow$'': Since $\tau \subseteq \kappa$ there exists a function $f \colon \N \rightarrow \N$ such that $\tau_i \subseteq \kappa_{f(i)}$ for all $i \in \N$. We show that $\kappa(w) \leq f(\tau(w))$ for all words $w$. Let $\tau(w) = k$ for some $k \in \N$. Then it holds that $w \in \tau_k$ and therefore $w \in \kappa_{f(k)}$ as well. This means $\kappa(w) \leq f(k) = f(\tau(w))$.
\end{proof}

\begin{corollary}
    Two parameters are equivalent iff they bound the same set of languages.
\end{corollary}

Therefore the answer to the previous question is that a parameter represents a set of languages. However, not every set of languages can be interpreted as a parameter. We are only interested in sets of languages that can be represented by a parameter. To distinguish between such sets of languages and parameters let us call the former ones parameterizations. 

\begin{definition}
    A set of languages $\mathbb{K}$ over an alphabet $\Sigma$ is a parameterization if there exists a parameter $\kappa$ over $\Sigma$ such that $\mathbb{K} = \mathbb{K}(\kappa)$.
\end{definition}    

\begin{theorem}
    A set of languages $\mathbb{K}$ over $\Sigma$ is a parameterization iff the following holds:
    \begin{enumerate}
        \item $\mathbb{K}$ is closed under union
        \item $\mathbb{K}$ contains $\{w\}$ for every word $w$ over $\Sigma$        
        \item there exists a countable subset $\mathbb{K}'$ of $\mathbb{K}$ such that the closure of $\mathbb{K}'$ under subsets equals $\mathbb{K}$
    \end{enumerate}
    \label{thm:param}
\end{theorem}
\begin{proof}    
    ``$\Rightarrow$'': Let $\mathbb{K}$ be a parameterization over $\Sigma$. This means there exists a parameter $\kappa$ over $\Sigma$ such that $\mathbb{K} = \mathbb{K}(\kappa)$. Clearly, $\mathbb{K}(\kappa)$ is closed under subsets. Let $L,L'$ be languages over $\Sigma$ which are both bounded by $\kappa$. This means $L \subseteq \kappa_i$ and $L' \subseteq \kappa_j$ for some $i,j \in \N$. We assume w.l.o.g.~that $i \leq j$ and therefore $L \cup L' \subseteq \kappa_j$. Therefore $\mathbb{K}(\kappa)$ is closed under union. Since $\kappa$ is total it follows that $\{w\}$ is in $\mathbb{K}(\kappa)$ for every word $w$. The countable subset $\mathbb{K}'$ of $\mathbb{K}(\kappa)$ such that the closure of $\mathbb{K}'$ under subsets equals $\mathbb{K}(\kappa)$ is given by $\{\kappa_1, \kappa_2, \dots \} $.

    ``$\Leftarrow$'': Let $\mathbb{K}$ be a set of languages over $\Sigma$ which satisfies the above three conditions. Observe that the third condition implies that $\mathbb{K}$ is closed under subsets. We construct a parameter $\kappa$ over $\Sigma$ such that $\mathbb{K} = \mathbb{K}(\kappa)$. Let $\mathbb{K}' = \{L_1,L_2,\dots \}$ be the countable subset of $\mathbb{K}$ whose closure under subsets equals $\mathbb{K}$. Let $L'_c = \bigcup_{i=1}^c L_i$. It holds that $L'_c$ is in $\mathbb{K}$ for every $c \in \N$ because $\mathbb{K}$ is closed under union. Then $\kappa(w)$ being defined as the least $k$ such that $w \in L'_k$ yields the required parameter.    
\end{proof}

Therefore it is more accurate to understand a parameterized problem as a tuple $(L,\mathbb{K})$ where $L$ is a language and $\mathbb{K}$ is a parameterization, both over the same alphabet. This alternative view on parameterized problems leads to an interesting different perspective on parameterized complexity which, however, we do not address here.
The important observation in our context is the following. If we have a set of graph classes $\mathbb{A}$ which satisfies the conditions of Theorem~\ref{thm:param} then asking about the parameterized complexity of some problem parameterized by $\mathbb{A}$ is a well-defined question. We show how Theorem~\ref{thm:param} can be applied to classes such as $\gccp$. 

\begin{lemma}
    For every countable set of languages $\ccex$ that contains all finite languages and for which $\gccex$ is closed under union it holds that $\gccex$ is a parameterization.
\end{lemma}
\begin{proof}
    We show that $\gccex$ satisfies the three conditions of Theorem~\ref{thm:param}. $\gccex$ is closed under union by assumption. 
    Furthermore, every singleton graph class lies in $\gccex$. This can be shown by using a  look-up table of finite size as label decoder. More precisely, for a graph $G$ on $n$ vertices a language $L$ that contains only words of length $2\log n$ can be constructed such that $\{G\} \subseteq \gr{F_L,1}$. Since $\ccex$ contains all finite languages it contains $L$ as well.
    The countable subset of $\gccex$ such that its closure under subsets equals $\gccex$ is given by the set of graph classes $\gr {S}$ for every labeling scheme $S$ in $\gccex$. That this set is countable follows from the fact that there are only countably many labeling schemes $S$ in $\gccex$ since $\ccex$ is countable. 
\end{proof}

\begin{corollary}
    $\gccac, \gccl, \gccp, \gccnp, \gccexp$ and $\gccr$ are parameterizations.
    \label{corol:param}
\end{corollary}

It is not difficult to see that $\gccfoeq$, $\gccfolt$, $\gccfo$ and $\gccfoq$ are parameterizations as well; their closure under union is shown in Fact~\ref{fact:gfounion}. An example of a set of graph classes which is no parameterization is the class $\gccall$. It can be shown that there exists no countable subset of $\gccall$ whose closure under subsets equals $\gccall$ by a diagonalization argument. 

For readers not familiar with parameterized complexity we give a rough description of the complexity classes  mentioned in Figure~\ref{fig:alg_side}. 
A parameterized problem $(L,\mathbb{K})$ is in FPT (fixed-parameter tractable) if there exists a $c \in \N$ such that for all $K \in \mathbb{K}$ it holds that $L$ can be decided in $\TIME(n^c)$ if one only considers inputs from $K$. The multiplicative constant hidden in the big-oh can depend on $K$. 
A parameterized problem $(L,\mathbb{K})$ is in XP if for all $K \in \mathbb{K}$ it holds that $L$ is in $\P$ if one only considers inputs from $K$. In contrast to FPT the degree of the polynomial that bounds the runtime is not fixed but can depend on $K$.
The classes FPT and XP are regarded as the analogon of $\P$ and $\EXP$ in the parameterized world. 
If a parameterized problem is W[1]-hard then this can be seen as evidence that it is not in FPT.  

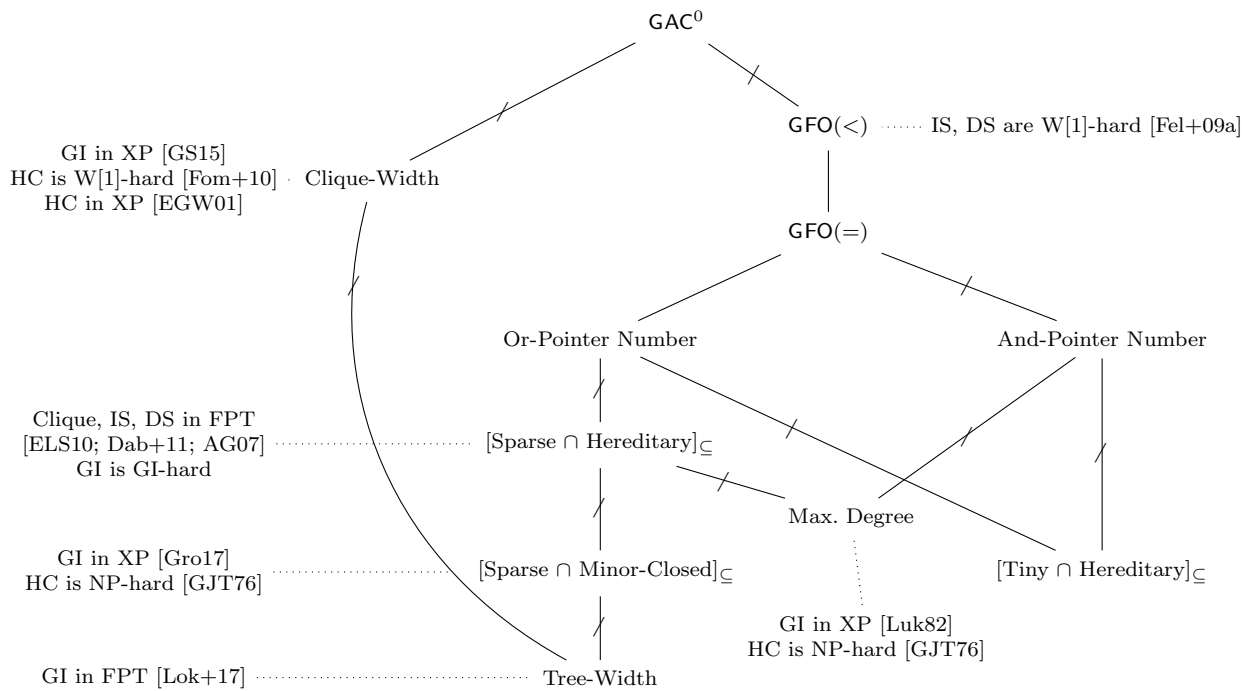
\begin{figure}
    \centering
    \begin{tikzpicture}[shorten >=1pt,auto,node distance=1.2cm,
  main node/.style={circle,draw}]
\usetikzlibrary{shapes.misc}

\newcommand*{\movey}{1.4}%
\newcommand*{\movex}{2}%
\newcommand*{\movexsi}{0.02}%

    \node (gac) at (-1*\movex,-6.5*\movey) {$\gccac$};   

	\node (gcw) at (-3*\movex,-8*\movey) {Clique-Width};    
	\draw[-] (gcw) to  (gac);
	\node (notsign) at (-2.2*\movex,-7.35*\movey) {$\not$};
	\node (notsign) at (-0.55*\movex,-7*\movey) {$\not$};	
	\node (notsign) at (-3.18*\movex,-9*\movey) {$\not$};

    \node[align=center] (gcwinfo) at (-4.5*\movex,-8*\movey) {GI in XP \cite{gicwxp}\\HC is W[1]-hard \cite{hccwhard}\\HC in XP \cite{hccwxp}};
	\draw[-,dotted] (gcwinfo) to (gcw);  

    \node (gfolt) at (0,-7.5*\movey) {$\gccfolt$};    
   	\draw[-] (gfolt) to  (gac);

    \node[align=center] (gfoinfo) at (1.7*\movex,-7.5*\movey) {IS, DS are W[1]-hard \cite{mihard}};
	\draw[-,dotted] (gfoinfo) to (gfolt);

    \node (gfoeq) at (0,-8.5*\movey) {$\gccfoeq$};
    \draw[-] (gfoeq) to  (gfolt);


    \node (gtw) at (-1.5*\movex,-12.5*\movey-0.3) {Tree-Width};     
	\path[-]
		(gtw) edge[bend angle=38,bend left] (gcw);     

    \node[align=center] (gtwinfo) at (-4.5*\movex,-12.5*\movey-0.3) {GI in FPT \cite{gitwfpt}};
    \draw[-,dotted] (gtwinfo) to (gtw);
    

    \node (gsmc) at (-1.5*\movex,-11.5*\movey-0.3) {\: [Sparse $\cap$ Minor-Closed]$_{\subseteq}$};
    \draw[-] (gtw) to node {\(\not\)} (gsmc);
    
    \node[align=center] (gsmcinfo) at (-4.5*\movex,-11.5*\movey-0.3) {GI in XP \cite{gipmcxp} \\ HC is NP-hard \cite{hchard}};
    \draw[-,dotted] (gsmcinfo) to (gsmc);

    \node (gsh) at (-1.5*\movex,-10.5*\movey) {[Sparse $\cap$ Hereditary]$_{\subseteq}$};
    \draw[-] (gsmc) to node {\(\not\)} (gsh);
    
    \node[align=center] (gusinfo) at (-4.5*\movex,-10.5*\movey) {Clique, IS, DS in FPT\\\cite{cliqueusfpt,dabrowski,dsusfpt}\\GI is GI-hard};
	\draw[-,dotted] (gusinfo) to (gsh);

    \node (gdeg) at (0.3,-11.2*\movey) {Max.~Degree};
    \draw[-] (gdeg) to   (gsh);
    \node[]  (shmdst) at (-1.5*\movex+1.5,-10.5*\movey-0.5) {\(\not\)};

	\node[align=center] (gdeginfo) at (0.3+0.2,-12.2*\movey-0.2) {GI in XP \cite{gidegxp} \\ HC is NP-hard \cite{hchard}};
    \draw[-,dotted] (gdeginfo) to (gdeg);

    \node (gopn) at (-1.5*\movex,-9.5*\movey) {Or-Pointer Number};    
    \draw[-] (gsh) to node {\(\not\)} (gopn);

    \draw[-] (gopn)  to  (gfoeq);

    \node (gapn) at (1.8*\movex,-9.5*\movey) {And-Pointer Number};  
    \draw[-] (gapn)  to  (gfoeq);  
    \node[]  (gfoapn) at (1.8*\movex-1.9,-10*\movey+1.4) {\(\not\)};

    \draw[-] (gdeg) to  (gapn);
    \node[]  (apnmdst) at (1.8*\movex-1.9,-10*\movey-0.6) {\(\not\)};

    \node (gth) at (1.8*\movex,-11.5*\movey-0.3) {[Tiny $\cap$ Hereditary]$_{\subseteq}$};
    \draw[-] (gth) to node {\(\not\)} (gapn);


    \draw[-] (gth) to   (gopn);
    \node[]  (opnth) at (-1.5*\movex+2.4,-9.5*\movey-1.2) {\(\not\)};





\end{tikzpicture}
    \caption{Parameterized complexity of graph classes with implicit representations}
    \label{fig:alg_side}                
\end{figure}  

In Figure~\ref{fig:alg_side} we assess the parameterized complexity of graph isomorphism (GI),  Hamiltonian cycle (HC), clique, independent set (IS) and dominating set (DS) for structural parameterizations below $\gccac$. The last three problems are additionally parameterized by the solution size which is part of the input. Let us give a brief explanation of this figure. The problems IS and DS are W[1]-hard for $\gccfolt$ because in \cite{Fellows} these problem are shown to be already W[1]-hard for unit 2-interval graphs which are contained in $\gccfolt$. The statement that GI is GI-hard for sparse and hereditary graph classes means that there is a sparse and hereditary graph class for which GI is GI-hard. An analogous interpretation is meant for the NP-hardness results of HC. It is interesting to note that in \cite{hchard} a stronger result is proven than what is shown in this figure: HC is already NP-hard on planar graphs with degree at most three. This implies that even when parameterizing by proper minor-closed graph classes with bounded degree it is hopeless to find XP-algorithms for HC. On the positive side, we found no results that indicate that clique is not in FPT on $\gccfolt$ and DS or IS are not in FPT on $\gccfoeq$. 
Is the clique problem in FPT on $\gccfolt$? Are IS and DS in FPT on $\gccfoeq$? Is GI in XP when parameterized by the and-pointer number? Instead of $\gccfoeq$ and $\gccfolt$ one can also consider $[\gcforest]_{\clonebf}$ and $[\gcinterval]_{\clonebf}$ as parameterizations for which it might be easier to find efficient algorithms. If unit 2-interval graphs are in $[\gcinterval]_{\clonebf}$ then IS and DS are W[1]-hard on $[\gcinterval]_{\clonebf}$ due to the results from \cite{mihard}. 

\section{Regular Labeling Schemes}
\label{sec:reg}
One of our main objectives is to identify suitable classes of labeling schemes against which lower bounds for hereditary graph classes can be proved. 
Suitable means that such a class should contain most of the graph classes that are known to have a labeling scheme while still possessing enough structure to be amenable to analysis. For the latter condition we consider closure under algebraic and subgraph reductions and being a subset of $\cogcsmallh$ as criteria. We show that a careful definition of labeling schemes in terms of regular languages yields a good candidate. The resulting class which we call $\gccreg$ (definition follows) contains every hereditary graph class that we know to have a labeling scheme, just like $\gccac$. We also show that it is closed under both reduction notions. Unfortunately, we are not able to resolve the question of whether $\gccreg$ is a subset of $\cogcsmallh$.

Clearly, the generic mechanism in Definition~\ref{def:genericgcc} which turns a set of languages into a set of labeling schemes is not adequate for regular languages because the order of the input is crucial. Instead of concatenating the two labels of the vertices they should be interleaved in this case.
For two strings $x,y$ of equal length $n$ we write $x \wr y$ to denote the string $x_1 y_1 x_2 y_2 \dots x_n y_n$.

\begin{definition}
     Let $F \subseteq \{0,1\}^* \times \{0,1\}^*$ be a label decoder. We call $F$ regular if 
     $$ \set{ x \wr y } { n \in \N, x,y \in \{0,1\}^n, (x,y) \in F}$$
     is a regular language.  We call a labeling scheme regular if its label decoder is regular.
\end{definition}

Consider a labeling scheme $S=(F,2c)$ which is defined in terms of a label decoder $F'$ as follows. It holds that $(x_{\mathrm{out}}x_{\mathrm{in}},y_{\mathrm{out}}y_{\mathrm{in}}) \in F$ iff $(x_{\mathrm{out}},y_{\mathrm{in}}) \in F'$ for all $cm \in \N$ and $x_\alpha,y_\alpha \in \{0,1\}^{cm}$ and $\alpha \in \{ {\mathrm{in}}, {\mathrm{out}} \}$. This labeling scheme has the special property that the label of a vertex can be split into two parts of which one is responsible for the outgoing edges and the other for the ingoing edges. Stated differently, a graph $G$ is represented by $S$ iff there exist two labelings $\ell_{\mathrm{in}}, \ell_{\mathrm{out}}  \colon V(G) \rightarrow \{0,1\}^{c \log n}$  such that $(u,v) \in E(G)$ iff $(\ell_{\mathrm{out}}(u),\ell_{\mathrm{in}}(v)) \in F'$ for all $u, v$ in $V(G)$. The same idea can be applied to logical labeling schemes. The logical labeling scheme for dichotomic graphs is an example of this. This trick cannot be applied to regular labeling schemes and therefore we have to externally add this ability in order to get the most out of such labeling schemes.  
 
\begin{definition}
     Let $S=(F,c)$ be a labeling scheme. We say a graph $G$ with $n$ vertices is in $\gr[\mathrm{io}]{S}$ if there exist labelings $\ell_{\mathrm{in}}, \ell_{\mathrm{out}} \colon V(G) \rightarrow \{0,1\}^{c \log n}$ such that for all $u, v \in V(G)$  
     $$ (u,v) \in E(G) \Leftrightarrow (\ell_{\mathrm{out}}(u),\ell_{\mathrm{in}}(v)) \in F$$
\end{definition}

Observe that this definition can be applied to labeling schemes from $\gccac$ or higher without affecting the set of graph classes that are represented due to the previous trick.
 
 \begin{definition}
     A graph class $\mathcal{C}$ is in $\gccreg$ if there exists a regular labeling scheme $S=(F,c)$ such that $\mathcal{C} \subseteq \gr[\mathrm{io}]{S}$.          
 \end{definition}

\begin{lemma}
    $\gccreg$ is closed under $\bfreduction$ and $\sgreduction$.
\end{lemma}
\begin{proof}
     To see that $\gccreg$ is closed under $\bfreduction$ it suffices to check that it is closed under negation and conjunction (Corollary~\ref{corol:bfcompl}). Closure under negation follows from the fact that the complement of a regular label decoder is regular as well. Closure under conjunction can be shown by essentially the same argument that we are about to give for closure under $\sgreduction$. 
    
     To see that $\gccreg$ is closed under $\sgreduction$ reconsider the proof that $\gccac$ is closed under $\sgreduction$ (Lemma~\ref{lem:sgacclosure}). We proceed similarly but have to take the sequential nature of DFAs into account. More concretely, an $x_i$ has to be compared against $y_1,\dots, y_k$ for $i \in [k]$. Since $x_i$ cannot be remembered by the automaton we use multiple occurrences of the $x_i$'s and $y_i$'s which slightly increases the required label length from $cdk$ to $cdk^2$. We first construct a labeling scheme $S'$ which is not necessarily regular and then explain how to convert it into a regular one $S''$.
    
     Let $\mathcal{C} \sgreduction \mathcal{D}$ via $d,k \in \N$ and a $k^2$-ary boolean function $f$ and $\mathcal{D} \in \gccreg$ via $S=(F,c)$. We claim that the following labeling scheme $S'=(F',cdk^2)$ represents $\mathcal{C}$. For $m \in \N, i \in [k]$ and $x_i,y_i \in \{0,1\}^{cdm}$ we define
     
     $$(x_1x_2\dots x_{k^2},y_1y_2 \dots y_{k^2}) \in F' \Leftrightarrow 
     f \left( 
     \begin{matrix}
     z_1 	 &   z_2 & \dots  & z_k \\
     z_{k+1} & \ddots	& &   \\
     \vdots  & & & \\
     z_{k^2-k+1} &  & & z_{k^2}  
     \end{matrix}
     \right) = 1 $$
     with $z_i := \llbracket (x_i,y_i) \in F \rrbracket$ for $i \in [k^2]$.     
     
     Let $G$ be a graph in $\mathcal{C}$ with $n$ vertices. We construct labelings $\ell_{\mathrm{in}},\ell_{\mathrm{out}} \colon V(G) \rightarrow \{0,1\}^{cdk^2\log n}$ which show that $G$ is in $\gr[\mathrm{io}]{S'}$.      
     There exists a graph $H$ in $\mathcal{D}$ with $n^d$ vertices such that $G$ has an $(H,f)$-representation via a labeling $\ell \colon V(G) \rightarrow V(H)^k$. There exist labelings  $\ell_{\mathrm{in}}^H,\ell_{\mathrm{out}}^H \colon V(H) \rightarrow \{0,1\}^{c \log n}$ such that $(u,v) \in E(H) \Leftrightarrow (\ell_{\mathrm{out}}^H(u),\ell_{\mathrm{in}}^H(v)) \in F$ for all $u \neq v \in V(H)$. Let $u \in V(G)$ and $\ell(u) = (u_1,\dots,u_k)$. Let $u_i^\alpha = \ell_{\alpha}^H(u_i)$ for $\alpha \in \{ \mathrm{in}, \mathrm{out} \}$. 
     For a string $s$ and $k \in \N$ let $s^k$ denote the string which is obtained by concatenating $s$ $k$ times. 
     We define $\ell_{\mathrm{in}}(u)$ as $(u_1^{\mathrm{in}} u_2^{\mathrm{in}} \dots u_k^{\mathrm{in}})^k$ and $\ell_{\mathrm{out}}(u)$ as $(u_1^{\mathrm{out}})^k (u_2^{\mathrm{out}})^k \dots (u_k^{\mathrm{out}})^k$. Consider two vertices $u \neq v \in V(G)$. We claim that $(u,v) \in E(G)$ iff $(\ell_{\mathrm{out}}(u),\ell_{\mathrm{in}}(v)) \in F'$. If we plug in the definitions of $\ell_{\mathrm{out}}(u)$ and $\ell_{\mathrm{in}}(v)$ the right-hand side becomes
     $$ \left(     
      \underbrace{u_1^{\mathrm{out}}  \dots u_1^{\mathrm{out}}}_{k \text{ times}} \: \dots \: \underbrace{u_k^{\mathrm{out}}  \dots u_k^{\mathrm{out}}}_{k \text{ times}} \text{\Large , } 
     v_1^{\mathrm{in}} v_2^{\mathrm{in}} \dots v_k^{\mathrm{in}} 
     \: \dots \:
     v_1^{\mathrm{in}} v_2^{\mathrm{in}} \dots v_k^{\mathrm{in}}      
     \right) \in F' $$
     By definition of $F'$ this holds iff 
     $$f \left( 
\begin{matrix}
    \llbracket (u_1^{\mathrm{out}},v_1^{\mathrm{in}}) \in F \rrbracket 	 &   
    \llbracket (u_1^{\mathrm{out}},v_2^{\mathrm{in}}) \in F \rrbracket  & \dots  & 
    \llbracket (u_1^{\mathrm{out}},v_k^{\mathrm{in}}) \in F \rrbracket  \\
    \llbracket (u_2^{\mathrm{out}},v_1^{\mathrm{in}}) \in F \rrbracket  & \ddots	& &   \\
    \vdots  & & & \\
    \llbracket (u_k^{\mathrm{out}},v_1^{\mathrm{in}}) \in F \rrbracket  &  & & 
    \llbracket (u_k^{\mathrm{out}},v_k^{\mathrm{in}}) \in F \rrbracket  
\end{matrix}
\right) = 1 $$
    Furthermore, it holds that $(u_i^{\mathrm{out}},v_j^{\mathrm{in}}) \in F$ iff $(u_i,v_j) \in E(H)$ for all $i,j \in [k]$. Since $G$ has an $(H,f)$-representation via $\ell $ it follows that $(u,v) \in E(G)$ iff $f(A)=1$ with $A_{i,j} = \llbracket (u_i,v_j) \in E(H)  \rrbracket$. Therefore $G$ is in $\gr[\mathrm{io}]{S'}$.
    
    Now, let us consider the complexity of computing $F'$. To decide $(x_1 \dots x_{k^2},y_1 \dots y_{k^2}) \in F'$ we can evaluate the binary decision tree of $f$ where the $i$-th proposition corresponds to $\llbracket (x_i,y_i) \in F \rrbracket$ for $i \in [k^2]$. Since this tree has depth at most $k^2$ a finite number of states suffices to compute this. Also, the truth value of each proposition can be decided by a DFA because $F$ is regular. However, the difficulty is that a DFA does not know when $x_i$ and $y_i$ end and $x_{i+1}$ and $y_{i+1}$ begin. To resolve this one can introduce a special delimiter sign `$\#$' and define a new label decoder $F''$ over the alphabet $\{0,1, \# \}$ as  
    $$ (x_1 \# \dots \# x_{k^2} \# * ,y_1 \# \dots \# y_{k^2} \# * ) \in F'' \Leftrightarrow (x_1 \dots x_{k^2},y_1 \dots y_{k^2})  \in F'$$
    where `$*$' denotes an arbitrary string of a certain length. By choosing an adequate label length $c''$ one obtains a regular labeling scheme $(F'',c'')$ which represents $\mathcal{C}$. This labeling scheme can be re-encoded over the binary alphabet such that it remains regular. 
\end{proof}

\begin{theorem}
$\gccfolt \subseteq \gccreg \subseteq \gccp$.
\end{theorem}
\begin{proof}
    It is clear that every regular label decoder is polynomial-time computable and thus $\gccreg \subseteq \gccp$. To show that $\gccfolt \subseteq \gccreg$ it suffices to argue that the transitive closure of directed path graphs is in $\gccreg$ since this class is $\sgreduction$-complete for $\gccfolt$ and $\gccreg$ is closed under $\sgreduction$. This follows from the fact that the following is a regular label decoder:
    $$ \set{ (x,y) }{ m \in \N, \: x,y \in \{0,1\}^m, \text{ $x$ is lexicographically smaller than $y$} }$$    
\end{proof}
 
\begin{fact}
    Every graph class with bounded clique-width is in $\gccreg$. 
    \label{fact:cwreg}
\end{fact}
\begin{proof}
    Reconsider the labeling scheme described in the proof of Fact~\ref{fact:cwgac}. By using a suitable encoding this can be turned into a regular labeling scheme. Let $G$ be a graph with $n$ vertices and clique-width $k$. Let $T(G)$ be a binary decomposition tree of $G$.    
    A vertex $v$ of $G$ is labeled as follows. We start with an empty label for $v$ and explain how to iteratively construct its label by appending a string each time one moves a node down in $T(G)$.    
    Let $r$ be the root node of $T(G)$ and $S$ is the balanced $k$-module in the left child of $r$ which can be partitioned into $S_1,\dots,S_k$. Let $R$ be the subset of $[k]$ such that vertices in $V(G) \setminus S$ are adjacent to all vertices in $S_i$ for $i \in R$ and non-adjacent to all other vertices in $S$.
    Assume $v$ is in the left child of $r$ and $v \in S_i$ for $i \in [k]$. Let $s_i \in \{0,1\}^k$ be the string which only has a 1 at the $i$-th position. Then we append $0s_i\#$ to the label of $v$. The first bit tells us that $v$ is placed in the left child of the root node. The string $s_i$ encodes in which of the $k$ parts of $S$ $v$ lies. Assume $v$ is in the right child of $r$. Let $s \in \{0,1\}^k$ be the string which has a 1 at the $i$-th position iff $i \in R$ for all $i \in [k]$, i.e.~$s$ encodes the subset $R$. Append $1s\#$ to the label of $v$. Repeat this process until the leaf node of $T(G)$ which contains $v$ is reached. 
    We call a string of the form  `$\{0,1\}^{k+1}\#$' a block. 
    This means the label of a vertex in $G$ consists of $\mathcal{O}(\log n)$ blocks since $T(G)$ has depth $\mathcal{O}(\log n)$. To decode adjacency from two such labels one has to find the first blocks where the first bits differ (this means the vertices are placed in different subtrees) and then check for the remaining $k$ bits of the two blocks whether there is a position $i \in [k]$ such that both blocks have a 1 at that position. This can be computed by a DFA.  
\end{proof}

\begin{lemma}
    $\gccreg$ is closed under union.
    \label{lem:gregunion}
\end{lemma}
\begin{proof}
    Let $\mathcal{C},\mathcal{D} \in \gccreg$ via $S_1=(F_1,c_1)$ and $S_2=(F_2,c_2)$. We construct a regular label decoder  $F$ such that $\mathcal{C} \cup \mathcal{D} \subseteq \gr[\mathrm{io}]{S}$ with $S=(F,c)$ and $c=c_1+c_2+1$. The idea is to construct $F$ such that the first bits of the two labels determine whether $F_1$ or $F_2$ is used.         
    Since $\max \{c_1,c_2 \} \log n + 1 < c\log n$ the labels must contain some dummy bits in order to be of correct length. 
    We will place these dummy bits at the end of the labels and use a special delimiter sign `$\#$' to signal at what point they start. Formally, this means $F$ is a binary relation over $\{0,1,\#\}$. We give an incomplete specification of $F$ which, however, is sufficient to show that $\mathcal{C} \cup  \mathcal{D}$ is represented by $S$. Let $x,y \in \{0,1,\#\}^{cm}$ for $m \in \N$ and $s \in \{0,1\}$. If $x$ and $y$ are of the form $s \{0,1\}^* \#\{0,1\}^*$ and $x',y'$ denote the substrings of $x,y$ which occur between $s$ and $\#$ then
    $(x,y) \in F \Leftrightarrow (x',y') \in F_{s+1}$. It remains to check that $F$ is regular and $S$ indeed represents $\mathcal{C} \cup \mathcal{D}$. It is also a straightforward task to re-encode $F$ over the binary alphabet. 
\end{proof}

\begin{fact}
    $\gccreg$ is a parameterization. 
\end{fact}
\begin{proof}
    We apply Theorem~\ref{thm:param}. $\gccreg$ is closed under union. Also, every singleton graph class is contained in $\gccreg$ because $\gccfolt \subseteq \gccreg$ and $\gccfolt$ already contains every singleton graph class. The countable subset of $\gccreg$ such that its closure under subsets equals $\gccreg$ is given by the set of graph classes $\gr[\mathrm{io}]{S}$ where $S$ is a regular labeling scheme.
\end{proof}

\section{Summary and Open Questions}
\label{sec:final}
In Figure~\ref{fig:bigpicture} an overview of all the sets of graph classes that we have seen is given. First, we summarize the train of thought that motivated us to introduce the various concepts and what we perceive to be their importance in the context of studying the limitations of labeling schemes; this summary does not reflect the order of the paper. Afterwards we point out what we believe to be realistic and meaningful research directions.

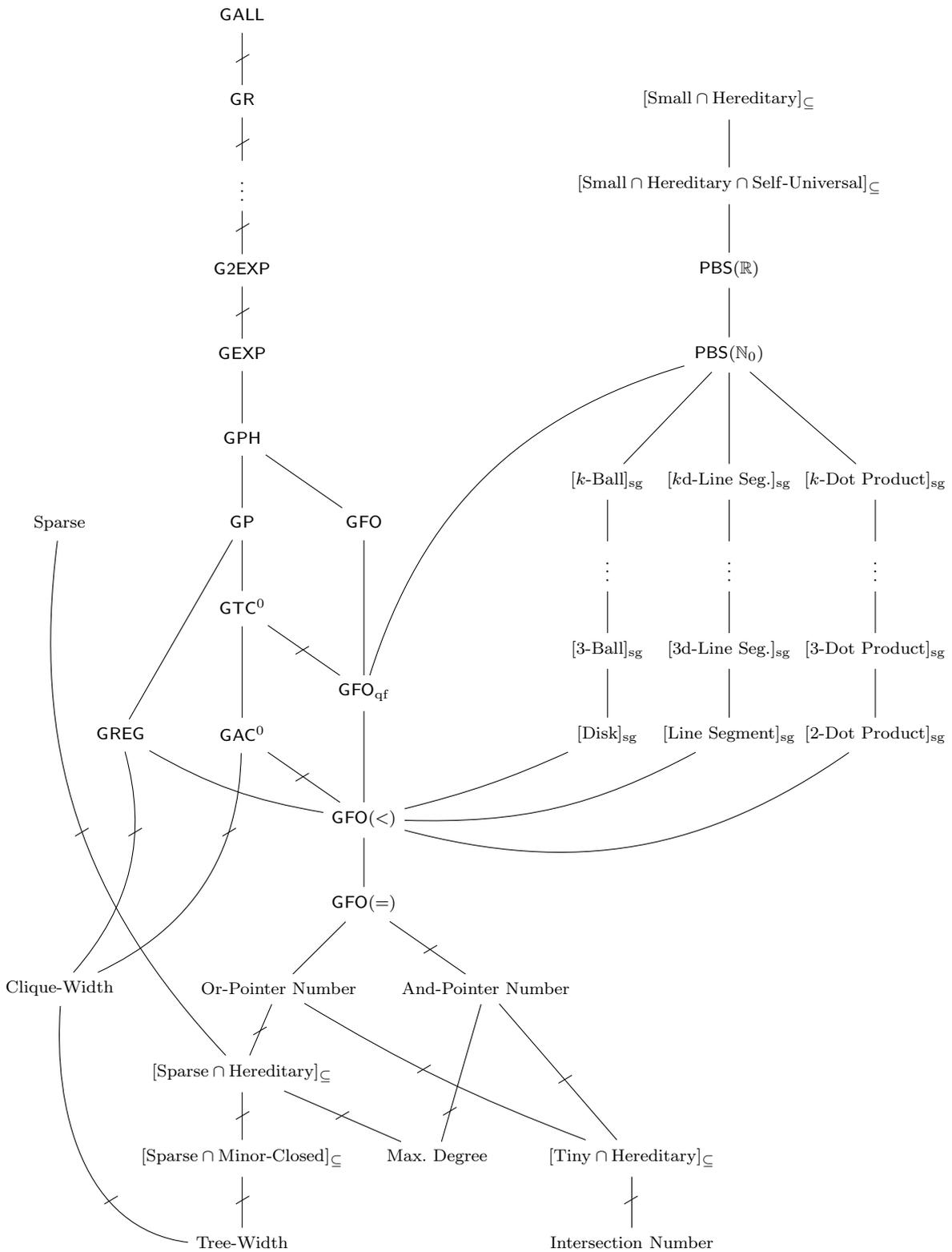
\begin{figure}
    \centering
    \begin{tikzpicture}[shorten >=1pt,auto,node distance=1.2cm,
  main node/.style={circle,draw}]
\usetikzlibrary{shapes.misc}

\newcommand*{\movey}{1.4}%
\newcommand*{\movex}{2}%
\newcommand*{\movexsi}{0.02}%

\node (gsparse) at (-1.5*\movex,-4*\movey) {Sparse};

\node (gcwi) at (-1.5*\movex,-9.5*\movey) {Clique-Width};

\node (gall) at (0,2*\movey) {$\gccall$};

\node (gr) at (0,\movey) {$\gccr$};

\node (grdots) at (0,0) {$\vdots$};

\node (gtwoexp) at (0,-\movey) {$\gcctwoexp$};

\node (gexp) at (0,-2*\movey) {$\gccexp$};

\node (gph) at (0,-3*\movey) {$\gccph$};

\node (gp) at (0,-4*\movey) {$\gccp$};

\node (gtc) at (0,-5*\movey) {$\gcctc$};

\node (gac) at (0,-6.5*\movey) {$\gccac$};

\node (greg) at (-\movex,-6.5*\movey) {$\gccreg$};

\node (gfoq) at (\movex,-4*\movey) {$\gccfoq$};

\node (gram) at (\movex,-6*\movey) {$\gccfo$};

\node (gfolt) at (\movex,-7.5*\movey) {$\gccfolt$};
\node (gfoeq) at (\movex,-8.5*\movey) {$\gccfoeq$};


\node (gdisk) at (3*\movex,-6.5*\movey) {$[\text{Disk}]_\mathrm{sg}$};
\node (gthreeball) at (3*\movex,-5.5*\movey) {$[\text{3-Ball}]_\mathrm{sg}$};
\node (gdotsball) at (3*\movex,-4.5*\movey) {$\vdots$};
\node (gkball) at (3*\movex,-3.5*\movey) {$[\text{$k$-Ball}]_\mathrm{sg}$};

\node (gline) at (4*\movex,-6.5*\movey) {$[\text{Line Segment}]_\mathrm{sg}$};
\node (gthreeline) at (4*\movex,-5.5*\movey) {$[\text{3d-Line Seg.}]_\mathrm{sg}$};
\node (gdotsline) at (4*\movex,-4.5*\movey) {$\vdots$};
\node (gkline) at (4*\movex,-3.5*\movey) {$[\text{$k$d-Line Seg.}]_\mathrm{sg}$};

\node (gtwodp) at (5.2*\movex,-6.5*\movey) {$[\text{2-Dot Product}]_\mathrm{sg}$};
\node (gthreedp) at (5.2*\movex,-5.5*\movey) {$[\text{3-Dot Product}]_\mathrm{sg}$};
\node (gdotsdp) at (5.2*\movex,-4.5*\movey) {$\vdots$};
\node (gkdp) at (5.2*\movex,-3.5*\movey) {$[\text{$k$-Dot Product}]_\mathrm{sg}$};

\node (gopn) at (\movex-\movex+0.6,-9.5*\movey) {Or-Pointer Number};
\node (gapn) at (\movex+\movex,-9.5*\movey) {And-Pointer Number};

\node (gsh) at (\movex-\movex,-10.5*\movey) {$\cogcsh$};
\node (gsmc) at (\movex-\movex,-11.5*\movey) {$\cogcpmc$};
\node (gtw) at (\movex-\movex,-12.5*\movey) {Tree-Width};

\node (gdeg) at (1.6*\movex,-11.5*\movey) {Max.~Degree};

\node (gth) at (3.7*\movex-0.5*\movex,-11.5*\movey) {$\cogcth$};
\node (gecc) at (3.7*\movex-0.5*\movex,-12.5*\movey) {Intersection Number};

\node (gsmallh) at (5*\movex-\movex,-2*\movey) {$\gccpbs(\N_0)$};
\node (gsmallhtwo) at (5*\movex-\movex,-1*\movey) {$\gccpbs(\mathbb{R})$};
\node (gsmallhthree) at (5*\movex-\movex,0*\movey) {$\cogcsmallhsu$};
\node (gsmallhfour) at (5*\movex-\movex,1*\movey) {$\cogcsmallh$};

\path[-]
(gr) edge (gall)
(grdots) edge (gr)
(gtwoexp) edge (grdots)
(gexp) edge (gtwoexp)
(gph) edge (gexp)
(gp) edge (gph)
(gtc) edge (gp)
(gac) edge (gtc)

(greg) edge (gp)
(gfolt) edge[bend angle=10,bend left] (greg)

(gfoq) edge (gph)
(gram) edge (gtc)
(gfolt) edge (gac)

(gfolt) edge[bend angle=5,bend right] (gdisk)
(gfolt) edge[bend angle=15,bend right] (gline)
(gfolt) edge[bend angle=25,bend right] (gtwodp)


(gtwodp) edge (gthreedp)
(gthreedp) edge (gdotsdp)
(gdotsdp) edge (gkdp)

(gram) edge (gfoq)
(gfolt) edge (gram)
(gfoeq) edge (gfolt)

(gopn) edge (gfoeq)
(gapn) edge (gfoeq)

(gtw) edge (gsmc)
(gsmc) edge (gsh)
(gsh) edge (gopn)
(gdeg) edge (gsh)
(gdeg) edge (gapn)
(gecc) edge (gth)
(gth) edge (gapn)

(gth) edge[bend angle=5,bend left] (gopn)

(gsh) edge[bend angle=25,bend left] (gsparse)

(gdisk) edge (gthreeball)
(gthreeball) edge (gdotsball)
(gdotsball) edge (gkball)
(gkball) edge (gsmallh)

(gline) edge (gthreeline)
(gthreeline) edge (gdotsline)
(gdotsline) edge (gkline)
(gkline) edge (gsmallh)

(gkdp) edge (gsmallh)

(gsmallh) edge (gsmallhtwo)
(gsmallhtwo) edge (gsmallhthree)
(gsmallhthree) edge (gsmallhfour)

(gram) edge[bend angle=28,bend left] (gsmallh)

(gcwi) edge[bend angle=28,bend right] (greg)
(gcwi) edge[bend angle=32,bend right] (gac)
(gtw) edge[bend angle=32,bend left,out=55]  (gcwi)
;


\node (si) at (\movexsi,1.5*\movey) {\strictinclusion};
\node (si) at (\movexsi,0.5*\movey) {\strictinclusion};
\node (si) at (\movexsi,-0.5*\movey) {\strictinclusion};
\node (si) at (\movexsi,-1.5*\movey) {\strictinclusion};

\node (sitw) at (\movex-\movex+\movexsi,-12.5*\movey+0.5*\movey) {\strictinclusion};
\node (sismc) at (\movex-\movex+\movexsi,-11.5*\movey+0.5*\movey) {\strictinclusion};

\node (sishapnfoeq) at (2.4*\movex-\movex+\movexsi+0.285,-9.8*\movey+0.5*\movey+0.37) {\strictinclusion};

\node (sishopn) at (\movex-\movex+\movexsi+0.285,-10.5*\movey+0.5*\movey) {\strictinclusion};

\node (sishs) at (-1.32*\movex+\movexsi,-7.65*\movey) {\strictinclusion};

\node (sicwac) at (-0.11*\movex+\movexsi,-7.65*\movey) {\strictinclusion};
\node (sicwreg) at (-0.882*\movex+\movexsi,-7.65*\movey) {\strictinclusion};
\node (sitwcw) at (-1.08*\movex+\movexsi,-12*\movey) {\strictinclusion};

\node (siecc) at (3.7*\movex-0.5*\movex,-12*\movey) {\strictinclusion};

\node (sidegsh) at (2.4*\movex-\movex-0.57*\movex,-12.6*\movey+1.6*\movey) {\strictinclusion};
\node (sidegsapn) at (2.252*\movex-\movex+0.5*\movex-0.08,-12.8*\movey+1.85*\movey) {\strictinclusion};

\node (sithopn) at (2*\movex-\movex+0.5*\movex,-12.3*\movey+1.85*\movey) {\strictinclusion};
\node (sithapn) at (2.6*\movex-\movex+0.36*\movex+0.7*\movex,-12.6*\movey+1.8*\movey+0.32) {\strictinclusion};

\node (sifolt) at (0.5*\movex,-7*\movey) {\strictinclusion};
\node (siram) at (0.5*\movex,-5.5*\movey) {\strictinclusion};

\end{tikzpicture}
    \caption{Landscape of small graph classes}
    \label{fig:bigpicture}                
\end{figure}  

The initial question which started our investigations was whether the computational aspect of label decoders matters at all with respect to the set of graph classes that can be represented. In Section~\ref{sec:hierarchy} we applied a diagonalization argument to affirmatively answer this question. Due to the brute force aspect of this argument it fails to separate classes below $\gcctwoexp$. Nonetheless, it shows that the definition of labeling schemes given by Muller \cite{muller} differs from the one given by Kannan, Naor and Rudich \cite{kannan}, i.e.~$\gccp \neq \gccr$. 
The graph classes which exhibit these separations are far removed from any natural graph class. The next and much more difficult question is whether a hereditary graph class can be shown to not have a labeling scheme under certain computational constraints. The first step is to find adequate computational constraints for this purpose. 
The first obvious candidate for this is $\gccac$ since $\AC^0$ is usually considered to be the smallest meaningful complexity class. To our surprise, every hereditary graph class that we found to have a labeling scheme can also be easily seen to be in $\gccac$. Moreover, we interpret the fact that $\gccac \not\subseteq \cogcsmallh$ (Theorem~$\ref{thm:acsmallh}$) as evidence that proving lower bounds against $\gccac$ in general (not only w.r.t.~hereditary graph classes) is very difficult due to the non-uniform constructions that can be realized by labeling schemes in $\gccac$. This prompted us to look for alternative models of computation that might be more suitable for proving lower bounds and led us to logical and regular labeling schemes.

Logical labeling schemes generalize the naive labeling schemes for many geometrical intersection graph classes such as interval graphs or circle graphs by interpreting the vertex labels as a constant number of polynomially bounded numbers. The quantifier-free fragment $\gccfo$ is well-behaved in the sense that it is a subset of $\cogcsmallh$. Furthermore, it is quite robust and enjoys a characterization in terms of constant-time RAM machines: $\gccfo$ is the set of graph classes which have a labeling scheme with a constant-time decidable label decoder on a RAM machine without division (\cite[Corol.~3.84]{chat}). When the label length constraint in quantifier-free logical labeling schemes is dropped (and thus we do not talk about labeling schemes anymore) one obtains what we named polynomial-boolean systems. They are interesting because they contain many of the candidates for the implicit graph conjecture such as line segment graphs or $k$-dot product graphs. Going below $\gccfo$ we find $\gccfolt$ and $\gccfoeq$ which have various complete graph classes under both types of reduction and contain a wealth of graph classes that have been intensely studied from a graph-theoretical and an algorithmic perspective. It is notable that $\gccfoeq$ and $\gccfolt$ are closed under hereditary closure, i.e.~if they contain a graph class $\mathcal{C}$ then they also contain its hereditary closure $[\mathcal{C}]_{\subseteq}$. This is not the case for $\gccfo$ unless $\gccfo = \gccpbs(\N_0)$ (Theorem~\ref{thm:gccfohcpbs}).
The pointer numbers are two very simple families of labeling schemes from $\gccfoeq$ that already capture uniformly sparse and bounded degree graph classes. They can be used as a starting point to examine the expressiveness and algorithmic properties of labeling schemes.

An integral part of complexity theory are reductions. Initially, it was not clear for us what an adequate reduction notion in the context of labeling schemes should look like. The first important insight was that it should be a relation on graph classes as opposed to labeling schemes. Furthermore, the smallest class of interest $\gccfoeq$ should be closed under such a reduction notion. This naturally led us to consider ways to define such a relation in terms of boolean functions. Algebraic and subgraph reductions satisfy these requirements.
Intuitively, they allow us to compare the complexity of the adjacency structure of graph classes in a combinatorial way. Roughly speaking, algebraic reductions are a special case of subgraph reductions (Lemma~\ref{lem:bftosg}). In the context of labeling schemes they allow us to compare graph classes for which no labeling schemes are known. In particular, it would be interesting to see whether candidates for the implicit graph conjecture can be reduced to each other. For example, can every disk in the plane be assigned to a constant number of line segments in the plane such that two disks intersect iff a boolean combination of their corresponding line segments intersect? This would imply a subgraph reduction from disk graphs to line segment graphs and it would mean that constructing a labeling scheme for disk graphs is not harder than for line segment graphs. 

Another aspect that we considered are algorithmic properties of graph classes with implicit representations.
Can certain algorithmic problems be solved efficiently on graph classes with a labeling scheme of very low complexity? A positive result seems plausible since one would expect such graph classes to have a rather simple adjacency structure. In Section~\ref{sec:algo} we observed that this kind of question naturally fits into the framework of parameterized complexity. More specifically, classes such as $\gccfoeq$ or $\gccfolt$ can be understood as parameterizations (see Theorem~\ref{thm:param} and the paragraph after Corollary~\ref{corol:param}). 
For example, is the dominating set problem fixed-parameter tractable when parameterized by the or-pointer number and the solution size? This would generalize the result for uniformly sparse graphs classes. See Figure~\ref{fig:alg_side} and the subsequent paragraphs for more algorithmic questions in that direction.

Let us turn back to the leading question of finding a small and hereditary graph class that does not have a labeling scheme under certain complexity constraints. The most expressive classes of labeling schemes against which proving lower bounds does not seem inconceivable (given the current state of knowledge) are $\gccac$, $\gccfo$ and $\gccreg$. While we believe that there exists a small and hereditary graph class that does not reside in $\gccac$, we also presume that proving the existence of such a graph class is out of reach. In the case of $\gccreg$ we cannot faithfully claim the situation to be different due to a lack of intuition. 
To better understand the expressiveness of regular labeling schemes we suggest to determine whether $\gccreg$ is a subset of $\cogcsmallh$ and whether it is closed under hereditary closure (the latter implies the former).
 This leaves us with $\gccfo$. Unlike in the case of the other two classes the vertex labels cannot be accessed bitwise by labeling schemes in $\gccfo$ but have to be interpreted as numbers. A concrete example of this limitation is exhibited by the labeling scheme for graph classes with bounded clique-width (Fact \ref{fact:cwgac} and \ref{fact:cwreg}).  While all such graph classes are in $\gccac$ and $\gccreg$ it is not clear whether they are also in $\gccfo$. 
Moreover, the class $\gccfo$ also has practical merit due to its RAM characterization. This motivates us to state the following variant of the implicit graph conjecture (IGC):   

\begin{conjecture}[Weak Implicit Graph Conjecture]
    Every small and hereditary graph class is in $\gccfo$. Stated differently, $\gccfo = \cogcsmallh$.
\end{conjecture}

Every candidate for the IGC (small, hereditary graph classes not known to be in $\gccp$) is also a candidate for the weak IGC. Additionally, graph classes with bounded clique-width are candidates for the weak IGC but not for the IGC because they are in $\gccac$ and thus in $\gccp$. 

A related task is to show that none of the six classes which are adjacent to $\gccfolt$ in Figure~\ref{fig:bigpicture} coincide with it. For the seventh neighbor $\gccac$ this separation follows from the fact that it is not a subset of $\cogcsmallh$. In order to separate $\gccfolt$ from $\gccfo$ it already suffices to separate $\gccfolt$ from $\gccpbs(\N_0)$ (Corollary~\ref{corol:pbscol}). Moreover, it suffices to only consider hereditary graph classes when comparing $\gccfolt$ and $\gccpbs(\N_0)$ because both are closed under hereditary closure.  
The class $\gccfolt$ is particularly interesting because it contains many well-studied graph classes and seems to be unsophisticated enough in order to establish a thorough understanding. The following two questions are aimed at developing such an understanding.  
Let us say a graph class is f-hereditary if it is characterized by a finite set of forbidden induced subgraphs. Can the set of (undirected) f-hereditary graph classes in $\gccfolt$ be characterized in terms of their forbidden induced subgraphs? Stated differently, given a finite set of forbidden induced subgraphs decide whether the graph class induced by this set is in $\gccfolt$. A simple, first step is to identify f-hereditary graph classes in- and outside of $\gccfolt$. 
The second question is: does $\gccfolt$ have a $\bfreduction$-complete graph class when restricted to undirected graph classes? We have shown that interval graphs are $\sgreduction$-complete for $\gccfolt$ but it is not clear whether they are also $\bfreduction$-complete for the set of undirected graph classes in $\gccfolt$. For instance, it is not clear whether $k$-interval graphs are $\bfreduction$-reducible to interval graphs for $k \geq 2$. The same two questions can be asked about $\gccfoeq$.

The following paragraphs raise questions that can be treated independently of the concept of labeling schemes. More specifically, they deal with questions regarding the two reduction notions that we introduced.

It is trivially true that $\bfreduction$-reductions are weaker than $\sgreduction$-reductions because a directed graph class cannot be $\bfreduction$-reducible to an undirected one whereas in the case of $\sgreduction$-reductions this is possible. Does this also hold for non-trivial reasons? Let us say $\mathcal{C} \bfreduction^* \mathcal{D}$ if there exists an $n_0 \in \N$ such that $ \mathcal{C}_{\geq n_0} \bfreduction \mathcal{D}$ ($\mathcal{C}_{\geq n_0}$ means the set of graphs with at least $n_0$ vertices in $\mathcal{C}$). Are there undirected, hereditary graph classes $\mathcal{C}$ and $\mathcal{D}$ such that $\mathcal{C} \sgreduction \mathcal{D}$ holds but $\mathcal{C} \bfreduction^* \mathcal{D}$ does not hold? Candidates for $\mathcal{C}$ and $\mathcal{D}$ are $k$-interval graphs and interval graphs, respectively. The relation $\bfreduction^*$ is used to exclude finite anomalies.  

The set of small and hereditary graph classes is a rich class. Finding a characterization of this class is a very ambitious task because the complexity of such a characterization can be expected to reflect the richness of this class. The following question is an attempt at guessing what such a characterization could look like.
Is there a small and hereditary graph class such that every other small and hereditary graph class is $\bfreduction$-reducible to it? This is equivalent to asking whether $\cogcsmallh$ has a $\bfreduction$-complete graph class. If this were to be true then the adjacency structure of every small and hereditary graph class would just be a boolean combination of one such particular graph class. 
We do not believe that such a complete graph class exists but we think that a refutation of this would be insightful by itself.

Can any inclusion or collapse be shown for the classes between $\gccfolt$ and $\gccpbs(\N_0)$? For instance, are disk graphs $\sgreduction$-reducible to line segment graphs? A collapse would mean that the collapsed classes can be effectively treated as a single candidate for the IGC.

\printbibliography[heading=bibintoc]

\end{document}